\documentclass{lmcs}
\pdfoutput=1

\usepackage{lastpage}
\lmcsdoi{18}{1}{32}
\lmcsheading{}{\pageref{LastPage}}{}{}%
{Apr.~24,~2020}{Feb.~17,~2022}{}

\usepackage[shortlabels]{enumitem}

\usepackage{mathdots}
\usepackage{proof}
\usepackage{multicol}
\usepackage{amssymb}
\usepackage{stmaryrd}
\usepackage{float}
\usepackage[T1]{fontenc}
\usepackage[utf8x]{inputenc}

\usepackage[all]{xy}
\usepackage{tikz-cd}

\usepackage{bookmark}

\newcommand*{\covar}[1]{\overline{#1}}
\newcommand*{\sep}{\mid\ }
\newcommand*{\negF}[1]{\overline{#1}}

\newcommand*{\prove}{\triangleright_{\hr}}
\newcommand*{\proveNC}{\triangleright_{\hr\backslash\{\text{CAN}\}}}
\newcommand*{\proveM}{\triangleright_{\hmr}}
\newcommand*{\proveMNC}{\triangleright_{\hmr\backslash\{\text{CAN}\}}}
\newcommand*{\semProve}{\mathcal{A}_{\textnormal{Riesz}}\vdash}

\newcommand*{\sem}[1]{\llparenthesis #1 \rrparenthesis}
\newcommand*{\app}[2]{#1;#2}

\newcommand\nodeC[1]{*+[o][F]{#1}}

\newcommand{\hr}{\textbf{HR}}
\newcommand{\hmr}{\textbf{HMR}}
\newcommand{\gasystem}{\textbf{GA}}

\begin{document}

\title[Proof Theory of Riesz Spaces and Modal Riesz Spaces]{Proof Theory of Riesz Spaces\\ and Modal Riesz Spaces}

\author[C.~Lucas]{Christophe Lucas}	
\address{CNRS \& ENS Lyon}	
\email{\{\texttt{christophe.lucas},\texttt{matteo.mio}\}\texttt{@ens-lyon.fr}}  
\thanks{This work has been supported by the European Research Council (ERC) under the European Union’s Horizon 2020 programme (CoVeCe, grant agreement No 678157),  by the LABEX MILYON (ANR-10-LABX-0070) of Université de Lyon, within the program ``Investissements d'Avenir'' (ANR-11-IDEX- 0007) and by the French project ANR-20-CE48-0005 QuaReMe.}

\author[M.~Mio]{Matteo Mio}	

\begin{abstract}
  \noindent We design hypersequent calculus proof systems for the theories of Riesz spaces and modal Riesz spaces and prove the key theorems: soundness, completeness and cut-elimination. These are then used to obtain completely syntactic proofs of some interesting results concerning the two theories. Most notably, we prove a novel result: the theory of modal Riesz spaces is decidable. This work has applications in the field of logics of probabilistic programs since modal Riesz spaces provide the algebraic semantics of the Riesz modal logic underlying the probabilistic $\mu$-calculus.
\end{abstract}

\maketitle
\tableofcontents

\section{Introduction}

Riesz spaces, also known as vector lattices, are real vector spaces equipped with a lattice order ($\leq$) such that the vector space operations of addition and scalar multiplication are compatible with the order in the following sense:
\begin{align*}
&x\leq y \Longrightarrow x + z \leq y + z & \\ &x\leq y \Longrightarrow  rx \leq ry &\textnormal{ for every positive scalar $r\in\mathbb{R}_{\geq 0}$}.
\end{align*}

The simplest example of Riesz space is the linearly ordered set of real numbers $(\mathbb{R},\leq)$ itself. More generally, for a given set $X$, the space of all functions $\mathbb{R}^X$ with operations and order defined pointwise is a Riesz space. If $X$ carries some additional structure, such as  a topology or a $\sigma$-algebra,  then the spaces of continuous and measurable functions both constitute Riesz subspaces of $\mathbb{R}^X$. For this reason, the study of Riesz spaces originated at the intersection of functional analysis, algebra and measure theory and was  pioneered in the 1930's by F.~Riesz, G.~Birkhoff, L.~Kantorovich and H.~Freudenthal among others. Today, the study of Riesz spaces constitutes a well-established field of research. We refer to~\cite{Luxemburg,JVR1977} as standard references.

The definition of Riesz spaces merges the notions of lattice order and that of real vector spaces. The former is pervasive in logic and the latter is at the heart of probability theory (e.g., convex combinations, linearity of the expected value operator, \emph{etc.}) Dexter~Kozen  was the first to observe in a series of seminal works (see, e.g.,~\cite{Kozen1981,Kozen1983}) that, for the above  reasons, the theory of Riesz spaces provides a convenient mathematical setting for the study and design of \emph{probabilistic logics}. Probabilistic logics are formal languages conceived to express correctness properties of probabilistic transition systems (e.g., Markov chains, Markov decision processes, \emph{etc}.) representing the formal semantics of computer programs using probabilistic operations such as random bit generation. In a series of recent works~\cite{MIO2012b,MioSimpsonFICS2013,Mio2012c,MIO2012a,MIO11,MioSimpsonFI2017,miolics2017,mio18,MIO2014a,FMM2020}, following Kozen's program, the second author has introduced a simple probabilistic modal logic called \emph{Riesz Modal Logic}. Importantly, once extended with fixed-point operators in the style of the modal $\mu$-calculus~\cite{Kozen83}, this logic is sufficiently expressive to interpret other popular probabilistic logics for verification such as \emph{probabilistic CTL} (see, e.g., chapter 8 in~\cite{BaierKatoenBook} for an introduction to this logic). One key contribution from~\cite{miolics2017,FMM2020} is a duality theory which provides a bridge between the probabilisitic transition system semantics of the Riesz modal logic and its algebraic semantics given in terms of so-called \emph{modal Riesz spaces}.

A \emph{modal Riesz space} is a structure $(V, \leq,\Diamond)$ such that $(V,\leq)$ is a Riesz space and $\Diamond$ is a unary operation $\Diamond : V\rightarrow V$ satisfying certain axioms (see Definition~\ref{defi:modal_riesz_space} for details). Terms without variables in the signature of modal Riesz spaces are exactly terms of the Riesz modal logic of~\cite{miolics2017,FMM2020}. As a consequence of the duality theory, two terms are equivalent in the transition semantics if and only if they are provably equal in the equational theory of modal Riesz spaces. This is a complete axiomatisation result (see~\cite{miolics2017,FMM2020} for details). 

One drawback of equational axiomatisations, such as that of~\cite{miolics2017,FMM2020}, is that the underlying proof system of \emph{equational logic} is not well-suited for proof search. It is indeed often difficult to find proofs for even simple equalities. The source of this difficulty lies in the transitivity rules of equational logic:

$$
    \infer[\text{Trans}]{A = C}{A = B \ \ \ \ & \ \ \ \ B = C} 
$$
For proving the equality $A=C$ it is sometimes necessary to come up with an additional term $B$ and prove the two equalities $A=B$ and $B=C$. Since $B$ ranges over all possible terms, the proof search endeavour faces an infinite branching in possibilities. It is therefore desirable to design alternative proof systems that are better behaved from the point of view of proof search, in the sense that the choices available during the proof construction process are reduced to the bare minimum. 

The mathematical field of \emph{structural proof theory} (see~\cite{Buss98} for an overview), originated with the seminal work of Gentzen on his \emph{sequent calculus} proof system LK for classical propositional logic~\cite{Gentzen1934}, investigates such proof systems. The key technical result regarding the sequent calculus, called the \emph{cut-elimination theorem}, implies that when searching for a proof of a statement, only certain terms need to be considered: the so-called \emph{sub-formula property}. This simplifies significantly, in practice, the \emph{proof search} endeavour. 

The original system LK of Gentzen has been extensively investigated and generalised. For example, sequent calculi for several substructural logics, linear logic, many modal logics and fixed-point temporal logics have been designed. One variant of sequent calculus, called \emph{hyper-sequent calculus}, originally introduced by Avron in~\cite{avron1987} and independently by Pottinger in~\cite{Pottinger1983}, allows for the manipulation of non-empty lists of sequents (hence the \emph{hyper} adjective) rather than just sequents. The machinery of hyper-sequent calculi has been then developed in several works (see, e.g.,~\cite{DBLP:conf/tableaux/CiabattoniLR19, DBLP:journals/logcom/BaazCF03, DBLP:conf/lics/CiabattoniGT08, DBLP:conf/aiml/Ciabattoni18, MOGbook, DBLP:journals/logcom/Ciabattoni01, DBLP:conf/lics/Lahav13,DBLP:journals/logcom/BaazCF03}).

\paragraph{\textbf{First Contribution: Proof Theory of Riesz Spaces}}

The first contribution of this work is the design of a \emph{hypersequent calculus} proof system \hr\ for the theory of Riesz spaces, together with the proof of the cut-elimination theorem. From this we obtain new proofs, based on purely syntactic methods, of well-known results  such as the fact that the equational theory of Riesz spaces is decidable and that  the equational theory of Riesz spaces with real ($\mathbb{R}$) scalars is a conservative extension of the theory of Riesz spaces with rational ($\mathbb{Q}$) scalars. These results are presented in Section~\ref{modal_free_section}.

Our hypersequent calculus \hr\ is based on, and extends, the hypersequent calculus $\textbf{GA}$ for the theory of \emph{lattice ordered Abelian groups} of~\cite{MOG2005,MOGbook}. From a technical point of view, the difficulty in extending their work to Riesz spaces lies mostly in the design of appropriate derivation rules for dealing with real $(\mathbb{R})$ scalars. Our design choices have been driven by the two main goals: prove the cut-elimination theorem and preserve as much as possible the sub-formula property of the system. It is our belief that this type of rules might be of general interest in the field of proof theory (see, e.g.,~\cite{DBLP:conf/ijcai/KulackaPS13}) and can be potentially re-used for designing proof systems for other quantitative logics.

Beside being the first structural proof system for Riesz spaces, a mathematically natural and well studied class of structures, the hypersequent calculus proof system \hr\  is the basis for our second and most technically challenging contribution.

\paragraph{\textbf{Second Contribution: Proof Theory of Modal Riesz Spaces}}

The second and more technically challenging contribution of this work is the design of the \emph{hypersequent calculus} proof system \hmr, together with the proof of the cut-elimination theorem, for the theory of modal Riesz spaces or, equivalently (via the duality theory of~\cite{miolics2017,FMM2020}) for the Riesz modal logic. To the best of our knowledge, this is the first sound, complete and structural proof system for a probabilistic logic designed to express properties of probabilistic transition systems. From the cut-elimination theorem for \hmr\ we derive a new result: the equational theory of modal Riesz spaces is decidable. Our proof is based on purely syntactical methods and does not rely, as it is often the case in decidability results of modal logics, on model theoretic properties (e.g., the finite model property) or on techniques from automata theory. These results are presented in Section~\ref{modal_section}.

The hypersequent calculus \hmr\ is based on and extends the hypersequent calculus \hr\ with new rules dealing with the additional connectives ($\Diamond$ and 1) available in the signature of modal Riesz spaces. While this extension might superficially seem simple, it introduces significant complications in the proof of the cut-elimination theorem. The proof technique adopted in~\cite{MOG2005,MOGbook} for proving the cut-elimination theorem of $\textbf{GA}$ does not seem to be applicable (as discussed in Section~\ref{subsec:main_results}). We therefore follow a different approach based on a technical result (the \emph{M-elimination} theorem) which states that one of the rules (M) of the system  \hmr\  can be safely removed from the system without affecting completeness. In order to simplify as much as possible the exposition of our cut-elimination proof for \hmr, we prove the cut-elimination theorem for the system $\hr$ also using the technique based on the M-elimination theorem even though the cut-elimination theorem for \hr\ could also be obtained following the approach of~\cite{MOG2005,MOGbook}. This will serve as a preparatory work for the more involved proof of cut-elimination for \hmr.

\emph{Earlier Work.} The content of this section is based on an earlier conference paper of the two authors~\cite{DBLP:conf/fossacs/LucasM19} where a restricted version of the present system \hmr, without rules dealing with real ($\mathbb{R}$) scalars, has been introduced.

\paragraph{\textbf{Third Contribution: Formalisation with the {Coq} Proof Assistant}} All the definitions, statements and proofs presented in this work have been formalised by the first author using the \href{https://coq.inria.fr/}{{Coq} proof assistant} and are freely available \href{https://github.com/clucas26e4/riesz-logic}{online}~\cite{coq-formalisation}. Since all proofs are constructive, the formalisation provides algorithms for, e.g., performing the cut-elimination procedure or for deciding the equational theory of modal Riesz spaces\footnote{The output of our decision procedure is a formula, in the first order theory of the real-closed field $(\mathbb{R}, +, \times)$. This theory is decidable~\cite{tarski1951} and can be verified for satisfiability using external tools.}. The  {Coq} code has been implemented during the process of peer-review of this article and, as such, it has not been scrutinised by the anonymous reviewers. It should be therefore taken as an external addendum of the present article.

\paragraph{\textbf{Organisation of this work}}
This paper is structured in the following main sections:
\begin{description}[font=\normalfont\itshape]
\item[Section~\ref{sec:background} - Technical Background] in this section we give the basic definitions and results regarding Riesz spaces and modal Riesz spaces and fix some notational conventions.
\item[Section~\ref{modal_free_section} - Hypersequent Calculus for Riesz Spaces] this section is devoted to our hypersequent calculus \hr\ proof system for the theory of Riesz spaces. This section is structured in several subsections, each presenting in details a result regarding \hr.
\item[Section~\ref{modal_section} - Hypersequent Calculus for Modal Riesz Spaces] this section is devoted to our hypersequent calculus \hmr\ proof system for the theory of modal Riesz spaces. The structure of this section matches exactly that of Section~\ref{modal_free_section}. This should allow for an easier comparison of the two systems and their technical differences.
\item[Section~\ref{conclusion:sec} - Conclusions] some final remarks and directions for future work.

Finally, for convenience, we have included the rules of the proof system \gasystem\ from~\cite{MOG2005,MOGbook} in Appendix~\ref{ga:section}.
\end{description}


\section{Technical Background}
\label{sec:background}

This section provides the necessary definitions and basic results regarding Riesz spaces (the books~\cite{Luxemburg,JVR1977} are standard references) and modal Riesz spaces from~\cite{miolics2017,FMM2020}, which play a key role in the the duality theory of the Riesz modal logic.

\subsection{Riesz Spaces}\label{sec:background:riesz}
This section contains the basic definitions and results related to Riesz spaces. We refer to~\cite{Luxemburg,JVR1977} for a comprehensive reference to the subject.

A Riesz space is an algebraic structure $(R,0,+,(r)_{r\in\mathbb{R}},\sqcup,\sqcap)$ such that 
$(R,0,+,(r)_{r\in\mathbb{R}})$ is a vector space over the reals, $(R,\sqcup,\sqcap)$ is a lattice and the induced order $(a\leq b \Leftrightarrow a\sqcap b = a)$ is compatible with addition and with the scalar multiplication, in the sense that: (i)  for all $a,b,c\in R$, if $a\leq b$ then $a+c\leq b+c$, and (ii) if $a\geq b$ and $r\in \mathbb{R}_{\geq 0}$ is a non-negative real, then $r a\geq rb$. Formally we have:

\begin{defi}[Riesz Space]\label{techback:definition:riesz_space}
The \emph{language} $\mathcal{L}_R$ of Riesz spaces is given by the (uncountable) signature $\{ 0,+, (r)_{r\in\mathbb{R}}, \sqcup, \sqcap\}$ where $0$ is a constant, $+$, $\sqcup$ and $\sqcap$ are binary functions and $r$ is a unary function, for all $r\in\mathbb{R}$. A \emph{Riesz space} is a $\mathcal{L}_R$-algebra satisfying the set $\mathcal{A}_{\textnormal{Riesz}}$ of equational axioms of Figure~\ref{axioms:of:riesz:spaces}. We use the standard abbreviations of $-x$ for $(-1)x$ and $x\leq y$ for $x\sqcap y = x$.
\end{defi}

\begin{figure}[ht]
\begin{center}
 \begin{enumerate}
\item Axioms of real vector spaces:
\begin{itemize}
\item Additive group: $x + (y + z) = (x + y) + z$, $x + y = y + x$, $x + 0 = x$, $x - x= 0$,
\item Axioms of scalar multiplication: $r_1(r_2 x) = (r_1\cdot r_2) x$, $1x = x$, $r(x+y) = (rx) + (ry)$, $(r_1 + r_2)x = (r_1 x) + (r_2 x)$,
\end{itemize}
\item Lattice axioms:    (associativity) $x \sqcup (y \sqcup z) = (x \sqcup y) \sqcup z$,  $x \sqcap (y \sqcap z) = (x \sqcap y) \sqcap z$, (commutativity) $z \sqcup y = y \sqcup z$, $z \sqcap y = y \sqcap z$,
(absorption) $z \sqcup (z \sqcap y) = z$, $z \sqcap (z \sqcup y) = z$. 
\item Compatibility axioms:  
\begin{itemize}
\item $(x \sqcap y) + z \leq  y + z $,
\item $r (x \sqcap y) \leq  ry$, for all scalars $r\geq 0$. 
\end{itemize}
\end{enumerate}
\end{center}
\caption{Set $\mathcal{A}_{\textnormal{Riesz}}$ of equational axioms of Riesz spaces.}
\label{axioms:of:riesz:spaces}
\end{figure}

\begin{rem}\label{remark_variety}
  Note how the compatibility axioms have been equivalently formalised in Figure~\ref{axioms:of:riesz:spaces} as inequalities and not as implications by using $(x\sqcap y)$ and $y$ as two general terms automatically satisfying the hypothesis $(x\sqcap y)\leq y$. Moreover the inequalities can be rewritten as equations using the lattice operations ($x\leq y \Leftrightarrow x\sqcap y = x$) as follows:
  \begin{itemize}
  \item $(x \sqcap y) + z \leq  y + z $ can be rewritten as $((x \sqcap y) + z) \sqcap (y + z) = (x \sqcap y) + z$ and
  \item $r (x \sqcap y) \leq  ry$ can be rewritten as $r(x\sqcap y) \sqcap ry = r(x\sqcap y)$.
  \end{itemize}
  Hence, since Riesz space are axiomatised by a set of equations, the family of Riesz spaces is a variety in the sense of universal algebra.
\end{rem}

\begin{exa}\label{example:techback1}
  The real numbers $\mathbb{R}$ together with their standard linear order $(\leq)$, expressed by taking $r_1 \sqcap r_2 = \min(r_1, r_2)$ and  $r_1 \sqcup r_2 = \max(r_1, r_2)$, is a Riesz space. This is a fundamental example also due to the following fact (see, e.g.,~\cite{LvA2007} for a proof): the real numbers $\mathbb{R}$ is complete for the quasiequational theory of Riesz spaces. In particuliar, this means that for any two terms $A,B$, we have that the equality $A=B$ holds in all Riesz spaces if and only if $A=B$ holds in the Riesz space $(\mathbb{R},\leq)$. This provides a practical method for establishing if an equality is derivable from the axioms of Riesz spaces and since the first order theory of the real numbers is decidable~\cite{tarski1951}, so is the equational theory of Riesz spaces. For example, $-( \max(r_1,r_2) ) = \min( -r_1, -r_2)$ holds universally in   $\mathbb{R}$ and therefore $ -(x \sqcup y) = (-x) \sqcap (-y)$ holds in all Riesz spaces.
\end{exa}

\begin{exa}\label{example:techback2}
For a given set $X$, the set $\mathbb{R}^X$ of functions $f: X\rightarrow \mathbb{R}$ is a Riesz space when all operations are defined pointwise: $(rf)(x) = r( f(x)) $, $ (f+g)(x) = f(x) + g(x)$, $ (f\sqcup g)(x) = f(x) \sqcup g(x)$, $ (f\sqcap g)(x) = f(x) \sqcap g(x)$. Thus, for instance, the space of $n$-dimensional vectors $\mathbb{R}^n$ is a Riesz space whose lattice order is not linear.
\end{exa}

\begin{conv}
We use the capital letters $A,B, C$ to range over terms built from a set of variables ranged over by $x, y, z$. We write $A[B/x]$ for the term, defined as expected, obtained by substituting all occurrences of the variable $x$ in the term $A$ with the term $B$.
\end{conv}

As observed in Remark~\ref{remark_variety}, the family of Riesz spaces is a variety of algebras. This means, by Birkhoff  completeness theorem,  that two terms $A$ and $B$ are equivalent in all Riesz spaces if and only if the identity $A=B$ can be derived using the familiar deductive rules of equational logic, written as  $\mathcal{A}_{\textnormal{Riesz}}\vdash A=B$.
\begin{defi}[Deductive Rules of Equational Logic]\label{eq_logic_rules}
Rules for deriving identities between terms from a set $\mathcal{A}$ of equational axioms:  

\begin{center}
   \hfill $\infer[\text{Ax}]
        {\mathcal{A}\vdash A = B}
        { (A=B)\in\mathcal{A}}$ \hfill
   $\infer[\text{refl}]
        {\mathcal{A}\vdash A = A}
        {}$ \hfill  
    $\infer[\text{sym}]
        {\mathcal{A}\vdash A = B}
        {\mathcal{A}\vdash B = A}$ \hfill 
         $\infer[\text{ctxt}]
        {\mathcal{A}\vdash C[A] = C[B]}
        {\mathcal{A}\vdash A = B}$ \hfill
\end{center}
\begin{center}
    \hfill $\infer[\text{trans}]
        {\mathcal{A}\vdash A = C}
        {\mathcal{A}\vdash A = B & \mathcal{A} \vdash B = C}$ \hfill 
  $\infer[\text{subst}]
  {\mathcal{A}\vdash A[C/x] = B[C/x]}
  {\mathcal{A}\vdash A  = B}$ \hfill
\end{center}
\noindent
where $A,B,C$ are terms of the algebraic signature under consideration built from a countable collection of variables and $C[\cdot ]$ is a context. \end{defi}

In what follows we denote with $\mathcal{A}_{\textnormal{Riesz}}\vdash A \leq B$ the judgment $\mathcal{A}_{\textnormal{Riesz}}\vdash A= A \sqcap B$.

The following elementary facts (see, e.g.,~\cite[\S 2.12]{Luxemburg} for proofs) imply that, in the theory of Riesz spaces, a proof system for deriving equalities can be equivalently seen as a proof system for deriving equalities with $0$ or inequalities. 
\begin{prop}The following assertions hold:
\begin{itemize}
\item $\mathcal{A}_{\textnormal{Riesz}} \vdash A = B$  $\ \Leftrightarrow\ $ $\mathcal{A}_{\textnormal{Riesz}} \vdash A - B = 0$,
\item $\mathcal{A}_{\textnormal{Riesz}}\vdash A = B$ $\ \Leftrightarrow\ $ $\big(\mathcal{A}_{\textnormal{Riesz}}  \vdash A \leq B \textnormal{ and }  \mathcal{A}_{\textnormal{Riesz}} \vdash  B \leq A\big)$.
\end{itemize}
\end{prop}

\begin{conv}
From now on, in the rest of this paper, it will be convenient to take the derived negation operation $(-A) = (-1)A$ as part of the signature and restrict all scalars $r$ to be strictly positive $(r>0$). The scalar $0\in\mathbb{R}$ can be removed by rewriting $(0)A$ as $0$. 
\end{conv}

\begin{defi}\label{nnf:definition}
A term $A$ is in \emph{negation normal form} (NNF) if the operator $(-)$ is only applied to variables. 
\end{defi}
For example, the term $(-x) \sqcap (-y)$ is in NNF, while the term $-(x\sqcup y)$ is not. 

\begin{lem}
  \label{lem:into_nnf}
Every term $A$ can be rewritten to an equivalent term in NNF.
\end{lem}
\begin{proof}
Negation can be pushed towards the variables by the following rewritings: $-(-A) = A$, $-( r A) = r(-A)$, $-(A+B) = (-A) + (-B)$,  $-(A\sqcup B) = (-A) \sqcap (-B)$ and $-(A\sqcap B) = (-A) \sqcup (-B)$ (see Lemma~\ref{lem:equalities}[7] below).
\end{proof}

Negation can be defined on terms in NNF as follows.

\begin{defi}
Given a term $A$ in NNF, the term $\overline{A}$ is defined as follows: $\overline{x} = -x$, $\overline{-x} = x$, $\overline{rA } = r\overline{A}$, $\overline{A+ B} = \overline {A}+ \overline{B}$, $\overline{A\sqcup B} = \overline {A}\sqcap \overline{B}$, $\overline{A\sqcap B} = \overline {A}\sqcup \overline{B}$. 
\end{defi}

The following are basic facts regarding negation of NNF terms.
\begin{prop}
For any term $A$ in NNF, the term $\overline{A}$ is also in NNF and it holds that $\mathcal{A}_{\textnormal{Riesz}}\vdash \overline{A} = -A$.
\end{prop}
\begin{proof}
  We prove the result by straightforward induction on $A$. See Lemma~\ref{lem:equalities}[7] below for the $\sqcup$ and $\sqcap$ cases.
\end{proof}

\begin{prop}
For any terms $A, B$ in NNF, it holds that $\negF{A}[B/x] = \negF{A[B/x]}$.
\end{prop}

\subsubsection{Technical lemmas regarding Riesz spaces}\label{some_small_lemmas_section}
We now list some useful facts that will be used throughout the paper. 
The following are useful derived operators frequently used in the theory of Riesz spaces:
\begin{center}
\begin{tabular}{| l |  l  | l |}
\hline
Symbol & Terminology & Definition\\
\hline
$A^+$ & The positive part & $A\sqcup 0$ \\
$A^-$ & The negative part & $(-A)\sqcup 0$ \\
$|A|$ & The absolute value & $A^+ + A^-$\\
\hline
\end{tabular}
\end{center}

\begin{lem}
  \label{lem:equalities}
  The following equations hold:
  \begin{itemize}
  \item (1) For all $A$ and $r>0$, $r(A^-) = (rA)^-$.
  \item (2) For all $A,B$, $A + B \leq 2(A \sqcup B)$
  \item (3) For all $A,B$, if $A \leq B$ then $B^- \leq A^-$.
  \item (4) For all $A,B$, $(A+B)^- \leq A^- + B^-$.
  \item (5) For all $r>0$, $0 \leq A$ if and only if $0 \leq rA$.
  \item (6) For all $A$,  $A = 0$ if and only if $-A = 0 $.
  \item (7) For all $A, B$, $-(A\sqcup B) = (-A) \sqcap (-B)$ and $-(A\sqcap B) = (-A) \sqcup (-B)$.
  \item (8) For all $A,B,C$, $A \sqcup (B \sqcap C) = (A \sqcup B) \sqcap (A \sqcup C)$ and $A \sqcap (B \sqcup C) =  (A \sqcap B) \sqcup (A \sqcap C)$.
  \item (9) For all $A,B,C$, $A + (B \sqcup C) = (A + B) \sqcup (A + C)$ and $A + (B \sqcap C) = (A + B) \sqcap (A + C)$. 
  \item (10) For all $A$, $0 \leq A \sqcup (-A)$.
  \end{itemize}
\end{lem}

Most notably, observe that Riesz spaces are distributive lattices (Lemma~\ref{lem:equalities}[8]), that sum distributes over lattice operations (Lemma~\ref{lem:equalities}[9]) and that the least upper bound of any element with its negation is always positive (Lemma~\ref{lem:equalities}[10]).

\begin{proof}
As mentioned in Example~\ref{example:techback1}, the Riesz space $\mathbb{R}$ is complete for the quasiequational theory of Riesz spaces. This means that a universally quantified Horn clause $\bigwedge_{i\in I} \mathcal{A}_{\textnormal{Riesz}}\vdash A_i=B_i \Rightarrow \mathcal{A}_{\textnormal{Riesz}}\vdash A=B$ holds in all Riesz spaces if and only if it holds in the Riesz space $(\mathbb{R},\leq)$. It is then straightforward to check the validity of all equations and equational implications in $\mathbb{R}$.
\end{proof}

\begin{lem}
  \label{lem:condMaxPos}
  For all $A,B$, $A \sqcup B \geq 0$ if and only if $A^- \sqcap B^- = 0$.
\end{lem}
\begin{proof} For all $A,B$ we have:
  \[
    \begin{array}{lll}
      0 \sqcap (A \sqcup B) & = & (A \sqcap 0) \sqcup (B \sqcap 0) \\
      & = & - (((-A) \sqcup (-0)) \sqcap ((-B) \sqcup (-0))) \\
      & = & - (A^- \sqcap B^-)
    \end{array}
  \]
Hence $0 \sqcap (A \sqcup B) = 0$ if and only if $-(A^- \sqcap B^-) = 0$ if and only (by Lemma~\ref{lem:equalities}[6]) $(A^- \sqcap B^-) = 0$.  The proof is complete recalling that $0\leq A\sqcup B$ means, by definition, that  $ 0 = 0\sqcap (A\sqcup B) $.\end{proof}

\subsection{Modal Riesz Spaces}

This section contains the basic definitions and results related to modal Riesz spaces, as introduced in~\cite{miolics2017,FMM2020}.

The language of modal Riesz spaces extends that of Riesz spaces with two symbols: a constant $1$ and a unary operator $\Diamond$.

\begin{defi}[Modal Riesz Space]
  \label{defi:modal_riesz_space}
The \emph{language} $\mathcal{L}^{\Diamond}_R$ of modal Riesz spaces is $\mathcal{L}_R \cup \{ 1, \Diamond\}$ where $\mathcal{L}_R$ is the language of Riesz spaces as specified in Definition~\ref{techback:definition:riesz_space}. A modal Riesz space is a $\mathcal{L}^{\Diamond}_R$-algebra satisfying the set $\mathcal{A}^{\Diamond}_{\textnormal{Riesz}}$ of axioms of Figure~\ref{axioms:of:modal_riesz:spaces}.
\end{defi}

\begin{figure}[ht]
\begin{center}
\begin{tabular}{l l l}
Axioms of Riesz spaces  & &  see Figure~\ref{axioms:of:riesz:spaces}\\
\hspace{2cm} $+$ 
\\
Positivity of $1$:  &$ \ $  & $0\leq 1$\\
Linearity of $\Diamond$: &$ \ $  &  $\Diamond( r_1 A + r_2 B) = r_1\Diamond(A) + r_2\Diamond(B)$\\
Positivity of $\Diamond$: &$ \ $  &   $\Diamond(0\sqcup A) \geq 0$ \\
 $1$-decreasing property of $\Diamond$: &$ \ $  &   $\Diamond(1) \leq 1$\\
\end{tabular}
\end{center}
\caption{Set $\mathcal{A}^{\Diamond}_{\textnormal{Riesz}}$ of equational axioms of modal Riesz spaces.}
\label{axioms:of:modal_riesz:spaces}
\end{figure}

\begin{exa}
Every Riesz space $R$ can be made into a modal Riesz space by interpreting $1$ with any positive element and by interpreting $\Diamond$ as the identity function ($\Diamond(x) = x$) or the constant $0$ function $\Diamond(x) = 0$.
\end{exa}
\begin{exa}
The Riesz space $(\mathbb{R},\leq)$ of linearly ordered real numbers becomes a modal Riesz space by interpreting $1$ with the number $1$, and $\Diamond$ by any linear (due to the linearity axiom) function $x \mapsto r x$ for a scalar $r\in\mathbb{R}$ such that $r\geq 0$ (due to the positivity axiom) and $r\leq 1$ (due to the $1$-decreasing axiom).  
\end{exa}

\begin{exa}
Generalising the previous example, the Riesz space $\mathbb{R}^n$ (with operations defined pointwise, see Example~\ref{example:techback2}) becomes a modal Riesz space by interpreting $1$ with the constant $1$ vector and $\Diamond$ by a linear  (due to the linearity axiom) map $M:\mathbb{R}^n \rightarrow \mathbb{R}^n$, thus representable as a square matrix,
\begin{center}
$1 = \begin{pmatrix}
1 \\ 1 \\ \vdots \\ 1
\end{pmatrix} \ \ \ \ \ \ \Diamond = 
\begin{pmatrix}
r_{1,1} & r_{1,2} & \cdots & r_{1,n} \\
r_{2,1} & r_{2,2} & \cdots & r_{2,n} \\
\vdots  & \vdots  & \ddots & \vdots  \\
r_{n,1} & r_{n,2} & \cdots & r_{n,n} 
\end{pmatrix}$
\end{center}
such that all entries $r_{i,j}$ are non-strictly positive (due to the positivity axiom) and where all the rows sum up to a value $\leq1$, i.e., for all $1\leq i \leq n$ it holds that $ \sum^k_{j=1} r_{i,j} \leq 1$ (due to the $1$-decreasing axiom). Such matrices are known as sub-stochastic matrices. Each sub-stochastic matrix $M$ can be regarded as a probabilistic transition system (also referred to as Markov chain) whose set $S$ of states is $S=\{s_1,\dots, s_n\}$ and whose transition function $\tau_M:S\rightarrow\mathcal{D}^{\leq 1}(S)$, defined as:
$$
\tau_M (s_i) (s_j) = r_{i,j} 
$$
 assigns to each state $s_i\in S$ a sub-probability\footnote{A sub-probability distribution on a set $S$ is a function $d:S \rightarrow [0,1]$ such that $\sum_s d(s) \leq 1$.} distribution $\tau_M(s_i) \in \mathcal{D}^{\leq 1}(S)$ specifying the probability of reaching $s_j$ from $s_i$, for any $s_i,s_j\in S$.
 
 For a concrete example, consider the modal Riesz space $\mathbb{R}^2$ with $\Diamond$ interpreted by the matrix $M$ defined as:
\begin{center}
$M =
\begin{pmatrix}
\frac{1}{3} & \frac{1}{2}  \\[0.15em]
\frac{1}{3} & 0  \\
\end{pmatrix}$.
\end{center}
This modal Riesz space can be identified with the Markov chain having state space $S=\{s_1,s_2\}$ and transition function $\tau_M$ defined by: $\tau_M(s_1) = (s_1\mapsto \frac{1}{3}, s_2\mapsto \frac{1}{2})$ and $\tau_M(s_2) = (s_1\mapsto \frac{1}{3}, s_2\mapsto 0)$:
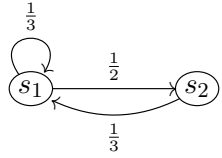
\begin{center}
$$
\SelectTips{cm}{}
	\xymatrix @=20pt {
		\nodeC{s_1} \ar@{->}[rr]^{\frac{1}{2}}\ar@{->}@(ul,ur)^{\frac{1}{3}}   & &  \nodeC{s_2} \ar@{->}@/^10pt/[ll]^{\frac{1}{3}}   	}
$$
\end{center}
From the state $s_1$ the computation progresses to $s_1$ itself with probability $\frac{1}{3}$, to $s_2$ with probability $\frac{1}{2}$ and it halts with probability $\frac{1}{6}$  (i.e., with the remaining probability $1-(\frac{1}{2}+\frac{1}{3})$).  From the state $s_2$ the computation progresses to $s_1$ with probability $\frac{1}{3}$ and it  halts with probability $\frac{2}{3}$.

\end{exa}

\begin{exa}[Transition Semantics]\label{logic:example1a}
Carrying on the previous example, given any Markov chain $(S,\tau_M)$ (i.e., equivalently, a modal Riesz space on $\mathbb{R}^n$ with $\Diamond$ interpreted by a sub-stochastic matrix $M$), each closed (i.e., without variables) modal Riesz term $A$ is interpreted as a function $\llbracket A\rrbracket=S\rightarrow\mathbb{R}$ (i.e., a vector in $\mathbb{R}^n$). This interpretation is inductively defined as:

\begin{center}
$\llbracket 0\rrbracket (s_i) = 0 \ \ \ \ \ \llbracket 1\rrbracket (s_i) = 1$

$\ $\\

$ \llbracket r A\rrbracket(s_i) = r \cdot\big(\llbracket A\rrbracket(s_i)\big) \ \ \ \ \ \ \llbracket A+B\rrbracket(s_i) = \llbracket A\rrbracket(s_i) + \llbracket B\rrbracket(s_i)$

$\ $\\

$\llbracket A\sqcup B\rrbracket(s_i) = \max\{ \llbracket A\rrbracket(s_i) ,  \llbracket B\rrbracket(s_i)\}  \ \ \ \ \ \ \llbracket A\sqcap B\rrbracket(s_i) = \min\{ \llbracket A\rrbracket(s_i) ,  \llbracket B\rrbracket(s_i)\}$

$\ $\\

$\llbracket \Diamond A\rrbracket(s_i) =\displaystyle \sum^n_{j=1} \big( \tau_M(s_i)(s_j)\cdot \llbracket A\rrbracket(s_j) \big)$
\end{center}
for all $s_i\in S$. The semantics $\llbracket A\rrbracket$ of a closed term $A$ can be understood as a (real-valued) quantitative property of states of Markov chains and $\llbracket \Diamond A\rrbracket$ denotes the expected value of $\llbracket A\rrbracket$ after a transition step.  This fact provides the motivation for the study of modal Riesz spaces in~\cite{miolics2017,FMM2020} in the context of logics for probabilistic programs.

Consider for example the Markov chain of the previous example: 
$$
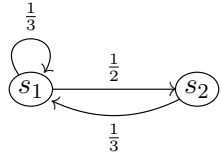
\SelectTips{cm}{}
	\xymatrix @=20pt {
		\nodeC{s_1} \ar@{->}[rr]^{\frac{1}{2}}\ar@{->}@(ul,ur)^{\frac{1}{3}}   & &  \nodeC{s_2} \ar@{->}@/^10pt/[ll]^{\frac{1}{3}}   	}
$$
and the modal Riesz terms $\Diamond 1$ and $\Diamond(\Diamond 1)$. They are interpreted as the two functions on $S=\{s_1,s_2\}$ illustrated as vectors below:
\begin{center}
$ \llbracket \Diamond 1\rrbracket = \begin{pmatrix}
  \frac{5}{6} \\[0.15em] \frac{1}{3}
\end{pmatrix}  \ \ \ \  \textnormal{ and } \ \ \ \  \llbracket \Diamond(\Diamond 1)\rrbracket = \begin{pmatrix}
  \frac{4}{9} \\[0.15em] \frac{5}{18}
\end{pmatrix}$
\end{center}
The term $\Diamond 1$ assigns to each state $s_i\in\{s_1,s_2\}$ the probability of making a computational step from $s_i$ to any other state (and thus not halting). Similarly, the term $\Diamond\Diamond 1$ is the vector assigning to each state $s_i$ the probability of making two consecutive computational steps from $s_i$. More complex terms can express sophisticated real-valued properties of Markov chains~\cite{miolics2017,FMM2020}. Importantly, in~\cite{FMM2020} the transition semantics\footnote{In fact, a generalised topological semantics is required for completeness, where Markov chains have  state spaces endowed with a compact Hausdorff topology. The examples above have finite state spaces, and are thus trivially compact Hausdorff.} is proved to be sound and complete with respect to the equational theory of modal Riesz spaces: two modal Riesz terms $A$ and $B$ are provably equal from the axioms of modal Riesz spaces if and only if $\llbracket A\rrbracket=\llbracket B\rrbracket$ when interpreted in all possible Markov chains. We refer to~\cite{FMM2020} for further details and comparisons with other logics having a real-valued semantics such as~\cite{DBLP:conf/ijcai/SchroderP11,StonePrakash,MioSimpsonFI2017,Kozen1983,MM07} among others.
\end{exa}

\begin{exa}\label{example:false-equality}
Consider the equality $\Diamond (x \sqcup y) = \Diamond(x) \sqcup \Diamond(y)$. Does it hold in all modal Riesz spaces? In other words, does $\mathcal{A}^{\Diamond}_{\textnormal{Riesz}}\vdash \Diamond (x \sqcup y) = \Diamond(x) \sqcup \Diamond(y)$? The answer is negative. Take as example the modal Riesz space $\mathbb{R}^2$ with:
\begin{center}
$1 = \begin{pmatrix}
1 \\ 1 \end{pmatrix} \ \ \ \ \ \ \Diamond = 
\begin{pmatrix}
\frac{1}{3} & \frac{2}{3}\\
0 & 0  \\
\end{pmatrix}$
\end{center}
and let $a = (1,0)$ and $b=(0,1)$. One verifies that $\Diamond(a\sqcup b) = (1,0)$ while $\Diamond(a) \sqcup \Diamond(b) =  (\frac{2}{3},0)$. This example shows that unlike the theory of Riesz spaces (cf. Example~\ref{example:techback1}), the theory of modal Riesz spaces cannot be generated by a linear model, i.e. a model where either $a \leq b$ or $b \leq a$ for all $a$ and $b$. Indeed, in any linear model, the equality $\Diamond (x \sqcup y) = \Diamond(x) \sqcup \Diamond(y)$ clearly holds while it does not hold in the example above.
\end{exa}

\begin{rem}
The choice of using the $\Diamond$ symbol for the unary operation of modal Riesz spaces might suggest the existence of a distinct De Morgan dual operator $\Box x = -\Diamond-x$. This is not the case since, due to linearity, $\Box x = \Diamond x$, i.e., $\Diamond$ is self dual. While using a different symbol such as ($\circ$) might have been a better choice, we decided to stick to $\Diamond$ for backwards compatibility with previous works on modal Riesz spaces~\cite{miolics2017,FMM2020,DBLP:conf/fossacs/LucasM19}. Another source of potential ambiguity lines in the ``modal'' adjective itself. Of course other axioms for $\Diamond$ can be conceived (e.g., $\Diamond(x\sqcup y ) = \Diamond(x) \sqcup \Diamond(y)$ instead of our $\Diamond(x + y ) = \Diamond(x)+\Diamond(y)$, see, e.g.,~\cite{DBLP:journals/lmcs/DiaconescuMS18}). Therefore different notions of modal Riesz spaces can be investigated, just like many types of classical modal logic exist (K, S4, S5, etc). Once again, our choice of terminology is motivated by backwards compatibility with previous works.
\end{rem}

We now expand the definitions and properties related to terms in negation normal form to modal Riesz spaces.

\begin{defi}\label{nnf:modal_definition}
A term $A$ is in \emph{negation normal form} (NNF) if the operator $(-)$ is only applied to variables and the constant $1$. 
\end{defi}

\begin{lem}
Every term $A$ can be rewritten to an equivalent term in NNF.
\end{lem}
\begin{proof}
Negation can be pushed towards the variables by the following rewritings: $-\Diamond(A) = \Diamond(-A)$ (see Lemma~\ref{lem:into_nnf} for the other operators).
\end{proof}

Negation can be defined on terms in NNF as follows.

\begin{defi}
Given a term $A$ in NNF, we expand the operator $\overline{A}$ as follows: $\overline{\Diamond A} = \Diamond \negF{A}$, $\negF{1} = -1$, $\negF{-1}=1$. 
\end{defi}

The following are basic facts regarding negation of NNF terms.
\begin{prop}
For any term $A$ in NNF, the term $\overline{A}$ is also in NNF and it holds that $\mathcal{A}_{\textnormal{Riesz}}\vdash \overline{A} = -A$.
\end{prop}

\begin{prop}
For any terms $A, B$ in NNF, it holds that $\negF{A}[B/x] = \negF{A[B/x]}$.
\end{prop}


\section{Hypersequent Calculus for Riesz Spaces}\label{modal_free_section}

In this section we introduce the hypersequent calculus \hr\ for the equational theory of Riesz spaces.

In what follows we proceed with a sequence of syntactical definitions and notational conventions necessary to present the rules of the system. We use the letters $A$, $B$, $C$ to range over Riesz terms in negation normal form (NNF, see Definition~\ref{nnf:definition}) built from a countable set of variables $x$, $y$, $z$ and negated variables $\overline{x}$, $\overline{y}$, $\overline{z}$.
The scalars appearing in these terms are all strictly positive and are ranged over by the letters $r,s,t\in{\mathbb{R}_{>0}}$. From now on, the term scalar should always be understood as strictly positive scalar.

\begin{defi}
A \emph{weighted term} is a formal expression $r. A$ where $r\in\mathbb{R}_{>0}$ and $A$ is a term. 
\end{defi}

Given a weighted term $r.A$ and a scalar $s$ we denote with $s.(r.A)$ the weighted term $(sr).A$. Thus we have defined (strictly positive) scalar multiplication on weighted terms.

We use the greek letters $\Gamma, \Delta, \Theta, \Sigma$ to range over possibly empty finite multisets of weighted terms. We often write these multisets as lists but they should always be understood as being taken modulo reordering of their elements. As usual, we write $\Gamma,\Delta$ for the concatenation of  $\Gamma$ and $\Delta$.

We adopt the following notation:
\begin{itemize}
\item Given a sequence $\vec{r}=(r_1,\dots r_n)$ of scalars and a term $A$, we denote with $\vec{r}.A$ the multiset $[r_1.A, \dots, r_n.A]$. When $\vec{r}$ is empty, the multiset $\vec{r}.A$ is also empty.
\item Given a multiset $\Gamma=[r_1.A_1, \dots, r_n.A_n]$ and a scalar $s>0$, we denote with $s.\Gamma$ the multiset $[s.r_1.A_1, \dots, s.r_n.A_n]$.
\item Given a sequence $\vec{s}=(s_1,\dots s_n)$ of scalars and a multiset $\Gamma$, we denote with $\vec{s}.\Gamma$ the multiset $s_1.\Gamma, \dots, s_n.\Gamma$.
\item Given two sequences $\vec{r}=(r_1,\dots r_n)$ and $\vec{s}=(s_1,\dots s_m)$ of scalars, we denote $\app{\vec{r}}{\vec{s}}$ the concatenation of the two sequences, i.e. the sequence $(r_1,\dots r_n,s_1, \dots s_m)$.
\item Given a sequence $\vec{s}=(s_1,\dots s_n)$ of scalars and a scalar $r$, we denote $(r\vec{s})$ the sequence $(rs_1,\dots rs_n)$.
\item Given two sequences $\vec{r}=(r_1,\dots r_n)$ and $\vec{s}=(s_1,\dots s_m)$ of scalars, we denote $\vec{r}\vec{s}$ the sequence $\app{r_1\vec{s}}{\app{\dots}{r_n\vec{s}}}$.
\item Given a sequence $\vec{s}=(s_1,\dots s_n)$ of scalars, we denote $\sum \vec{s}$ the sum of all scalars in $\vec{s}$, i.e. the scalar $\sum\limits_{i=1}^n s_i$.
\end{itemize}

\begin{defi}
A \emph{sequent} is a formal expression of the form $\vdash \Gamma$.
\end{defi}
If $\Gamma=\emptyset$, the corresponding empty sequent is simply written as $\vdash$.

\begin{defi}
A \emph{hypersequent} is a non-empty finite multiset of sequents, written as $\vdash \Gamma_1 | \dots | \vdash \Gamma_n$.
\end{defi}
We use the letter $G,H$ to range over hypersequents. Note that, under these notational conventions, the expression $\vdash \Gamma$ could either denote the sequent $\vdash \Gamma$ itself or the hypersequent $[\vdash \Gamma]$ containing only one sequent. The context will always determine which of these two interpretations is intended.

We now describe how sequents and hypersequents can be interpreted by Riesz terms. This means that  \hr\ is a \emph{structural proof system}, i.e., by manipulating sequents and hypersequents it in fact deals with terms of a certain specific form.

\begin{defi}[Interpretation]
  \label{defi:modal_free_interpretation}
We interpret weighted terms ($r.A$), sequents $\vdash\Gamma$ and hypersequents $G$ as the Riesz terms $\sem{r.A}$, $\sem{\vdash\Gamma}$ and $\sem{G}$, respectively, as follows:
\begin{center}

\begin{tabular}{l|  l l l}
& Syntax &  & Term interpretation $\sem{\_}$ \\
\hline
Weighted terms & $r.A$ & & $rA$\\
 Sequents & $\vdash r_1.A_1,\dots, r_n.A_n$ & & $\sem{r_1.A_1} + \dots + \sem{r_n.A_n}$\\
 Hypersequents & $\vdash\Gamma_1 | \dots | \vdash\Gamma_n$ && $\sem{\vdash\Gamma_1} \sqcup \dots \sqcup \sem{\vdash\Gamma_n}$
\end{tabular}
\end{center}
\end{defi}

Hence a weighted term is simply interpreted as the term scalar-multiplied by the weight. A sequent is interpreted as sum $(\sum)$ and a hypersequent is interpreted as a join of sums ($\bigsqcup \sum)$.

\begin{exa}
The interpretation of the hypersequent:
$$
\vdash 1.x, 2.(y\sqcap z) \ | \ \vdash 2.(3\overline{x} \sqcap y) 
$$
is the Riesz term:
$$\big(1x  + 2(y\sqcap z)\big) \sqcup  \big( 2(3\overline{x} \sqcap y) \big).$$
\end{exa}

The hypersequent calculus \hr\ is a deductive system for deriving hypersequents whose interpretation is positive, i.e., the hypersequents $G$ such that $\mathcal{A}_{\textnormal{Riesz}} \vdash 0 \leq \sem{G}$. The rules of \hr\ are presented in Figure~\ref{rules:hr} and are very similar to the rules of the system GA of~\cite{MOG2005,MOGbook} (see Appendix~\ref{ga:section}) where the main difference is the use of weighted terms in sequents. We write $\prove G$ if the hypersequent $G$ is derivable in the system $\hr$.

\begin{figure}[h!]
  \begin{center}
  \scalebox{0.85}{
      \begin{minipage}{12.5cm}
        \textbf{Axiom:}
        \begin{center}
            
            $\infer[\text{INIT}]{\vdash}{}$
        \end{center}
        \textbf{Structural rules:}
        \begin{center}
          \begin{tabular}{ccc}
           $ \infer[\text{W}]{G \sep \vdash \Gamma}
                      {G}$
                      & &
                      $\infer[\text{C}]{G \sep \vdash \Gamma}
                                {G \sep \vdash \Gamma \sep \vdash \Gamma}$\\[0.5cm]
            $\infer[\text{S}]
                      {G \sep \vdash \Gamma_1 \sep  \vdash \Gamma_2}
                      {G \sep \vdash \Gamma_1, \Gamma_2}$
                      &  &
                      $\infer[\text{M}]
                                {G \sep  \vdash \Gamma_1 , \Gamma_2}
                                {G \sep \vdash \Gamma_1 & G \sep \vdash \Gamma_2}$         \\[0.5cm]
            $\infer[\text{T}]{G \sep \vdash \Gamma}{G \sep \vdash r.\Gamma}$ & & $\infer[\text{ID},\sum r_i = \sum s_i]{G \sep \vdash \Gamma, \vec{r}.x,\vec{s}.\covar{x}}{G \sep \vdash \Gamma}$
          \end{tabular}
        \end{center}
        \textbf{Logical rules:}
        \begin{center}
          \begin{tabular}{ccc}
            $\infer[0]
            {G \sep \vdash \Gamma , \vec{r}.0}
            {G \sep \vdash \Gamma}$
            &
            $\infer[+]
            {G \sep \vdash \Gamma, \vec{r}.(A + B)}
            {G \sep \vdash \Gamma , \vec{r}.A , \vec{r}.B}$
            &
            $\infer[\times]
            {G \sep \Gamma \vdash \Gamma , \vec{r}.(s A)}
            {G \sep \Gamma \vdash \Gamma , (s\vec{r}).A}$ \\[0.4cm]
          \end{tabular}
          \begin{tabular}{cc}
            $\infer[\sqcup]
            {G \sep\vdash \Gamma , \vec{r}.(A \sqcup B)}
            {G \sep\vdash \Gamma , \vec{r}.A \sep \vdash \Gamma , \vec{r}.B}$
            &
            $\infer[\sqcap]
            {G \sep \vdash \Gamma , \vec{r}.(A \sqcap B)}
              {G \sep  \vdash \Gamma , \vec{r}.A & G \sep \vdash \Gamma , \vec{r}.B}$ \\[0.2cm]
          \end{tabular}
                  \end{center}
                  
                   \textbf{CAN rule:}
                   \begin{center}
                     $\infer[\text{CAN}, \sum r_i = \sum s_i]{G \sep \vdash \Gamma}{G \sep \vdash \Gamma, \vec{s}.A , \vec{r}.\negF{A}} $
        \end{center}

            \end{minipage}
    }
  \end{center}
  \caption{Inference rules of \hr.}
  \label{rules:hr}
\end{figure}

The axiom INIT allows for the derivation of $(\vdash)$, the hypersequent containing only the empty sequent, thus it corresponds to the positivity of the constant $0$. The C rule (contraction) allows treating hypersequents as (always non-empty) sets of sequents. The M (mix) and S (split) rules are as in the system GA of~\cite{MOG2005,MOGbook}. We instead adopted the rule $\textnormal{ID}$, in place of the axiom $\textnormal{ID-ax}$ of GA (see Appendix~\ref{ga:section}). While the two are equivalent (i.e., mutually derivable) in presence of the other rules, the formulation of $\textnormal{ID}$ as a rule is convenient in the statement of the $M$-elimination theorem later on. The $T$ rule is novel, and can be seen as a real-valued variant of C (contraction) rule in that the weight of a sequent in the hypersequent can be multiplied by an arbitrary positive real number. Finally, note that the logical rules are all presented using the syntactic sugaring $\vec{r}.A$ described above. For example, one valid instance of the rule $(+)$ is the following:
$${\infer[+]{\vdash \Gamma, 2.(3y + x), \frac{1}{2}(3y+x)}{\vdash \Gamma, 2.3y, 2.x, \frac{1}{2}.3y, \frac{1}{2}.x}}$$
This effectively allows us to apply the rule to several terms in the sequent at the same time. This feature adds some flexibility in the process of derivation construction and simplifies some proofs, but it is not strictly required. All our results hold even in a variant of the \hr\ system where rules are allowed to act on only one term at the time.

\begin{conv}
  We often have to use the same rule multiple times when building a derivation. For convenience, we may write the rule only once with the number of times the rule is used as exponent, as follows:
  \[ \infer[W^2]{G \sep \vdash \Gamma \sep \vdash \Delta}{G} \]
  If the number of times a rule is used is not known, we use a wildcard as exponent, as in the following example where the weakening rule is used to remove all sequents appearing in G:
  \[ \infer[W^*]{G \sep \vdash \Gamma}{\vdash \Gamma} \]
\end{conv}

\begin{rem}
On the one hand, we could have introduced appropriate exchange (i.e., reordering) rules and defined sequents and hypersequents as lists, rather than multisets. In the opposite direction, we could have defined hypersequents as (non-empty) sets and dispose of the rules (C). Our choice is motivated by a balance between readability and fine control over the derivation steps in the proofs. 
\end{rem}

\begin{rem}\label{cutcomment1}
Note that the following CUT rule 
$$ \infer[\text{CUT},  \sum \vec{r} = \sum \vec{s}]{G \sep \vdash \Gamma_1, \Gamma_2}{G \sep \vdash \Gamma_1, \vec{r}.A & G \sep \vdash \Gamma_2, \vec{s}.\negF{A}} $$
is equivalent (i.e., mutually derivability) to the CAN rule in the \hr\ hypersequent calculus:
 
\begin{figure}[H]
    \[ \infer[\text{CAN},  \sum \vec{r} = \sum \vec{s}]{G \sep \vdash \Gamma_1, \Gamma_2}{\infer[\text{M}]{G \sep \vdash \Gamma_1, \Gamma_2, \vec{r}.A,\vec{s}.\negF{A}}{G \sep \vdash \Gamma_1, \vec{r}.A & G \sep \vdash \Gamma_2, \vec{s}.\negF{A}}} \]

  \caption{Derivability of the CUT rule.}
  \label{fig:enc_cut}
\end{figure}
\begin{figure}[H]
    \[ 
    \infer[\text{C}]{G \sep \vdash \Gamma}{\infer[\text{S}]{G \sep \vdash \Gamma \sep \vdash \Gamma}{\infer[\text{CUT}]{G \sep \vdash \Gamma, \Gamma}{\infer[+]{G \sep \vdash \Gamma, \vec{r}.(A + \negF{A})}{\infer[\text{CUT}]{G \sep \vdash \Gamma,\vec{r}.A,\vec{r}.\negF{A}}{G \sep \vdash \Gamma, \vec{r}.A,\vec{s}.\negF{A} & \infer[\text{W}^*]{G \sep \vdash \vec{s}.A,\vec{r}.\negF{A}}{\infer[\text{Lemma~\ref{lem:modal_free_ext_ID_rule}}]{\vdash \vec{s}.A,\vec{r}.\negF{A}}{\infer[\text{INIT}]{\vdash}{}}}}} & \infer[+]{G \sep \vdash \Gamma, \vec{r}.(\negF{A} + A)}{\infer[\text{CUT}]{G \sep \vdash \Gamma,\vec{r}.A,\vec{r}.\negF{A}}{G \sep \vdash \Gamma, \vec{r}.A,\vec{s}.\negF{A} & \infer[\text{W}^*]{G \sep \vdash \vec{s}.A,\vec{r}.\negF{A}}{\infer[\text{Lemma~\ref{lem:modal_free_ext_ID_rule}}]{\vdash \vec{s}.A,\vec{r}.\negF{A}}{\infer[\text{INIT}]{\vdash}{}}}}}}}}
     \]
    
  \caption{Derivability of the CAN rule.}
  \label{fig:enc_can}
\end{figure}
Our choice (following~\cite{MOGbook, MOG2005}) of presenting the system \hr\ using the CAN rule, rather than the equivalent CUT rule, is just motivated by elegance and technical convenience.
\end{rem}

\begin{exa}
  Example of derivation of the hypersequent $\vdash 1.\big((2x + 2\overline{y})\sqcup (y + \overline{x})\big)$ which consists of only one sequent. 
  \[ \infer[\sqcup]{\vdash 1.\big((2x + 2\overline{y})\sqcup (y + \overline{x})\big)}{\infer[+]{\vdash 1.(2x + 2\covar{y}) \sep \vdash 1.(y + \covar{x})}{\infer[\times]{\vdash 1.2x, 1.2\covar{y} \sep \vdash 1.(y + \covar{x})}{\infer[\times]{\vdash 2.x, 1.2\covar{y} \sep \vdash 1.(y + \covar{x})}{\infer[+]{\vdash 2.x, 2.\covar{y} \sep \vdash 1.(y + \covar{x})}{\infer[\text{T(multiplication by $2$)}]{\vdash 2.x, 2.\covar{y} \sep \vdash 1.y, 1.\covar{x}}{\infer[\text{S}]{\vdash 2.x, 2.\covar{y} \sep \vdash 2.y, 2.\covar{x}}{\infer[\text{ID}]{\vdash 2.x,2.y,2.\covar{x},2.\covar{y}}{\infer[\text{ID}]{\vdash 2.y, 2.\covar{y}}{\infer[\text{INIT}]{\vdash}{}}}}}}}}}} \]
\end{exa}

In what follows we say that an hypersequent $G$ has a CAN-free derivation (resp., M-free, T-free, \emph{etc.}) if it has a derivation that never uses the rule CAN (resp., rule M, rule T, \emph{etc.}).


\subsection{Main results regarding the system \hr}
\label{subsec:modal_free_main_results}

We are now ready to state the main results regarding the hypersequent calculus \hr. Each theorem will be proven in a separate subsection of this section.

 Recall that we write $\mathcal{A}_{\textnormal{Riesz}} \vdash A \geq B$ if the inequality $A\geq B$ is derivable in equational logic from the axioms of Riesz spaces and that we write $\prove G$ if the hypersequent $G$ is derivable in the \hr\ proof system.

Our first technical result states that the system \hr\ can derive all and only those hypersequents $G$ such that $\mathcal{A}_{\textnormal{Riesz}} \vdash \sem{G}\geq 0$.

\begin{thm}[Soundness]
\label{thm:modal_free_soundness}
For every hypersequent $G$, 
$$
\prove G\ \ \ \Longrightarrow\ \ \ \mathcal{A}_{\textnormal{Riesz}}\vdash \sem{G} \geq 0.
$$
\end{thm}

\begin{thm}[Completeness]
\label{thm:modal_free_completeness}
For every hypersequent $G$, 
$$
\mathcal{A}_{\textnormal{Riesz}} \vdash\sem{G} \geq 0 \ \ \ \Longrightarrow\ \ \  \prove G.
$$
\end{thm}

Our next theorem states that all the logical rules of the hypersequent calculus \hr\ are \emph{CAN-free invertible}. This means that if an hypersequent $G$ having the shape of the conclusion of a logical rule is derivable with a CAN-free derivation, then also the premises of that logical rule are derivable by CAN-free derivations. So, for example, in the case of the ($\sqcap$) rule, if the hypersequent $$G \sep \vdash \Gamma , \vec{r}.(A \sqcap B)$$ has a CAN-free derivation, then also 
$$G \sep \vdash \Gamma , \vec{r}.A \ \ \ \ \textnormal{and} \ \ \ \  G \sep \vdash \Gamma , \vec{r}.B $$ 
have CAN-free derivations.

\begin{thm}[CAN-free Invertibility]
\label{thm:modal_free_invertibility}
All the logical rules are CAN-free invertible.
\end{thm}

The invertibility theorem is very important for proof search. When trying to derive a hypersequent $G$ (without CAN applications) it is always possible to systematically apply the logical rules and reduce the problem of deriving $G$ (without CAN applications) to the problem of deriving a number of hypersequents $G_1, \dots G_n$ where no logical symbols appear. We call such reduced hypersequents without logical symbols \emph{atomic hypersequents}.

\begin{figure}[H]
  \[ \infer[\text{Logical rules}]{G}{\infer{\qquad \qquad \ddots \qquad \qquad}{G_1} & \vdots  & \infer{\qquad \qquad \iddots \qquad \qquad}{G_n}} \]
  \caption{Systematic application of the logical rules to reduce the logical complexity.}
  \label{fig:algo_to_inv}
\end{figure}

As we will discuss later (Theorem~\ref{thm:modal_free_decidability}), this procedure of simplification will lead to an algorithm for deciding if an arbitrary hypersequent $G$ is derivable in \hr\  or not.

The three theorems above are adaptations of similar results for the hypersequent calculus \gasystem\ of~\cite{MOGbook, MOG2005} for the theory of lattice ordered abelian groups. 

The following theorem, instead, appears to be novel. It is stated in the context of our system \hr\ but a similar result can be proved for \gasystem\ too.

\begin{thm}[M-elimination]
\label{thm:modal_free_m_elim}
If a hypersequent has a CAN-free derivation, then it has a CAN-free and M-free derivation.
\end{thm}

Our motivation for proving the above result is mostly technical. Indeed it allows us to prove our main theorem (Theorem~\ref{thm:modal_free_can_elim} below) in a rather simple way (different from that of~\cite{MOGbook, MOG2005}). However note how the M-elimination theorem is also useful from the point of view of proof search since it reduces the space of derivation trees to be explored.

We are now ready to state our main result regarding the system $\hr$.

\begin{thm}[CAN elimination]
\label{thm:modal_free_can_elim}
If a hypersequent $G$ has a derivation, then it has a CAN-free derivation.
\end{thm}

\begin{proof}[Proof sketch] The CAN rule has the following form:
$$\infer[\text{CAN}, \sum \vec{r} = \sum \vec{s}]{G \sep \vdash \Gamma}{G \sep \vdash \Gamma, \vec{s}.A , \vec{r}.\negF{A}} $$

We show how to eliminate one application of the CAN rule. Namely, we prove that if the premise $G \sep \vdash \Gamma, \vec{s}.A , \vec{r}.\negF{A}$ has a CAN-free derivation then the conclusion $G \sep \vdash \Gamma$ also has a CAN-free derivation. This of course implies the statement of the CAN-elimination theorem by using a simple inductive argument on the number of CAN's applications in a derivation.

As a preliminary step, we first invoke the M-elimination Theorem~\ref{thm:modal_free_m_elim} on the derivation of $G \sep \vdash \Gamma, \vec{s}.A , \vec{r}.\negF{A}$ to remove possible occurrences of the M rule. In other words, we can assume that the derivation of $G \sep \vdash \Gamma, \vec{s}.A , \vec{r}.\negF{A}$ does not contain applications of the M rule. This is important since the M rule is problematic to deal with in our inductive proof because its two premises can generally break the symmetry between  the weights of $A$ and $\negF{A}$ in the hypersequent. For instance, the induction hypothesis could not be used on the premises of the following instance of the M rule since the condition $\sum\vec{r} = \sum \vec{s}$ is not satisfied in either of the two premises: \[ \infer[\text{M}]{G \sep \vdash \Gamma, \vec{s}.A , \vec{r}.\negF{A}}{G\sep \vdash \Gamma, \vec{s}.A & G\sep \vdash \vec{r}.\negF{A}} \]

Hence, in what follows we assume that the  the derivation of $G \sep \vdash \Gamma, \vec{s}.A , \vec{r}.\negF{A}$ is M-free and the proof proceeds by induction on the structure of $A$.

The base case in when $A=x$, i.e., when $A$ is atomic.  Proving this case is relatively straightforward, once the critical case regarding the M rule can be ignored, as explained above.


For the inductive case, when $A$ is a complex term we invoke the invertibility theorem. For example, if $A=B+C$, the invertibility theorem states that $G \sep \vdash \Gamma, \vec{s}.B,  \vec{s}.C, \vec{r}.\negF{B},\vec{r}.\negF{C}$ must also have a CAN-free derivation (and also M-free by application of the M-elimination Theorem~\ref{thm:modal_free_m_elim}). We then note that, since $B$ and $C$ both have lower complexity than $A$, it follows from two applications of the inductive hypothesis that $G \sep \vdash \Gamma$ has a M-free CAN-free derivation, as desired.
\end{proof}

\begin{rem}
Note, with reference to Remark~\ref{cutcomment1}, that Theorems~\ref{thm:modal_free_m_elim} and~\ref{thm:modal_free_can_elim} together imply also a CUT-elimination theorem. 
\end{rem}

The CAN rule is not analytical, meaning that in its premise there is a term not appearing (even as a subterm) in the conclusion. This is why the above CAN-elimination is of key importance, especially in the context of proof search.

However there is another rule of \hr\ which is not analytical: the T rule. The following theorem shows that also the T rule is admissible if the scalars appearing in the end hypersequent $G$ are all rational numbers.

\begin{thm}[Rational T-elimination]
\label{thm:modal_free_conservativity}
If a hypersequent $G$ with only rational numbers has a CAN-free derivation, then it has also a CAN-free and T-free derivation. 
\end{thm}

It can be shown, however, that if $G$ contains irrational numbers, it is generally not possible to eliminate both rules CAN and T at the same time.

\begin{prop}
  \label{prof:can_or_T}
The system \hr\ without the CAN and T rules is incomplete. 
\end{prop}

As mentioned earlier, using the invertibility theorem, it is possible to reduce the problem of deriving an hypersequent $G$ to the problem of deriving a number $G_1, \dots G_n$ of atomic (i.e., without logical symbols) hypersequents. This leads us to the  following result.

\begin{thm}[Decidability]
\label{thm:modal_free_decidability}
There is an algorithm to decide whether or not a hypersequent has a derivation.
\end{thm}

We remark that the above statement follows easily by the soundness Theorem~\ref{thm:modal_free_soundness} and the completeness Theorems~\ref{thm:modal_free_completeness}  and the known fact that the equational theory of Riesz spaces is decidable (see Example~\ref{example:techback1}).  Interestingly, however, we present an alternative proof that will be adaptable, in Section~\ref{subsec:decidability}, to obtain a decidability result for the system \hmr. The proof is based on considering a generalization of the concept of derivation where the scalars appearing in the hypersequents can be variables, rather than numerical constants. For instance, the derivation
$$\infer[\text{ID}]{ \vdash \Gamma, r.x,r.\covar{x}}{ 
\infer[\text{INIT}]{\vdash}{}}$$
\noindent is valid for any $r\in\mathbb{R}_{>0}$  and, similarly, the derivation
$$\infer[\text{ID}, r = s + t]{ \vdash \Gamma, r.x,s.\covar{x}, t.\covar{x}}{ 
\infer[\text{INIT}]{\vdash}{}}$$

\noindent
is valid for any values of reals $(r,s,t)\in\mathbb{R}_{>0}^3$ such that $r = s + t$. Lastly, the hypersequent containing two scalar-variables $\alpha,\beta$ and two concrete scalars $s$ and $t$
 $$
\vdash (\alpha^2 - \beta).x, s.\covar{x}, t.\covar{x}
$$ 
is derivable for any  assignment of concrete assignments $r_1,r_2\in \mathbb{R}$ to $\alpha$ and $\beta$ such that $(r_1)^2-r_2 > 0$ and $(r_1)^2-r_2 = s+t$.   Hence a hypersequent can be interpreted as describing the set of possible assignments to these real-valued variables that result in a valid concrete (i.e., where all scalars are numbers and not variables) derivation. 

The main idea behind the proof of Theorem~\ref{thm:modal_free_decidability} is that it is possible, given an arbitrary hypersequent $G$, to construct (automatically) a formula in the first order theory of the real closed field (FO$(\mathbb{R},+,\times, \leq)$) describing the set of valid assignments. Since this theory is decidable and has quantifier elimination~\cite{tarski1951}, it is possible to verify if this set is nonempty and extract a valid assignment to variables.


\subsection{Relations with the calculus \gasystem\ and l-ordered Abelian groups}

As mentioned earlier, our hypersequent calculus system \hr\ for the theory of Riesz spaces is an extension of the system \gasystem\ of~\cite{MOGbook, MOG2005} for the theory of lattice-ordered Abelian groups (laG). The equational theory ($\mathcal{A}_{\textnormal{laG}}$) of lattice-ordered Abelian groups can be defined by removing, from the signature of Riesz spaces, the scalar multiplication operations and, accordingly, the equational axioms regarding scalar multiplication. Integer scalars (e.g., $-3x$) can still be used as a short hand for repeated sums (e.g., $-(x + x + x)$). The system \hr\ stripped out of scalars is essentially identical to the system \gasystem. 

From our Rational T-elimination Theorem~\ref{thm:modal_free_conservativity} we obtain as a corollary the fact that the theory of Riesz spaces is a proof-theoretic conservative extension of the theory of lattice-ordered Abelian groups. 

\begin{prop}
Let $A$ be a term in the signature of lattice-ordered Abelian groups (i.e., a Riesz term where all scalars are natural numbers). Then 
$$
 \mathcal{A}_{\textnormal{Riesz}}\vdash A \geq 0  \Leftrightarrow \mathcal{A}_{\textnormal{laG}}\vdash A\geq 0.
$$

\end{prop}
\begin{proof}
The ($\Leftarrow$) direction is trivial, since $\mathcal{A}_{\textnormal{Riesz}}$ is an extension of  $\mathcal{A}_{\textnormal{laG}}$ (using the same equalities as in Lemma~\ref{lem:into_nnf}, we can push the negation towards variables using the axioms of lattice-ordered Abelian groups).

 For the other direction, assume  $\mathcal{A}_{\textnormal{Riesz}}\vdash A \geq 0$. Then, by the completeness theorem, the hypersequent $\vdash A$ has a \hr\ derivation. Then, by the CAN-elimination and the rational T-elimination theorems, $\vdash A$ has a CAN-free and T-free derivation. This is essentially (the trivial translation details are omitted) translatable to a \gasystem\ derivation of $\vdash A$. Since the system \gasystem\ is sound and complete with respect to $\mathcal{A}_{\textnormal{laG}}$ we deduce that  $\mathcal{A}_{\textnormal{laG}}\vdash A\geq 0$ as desired.
\end{proof}

Similarly, we could define the theory of Riesz spaces over rationals ($\mathcal{A}_{\mathbb{Q}\textnormal{-Riesz}}$), defined just as Riesz spaces but over the field $\mathbb{Q}$ of rational numbers instead of the field $\mathbb{R}$ of reals. Again, from~\ref{thm:modal_free_conservativity}, we get the following conservativity result.

\begin{prop}
Let $A$ be a term in the signature of Riesz spaces over rationals. Then 
$$
 \mathcal{A}_{\textnormal{Riesz}}\vdash A \geq 0  \Leftrightarrow \mathcal{A}_{\mathbb{Q}\textnormal{-Riesz}} \vdash A\geq 0.
$$

\end{prop}

Both conservativity results are known as folklore in the theory of Riesz spaces. It is perhaps interesting, however, that here we obtain them in a completely syntactical (proof theoretic) way.

Compared to the proof method used in~\cite{MOGbook, MOG2005} to prove the CAN-elimination theorem, our approach is novel in that our proof is based on the M-elimination theorem. We remark here that a proof of all the theorems stated in this section could have been obtained without using the M-elimination theorem, and instead following the proof structure adopted in~\cite{MOGbook, MOG2005}. The proof technique based on the M-elimination theorem will be however of great value in proving the CAN elimination of the system \hmr{} in Section~\ref{modal_section}.


\subsection{Some technical lemmas}
\label{subsec:modal_free_tech_lemmas}

Before embarking in the proofs of the theorems stated in Section~\ref{modal_free_section}, we prove in this section a few useful routine lemmas that will be used often.

Our first lemma states that the following variant of the ID rule (see Figure~\ref{rules:hr}) where general terms $A$ are considered rather than just variables, is admissible in the proof system \hr. 
$$
 \infer[\text{ID},\sum \vec{r} = \sum \vec{s}]{G \sep \vdash \Gamma, \vec{r}.A,\vec{s}.\covar{A}}{G \sep \vdash \Gamma}
$$
Formally, we prove the admissibility of a slightly more general rule which can act on several sequents of the hypersequent at the same time.

\begin{lem}
\label{lem:modal_free_ext_ID_rule}
For all terms $A$, numbers $n>0$, and vectors $\vec{r_i}$ and $\vec{s}_i$, for $1\leq i \leq n$, such that   $\sum \vec{r_i} = \sum \vec{s_i}$,
\begin{center}
if $\prove \left [ \vdash \Gamma_i \right ] _{i = 1}^n$ then $\prove \left[ \vdash \Gamma_i, \vec{r_i}.A,\vec{s_i}.\negF{A}\right ] _{i = 1}^n$ 
\end{center}
\end{lem}
\begin{proof}
We prove the result by induction on $A$.
  \begin{itemize}
  \item If $A$ is a variable, we simply use the ID rule $n$ times.
  \item If $A= 0$, we use the $0$ rule $n$ times.
  \item If $A = sB$, we use the $\times$ rule $2n$ times and conclude with the induction hypothesis.
  \item If $A = B + C$, we use the $+$ rule $2n$ times and conclude with the induction hypothesis.
  \item For the case $A = B \sqcap C$ or $A = B \sqcup C$, we first use the $\sqcap$ rule $2^n -1$ times --- one time on the conclusion, then again on the two premises, then on the four premises and so forth until we used the $\sqcap$-rule for all sequents --- and then the $\sqcup$ rule $n$ times on each premise and the W rule $n$ times on each premise to remove the sequents with both $B$ and $C$ in them. We can then conclude with the induction hypothesis.
    
    \[ \scalebox{0.85}{
        \infer[\sqcap]{[\vdash \Gamma_i, \vec{r}_i.(B \sqcap C), \vec{s}_i.(\negF{B} \sqcup \negF{C})]_{i=1}^n}{\ddots & \infer{\vdots}{\infer[\sqcup^n]{[\vdash \Gamma_i, \vec{r}_i.B, \vec{s}_i.(\negF{B} \sqcup \negF{C})]_{i=0}^k \sep [\vdash \Gamma_i, \vec{r}_i.C, \vec{s}_i.(\negF{B} \sqcup \negF{C})]_{i=k+1}^n}{\infer[\text{W}^n]{[\vdash \Gamma_i, \vec{r}_i.B, \vec{s}_i.\negF{B}]_{i=0}^k \sep [\vdash \Gamma_i, \vec{r}_i.B, \vec{s}_i.\negF{C}]_{i=0}^k \sep [\vdash \Gamma_i, \vec{r}_i.C, \vec{s}_i.\negF{B}]_{i=k+1}^n \sep [\vdash \Gamma_i, \vec{r}_i.C, \vec{s}_i.\negF{C}]_{i=k+1}^n}{\infer[\text{IH}^2]{[\vdash \Gamma_i, \vec{r}_i.B, \vec{s}_i.\negF{B}]_{i=0}^k \sep [\vdash \Gamma_i, \vec{r}_i.C, \vec{s}_i.\negF{C}]_{i=k+1}^n}{[\vdash \Gamma_i]_{i=0}^n}} }} & \iddots}
      }\]
  \end{itemize}

  Note that the premises obtained after applying the $\sqcap$-rule can have a different shape than the displayed premise in the derivation above, where $B$ and $C$ were chosen. Indeed, the general shape of the premise can be any combination of $B$ and $C$ appearing in the sequents.
\end{proof}

The next result states that derivability in the \hr\ system is preserved by substitution of terms for variables.

\begin{lem}
  \label{lem:modal_free_subst}
  For all hypersequents $G$ and terms $A$, if $\prove G$ then $\prove G[A/x]$.
\end{lem}
\begin{proof}
  We prove the result by induction on the derivation of $G$. Most cases are quite straightforward, we simply use the induction hypothesis on the premises and then use the same rule. For instance, if the derivation finishes with
  \[ \infer[+]{G \sep \vdash \Gamma, \vec{r}.(B + C)}{G \sep \vdash \Gamma, \vec{r}.B, \vec{r}.C}\]
  by induction hypothesis $\prove G[A/x] \sep \vdash \Gamma[A/x], \vec{r}.B[A/x], \vec{r}.C[A/x]$
  so
  \[ \infer[+]{G[A/x] \sep \vdash \Gamma[A/x], \vec{r}.(B + C)[A/x]}{G[A/x] \sep \vdash \Gamma[A/x], \vec{r}.B[A/x], \vec{r}.C[A/x]}\]
  The only tricky case is when the ID rule is used on the variable $x$, where we conclude using Lemma~\ref{lem:modal_free_ext_ID_rule}.
\end{proof}

The next lemma states that the logical rules are invertible using the CAN rule, meaning that if the conclusion is derivable, then the premises are also derivable. The difference with Theorem~\ref{thm:modal_free_invertibility} is that the derivations of the premises introduce a CAN rule.

\begin{lem}
  \label{lem:modal_free_can_full_invertibility}
  All logical rules are invertible.
\end{lem}
\begin{proof}
  We simply use the CAN rule to introduce the operators. We will show the two most interesting cases, the other cases are trivial.
  \begin{itemize}
  \item The $\sqcap$ rule: we assume that $G \sep \vdash \Gamma, \vec{r}.(A \sqcap B)$ is derivable. The derivation of $G \sep \vdash \Gamma, \vec{r}.A$ is then:
    \[ \infer[\text{CAN}]{G \sep \vdash \Gamma, \vec{r}.A}{\infer[\text{M}]{G \sep \vdash \Gamma, \vec{r}.A, \vec{r}.(A \sqcap B), \vec{r}.(\negF{A} \sqcup \negF{B})}{G \sep \vdash \Gamma, \vec{r}.(A \sqcap B) & \infer[\sqcup]{\vdash \vec{r}.A, \vec{r}.(\negF{A} \sqcup \negF{B})}{\infer[\text{W}]{\vdash \vec{r}.A, \vec{r}.\negF{A} \sep \vdash \vec{r}.A, \vec{r}.\negF{B}}{\infer[\text{Lemma~\ref{lem:modal_free_ext_ID_rule}}]{\vdash \vec{r}.A, \vec{r}.\negF{A}}{\infer[\text{INIT}]{\vdash}{}}}}}} \]
    The derivation of $G \sep \vdash \Gamma, \vec{r}.B$ is similar.
  \item The $\sqcup$ rule: we assume that $G \sep \vdash \Gamma, \vec{r}.(A \sqcup B)$ is derivable. The derivation of $G \sep \vdash \Gamma, \vec{r}.A \sep \vdash \Gamma, \vec{r}.B$ is then:
    \[ \scalebox{0.85}{
        \infer[\text{CAN}]{G \sep \vdash \Gamma, \vec{r}.A \sep \vdash \Gamma, \vec{r}.B}{\infer[\text{M}]{G \sep \vdash \Gamma, \vec{r}.A, \vec{r}.(A \sqcup B), \vec{r}.(\negF{A} \sqcap \negF{B}) \sep \vdash \Gamma, \vec{r}.B}{\infer[\text{W}]{G \sep \vdash \Gamma, \vec{r}.(A \sqcup B) \sep \vdash \Gamma, \vec{r}.B}{G \sep \vdash \Gamma, \vec{r}.(A \sqcup B) } & \infer[\sqcap]{G \sep \vdash \vec{r}.A,\vec{r}.(\negF{A} \sqcap \negF{B}) \sep \vdash \Gamma, \vec{r}.B}{\infer[\text{W}^*]{G \sep \vdash \vec{r}.A,\vec{r}.\negF{A}\sep \vdash \Gamma, \vec{r}.B}{\infer[\text{Lemma~\ref{lem:modal_free_ext_ID_rule}}]{\vdash \vec{r}.A,\vec{r}.\negF{A}}{\infer[\text{INIT}]{\vdash}{}}} & \infer{G \sep \vdash \vec{r}.A,\vec{r}.\negF{B}\sep \vdash \Gamma, \vec{r}.B}{\Pi}}}} }\]
where $\Pi$ is the following derivation:
\[  \scalebox{0.85}{
    \infer[\text{CAN}]{G \sep \vdash \vec{r}.A,\vec{r}.\negF{B}\sep \vdash \Gamma, \vec{r}.B}{\infer[\text{M}]{G \sep \vdash \vec{r}.A,\vec{r}.\negF{B}\sep \vdash \Gamma, \vec{r}.B, \vec{r}.(A \sqcup B),\vec{r}.(\negF{A} \sqcap \negF{B})}{\infer[\text{W}]{G \sep \vdash \vec{r}.A,\vec{r}.\negF{B}\sep \vdash \Gamma,\vec{r}.(A \sqcup B)}{G \sep \vdash \Gamma,\vec{r}.(A \sqcup B)} & \infer[\text{W}^*]{G \sep \vdash \vec{r}.A,\vec{r}.\negF{B}\sep \vdash \vec{r}.B, \vec{r}.(\negF{A} \sqcap \negF{B})}{\infer[\sqcap]{ \vdash \vec{r}.A,\vec{r}.\negF{B}\sep \vdash \vec{r}.B, \vec{r}.(\negF{A} \sqcap \negF{B})}{\infer[\text{S}]{ \vdash \vec{r}.A,\vec{r}.\negF{B}\sep \vdash\vec{r}.B, \vec{r}.\negF{A}}{\infer[\text{Lemma~\ref{lem:modal_free_ext_ID_rule}}]{\vdash \vec{r}.A,\vec{r}.\negF{B},\vec{r}.B, \vec{r}.\negF{A}}{\infer[\text{Lemma~\ref{lem:modal_free_ext_ID_rule}}]{\vdash \vec{r}.B,\vec{r}.\negF{B}}{\infer[\text{INIT}]{\vdash}{}}}} & \infer[\text{W}]{\vdash \vec{r}.A,\vec{r}.\negF{B}\sep \vdash\vec{r}.B, \vec{r}.\negF{B}}{\infer[\text{Lemma~\ref{lem:modal_free_ext_ID_rule}}]{\vdash \vec{r}.B,\vec{r}.\negF{B}}{\infer[\text{INIT}]{\vdash}{}}}}}}}
  }\qedhere\]
  \end{itemize}
\end{proof}

\begin{rem}
  \label{rem:modal_free_inv_no_T_rule}
  The proof of invertibility does not introduce any new T rule, so if the conclusion of a logical rule has a T-free derivation then the premises also have T-free derivations.
\end{rem}
      
The next lemmas state that CAN-free derivability in the \hr\ system is preserved  by scalar multiplication.      
      
\begin{lem}
  \label{lem:modal_free_genT}
  Let $\vec{r}\in\mathbb{R}_{>0}$ be a non-empty vector and $G$ a hypersequent. If $\proveNC G \sep \vdash \vec{r}.\Gamma$ then $\proveNC G \sep \vdash \Gamma$.
\end{lem}
\begin{proof}
  We simply use the C,T and S rules :
  \[ \infer[\text{C}^*]{G \sep \vdash \Gamma}{\infer[\text{T}^*]{G \sep \vdash \Gamma \sep ... \sep \vdash \Gamma}{\infer[\text{S}^*]{G \sep \vdash r_1.\Gamma \sep ... \sep r_n.\Gamma}{G \sep \vdash \vec{r}.\Gamma}}}\qedhere \]
\end{proof}

\begin{lem}
  \label{lem:modal_free_genT2}
  Let $\vec{r}\in\mathbb{R}_{>0}$ be a vector and $G$ a hypersequent. If $\proveNC G \sep \vdash \Gamma$ then $\proveNC G \sep \vdash \vec{r}.\Gamma$.
\end{lem}
\begin{proof}
  We reason by induction on the size of $\vec{r}$.
  
  If the size of $\vec{r}$ is 0: Since $\vdash \vec{r}.\Gamma = \vdash$, we simply use the W rule until we can use the INIT rule:
  \[ \infer[\text{W}^*]{G \sep \vdash}{\infer[\text{INIT}]{\vdash}{}} \]
  
  If the size of $\vec{r}$ is 1: we can use the T rule:
\[ \infer[\text{T}]{G \sep \vdash r_1.\Gamma}{G \sep \vdash  (\frac{1}{r_1}r_1).\Gamma} \]

  Otherwise: Let $(r_1,...,r_{n+1}) = \vec{r}$. We can invoke the inductive hypothesis and conclude as follows: 
  \[ \infer[\text{M}]{G \sep \vdash r_1.\Gamma, ...,r_n.\Gamma,r_{n+1}.\Gamma}{\infer{G \sep \vdash r_1.\Gamma, ...,r_n.\Gamma}{G \sep \vdash \Gamma} & \infer[\text{T}]{G \sep \vdash r_{n+1}.\Gamma}{G \sep \vdash  (\frac{1}{r_{n+1}}r_{n+1}).\Gamma}}\qedhere \]
\end{proof}

The above lemmas have two useful corollaries.

\begin{cor}
  \label{cor:modal_free_andAlphaBeta}
    If $\proveNC G \sep \vdash \Gamma,\vec{r}.A, \vec{s}.A$ and $\proveNC G \sep \vdash \Gamma, \vec{r}.B, \vec{s}.B$ then $\proveNC G \sep \vdash \Gamma, \vec{r}.A, \vec{s}.B$.
\end{cor}
\begin{proof}
  If $\vec{r} = \emptyset$ or $\vec{s} = \emptyset$, the result is trivial. Otherwise
  \[
  \infer[\text{C}]{G \sep \vdash \Gamma, \vec{r}.A, \vec{s}.B}
        {\infer[\text{Lemma~\ref{lem:modal_free_genT}}]{G \sep \vdash \Gamma, \vec{r}.A, \vec{s}.B \sep \vdash \Gamma, \vec{r}.A, \vec{s}.B}{\infer[\text{Lemma~\ref{lem:modal_free_genT}}]{G \sep \vdash \vec{r}.\Gamma, (\vec{r}\vec{r}).A, (\vec{r}\vec{s}).B \sep \vdash \Gamma, \vec{r}.A, \vec{s}.B}{\infer[\text{S}]{G \sep \vdash \vec{r}.\Gamma, (\vec{r}\vec{r}).A,(\vec{r}\vec{s}).B \sep \vec{s}.\Gamma,(\vec{s}\vec{r}).A, (\vec{s}\vec{s}).B}{\infer[\text{M}]{G \sep \vdash \vec{r}.\Gamma,\vec{s}.\Gamma, (\vec{r}\vec{r}).A,(\vec{s}\vec{r}).A,(\vec{r}\vec{s}).B, (\vec{s}\vec{s}).B}
          {\infer[\text{Lemma }\ref{lem:modal_free_genT2}]{G \sep \vdash \vec{r}.\Gamma, (\vec{r}\vec{r}).A, (\vec{r}\vec{s}).A}
            {G \sep \vdash \Gamma , \vec{r}.A,\vec{s}.A}
            &
            \infer[\text{Lemma }\ref{lem:modal_free_genT2}]{G \sep \vdash \vec{s}.\Gamma, (\vec{s}\vec{r}).B, (\vec{s}\vec{s}).B} {G \sep \vdash \Gamma , \vec{r}.B,\vec{s}.B}}}}}}\qedhere
  \]
\end{proof}

\begin{cor}
  \label{cor:modal_free_orAlphaBeta}
If $\proveNC G \sep \vdash \vec{r}.A,\vec{s}.A, \Gamma \sep \vdash \vec{r}.B,\vec{s}.B , \Gamma \sep \vdash \vec{r}.A,\vec{s}.B,\Gamma$, then $\proveNC G \sep \vdash \vec{r}.A,\vec{s}.A, \Gamma \sep \vdash \vec{r}.B,\vec{s}.B, \Gamma$.
\end{cor}
\begin{proof}
  If $\vec{r} = \emptyset$ or $\vec{s} = \emptyset$, the result is trivial. Otherwise
  \[
  \infer[\text{C}^2]{G \sep \vdash \Gamma, \vec{r}.A,\vec{s}.A \sep \vdash \Gamma, \vec{r}.B,\vec{s}.B}{\infer[\text{Lemma~\ref{lem:modal_free_genT}}]{G \sep \vdash \Gamma, \vec{r}.A,\vec{s}.A \sep \vdash \Gamma, \vec{r}.B,\vec{s}. B \sep \vdash \Gamma, \vec{r}.A,\vec{s}.A \sep \vdash \Gamma, \vec{r}.B,\vec{s}.B}{\infer[\text{Lemma~\ref{lem:modal_free_genT}}]{G \sep \vdash \Gamma, \vec{r}.A,\vec{s}.A \sep \vdash \Gamma, \vec{r}.B,\vec{s}. B \sep \vdash \vec{r}.\Gamma, (\vec{r}\vec{r}).A,(\vec{r}\vec{s}).A \sep \vdash \Gamma, \vec{r}.B,\vec{s}.B}
        {\infer[\text{S}]{G \sep \vdash \Gamma, \vec{r}.A,\vec{s}.A \sep \vdash \Gamma, \vec{r}.B,\vec{s}. B \sep \vdash \vec{r}.\Gamma, (\vec{r}\vec{r}).A,(\vec{r}\vec{s}).A \sep \vdash \vec{s}.\Gamma, (\vec{s}\vec{r}).B,(\vec{s}\vec{s}).B}{\infer[\text{Lemma }~\ref{lem:modal_free_genT2}]{G \sep \vdash \Gamma, \vec{r}.A,\vec{s}.A  \sep \vdash \Gamma, \vec{r}.B,\vec{s}.B \sep \vdash \vec{r}.\Gamma, \vec{s}.\Gamma, (\vec{r}\vec{r}.)A, (\vec{r}\vec{s}).A , (\vec{s}\vec{r}).B, (\vec{s}\vec{s}).B}
          {G \sep \vdash \Gamma, \vec{r}.A,\vec{s}.A \sep \vdash \Gamma, \vec{r}.B,\vec{s}.B \sep \vdash \Gamma, \vec{r}.A,\vec{s}.B}}}}}\qedhere
\]
\end{proof}


\subsection{Soundness -- Proof of Theorem~\ref{thm:modal_free_soundness}}
\label{subsec:modal_free_soundness}

We need to prove that if there exists a \hr\ derivation of a hypersequent $G$ then $\sem{ G } \geq 0$ is derivable in equational logic (written $\mathcal{A}_{\textnormal{Riesz}}\vdash \sem{G}\geq 0$). This is done in a straightforward way by showing that each deduction rule of the system \hr\ is sound. The desired result then follows immediately by induction on the derivation of $G$.

  \begin{itemize}
  \item For the rule
    \[ \infer[\text{INIT}]{\vdash}{} \]
    The semantics of the hypersequent consisting only of the empty sequent is $\sem{ \vdash } = 0$ and therefore $\sem{ \vdash } \geq 0$, as desired.
  \item For the rule
    \[ \infer[\text{W}]{G \sep \vdash \Gamma}{G} \]
    the hypothesis is  $\sem{ G } \geq 0$ so
    \begin{eqnarray*}
      \sem{ G \sep \vdash \Gamma } &=& \sem{ G } \sqcup \sem{ \vdash \Gamma } \\
      &\geq& \sem{ G }\\
      &\geq& 0
    \end{eqnarray*}
  \item For the $\text{C}, \text{ID},+,0,\times$ and $\text{CAN}$ rules, it is immediate to observe that the interpretation of the only premise and the interpretation of its conclusion are equal, therefore the result is trivial.
  \item For the rule
    \[ \infer[\text{S}]{G \sep \vdash \Gamma_1 \sep \vdash \Gamma_2}{G \sep \vdash \Gamma_1,\Gamma_2} \]
    the hypothesis is  $\sem{ G \sep \vdash \Gamma_1,\Gamma_2 } \geq 0$ so according to Lemma~\ref{lem:condMaxPos}, $\sem{ G }^- \sqcap \sem{ \vdash \Gamma_1,\Gamma_2} ^- = 0$. Our goal is to prove that $\sem{G \sep \vdash \Gamma_1 \sep \vdash \Gamma_2} \geq 0$. Again, using Lemma~\ref{lem:condMaxPos}, we equivalently need to prove that 
    
    $$\sem{ G }^- \sqcap \sem{ \vdash \Gamma_1 \sep \vdash \Gamma_2 }^-  = 0.$$

    The above expression is of the form $A^{-}\sqcap B^{-}$, and since $A^{-}\geq 0$ always holds for every $A$ (see Section~\ref{some_small_lemmas_section}), it is clear that   $\sem{ G }^- \sqcap \sem{ \vdash \Gamma_1 \sep \vdash \Gamma_2 }^- \geq 0$. It remains therefore to show that $\sem{ G }^- \sqcap \sem{ \vdash \Gamma_1 \sep \vdash \Gamma_2 }^- \leq 0$. This is done as follows:

    \[
    \begin{array}{lllr}
      \sem{ G }^- \sqcap \sem{ \vdash \Gamma_1 \sep \vdash \Gamma_2 }^- & \leq & \sem{ G }^- \sqcap 2.\sem{ \vdash \Gamma_1 \sep  \vdash \Gamma_2 }^- & \text{since $\sem{ \vdash \Gamma_1 \sep  \vdash \Gamma_2 }^- \geq 0$} \\
      & = & \sem{ G }^- \sqcap (2.(\sem{ \vdash \Gamma_1} \sqcup \sem{ \vdash \Gamma_2 }))^- & \text{Lemma~\ref{lem:equalities}[1]} \\
      & \leq & \sem{ G }^- \sqcap (\sem{  \vdash \Gamma_1 } + \sem{ \vdash \Gamma_2 })^- & \text{Lemma~\ref{lem:equalities}[2-3]}\\
      & = &  \sem{ G }^- \sqcap (\sem{ \vdash \Gamma_1, \Gamma_2 })^- & \\
      & = & 0
    \end{array}
    \]
  \item For the rule
    \[ \infer[\text{M}]{G \sep \vdash \Gamma_1, \Gamma_2}{G \sep \vdash \Gamma_1 & G \sep \vdash \Gamma_2} \]
    the hypothesis is 
    \[ \sem{ G \sep \vdash \Gamma_1 } \geq 0 \]
    \[ \sem{ G \sep \vdash \Gamma_2 } \geq 0 \]
    so according to Lemma~\ref{lem:condMaxPos},
    \[ \sem{ G }^- \sqcap \sem{ \vdash \Gamma_1 }^- = 0 \]
    \[ \sem{ G }^- \sqcap \sem{ \vdash \Gamma_2 }^- = 0 \]
    Following the same reasoning of the previous case (S rule) our goal is to show that $\sem{ G }^- \sqcap \sem{ \vdash \Gamma_1, \Gamma_2 }^- \leq 0$. This is done as follows:
    \[
    \begin{array}{lllr}
      \sem{ G }^- \sqcap \sem{ \vdash \Gamma_1, \Gamma_2 }^- & = & \sem{ G }^- \sqcap (\sem{ \vdash \Gamma_1} + \sem{ \vdash \Gamma_2 })^- &  \\
      & \leq & \sem{ G }^- \sqcap (\sem{ \vdash \Gamma_1}^- + \sem{ \vdash \Gamma_2 }^-)& \text{Lemma~\ref{lem:equalities}[4]} \\
      & \leq & \sem{ G }^- \sqcap \sem{ \vdash \Gamma_1}^- + \sem{ G }^- \sqcap \sem{ \vdash \Gamma_2 }^- & \text{distributivity of $\sqcap$ over $+$}
    \end{array}
    \]

  \item For the rule
    \[ \infer[\text{T}]{G \sep \vdash \Gamma}{G \sep \vdash r.\Gamma} \]
    the hypothesis is  $\sem{ G \sep \vdash r.\Gamma } \geq 0$ so using Lemma~\ref{lem:condMaxPos}, we have
    \[ \sem{ G }^- \sqcap r.(\sem{ \vdash \Gamma } )^-= \sem{ G }^- \sqcap (\sem{ \vdash r.\Gamma } )^- = 0 \]
    By the same reasoning as for the S rule's case, our goal is to show that $\sem{ G }^- \sqcap \sem{ \vdash \Gamma } ^- \leq 0$. To do so, we need to distinguish between two cases: whether or not $r \geq 1$.
    
    If $r \geq 1$, then
    \[
      \begin{array}{lll}
        \sem{ G }^- \sqcap \sem{ \vdash \Gamma } ^- & \leq & \sem{ G }^- \sqcap r.\sem{ \vdash \Gamma } ^- \\
                                                                                & = & 0
      \end{array}
    \]

    Otherwise, Lemma~\ref{lem:equalities}[5] states that $\sem{ G }^- \sqcap \sem{ \vdash \Gamma } ^- \leq 0$ if and only if $r.(\sem{ G }^- \sqcap \sem{ \vdash \Gamma } ^-) \leq 0$, which is proven as follows:
    \[
      \begin{array}{lll}
        r.(\sem{ G }^- \sqcap \sem{ \vdash \Gamma } ^-) & = & (r.\sem{ G }^-) \sqcap (r.\sem{ \vdash \Gamma } ^-) \\
        & \leq & \sem{ G }^- \sqcap (r.\sem{ \vdash \Gamma } ^-) \\
                                                                                    & = & 0
      \end{array}
    \]

    In both cases $\sem{ G }^- \sqcap \sem{ \vdash \Gamma } ^- \leq 0$.
  \item For the rule
  \[ \infer[\sqcup]{G \sep \vdash \Gamma, \vec{r}.(A\sqcup B) }{G \sep \vdash \Gamma, \vec{r}. A  \sep \vdash \Gamma, \vec{r}. B } \]
  the hypothesis is $\sem{ G \sep \vdash \Gamma, \vec{r}. A  \sep \vdash \Gamma, \vec{r}. B } \geq 0$.
  So :
  \[
  \begin{array}{lllr}
  	\sem{ G \sep \vdash \Gamma, \vec{r}.(A\sqcup B) } & = &  \sem{ G } \sqcup \sem{ \vdash \Gamma, \vec{r}.(A\sqcup B) } & \\
  	& = & \sem{ G } \sqcup \sem{ \vdash \Gamma, \vec{r}.A } \sqcup \sem{ \vdash \Gamma, \vec{r}.B } & \text{distributivity of $\sqcup$ over $+$} \\
  	& \geq & 0 &
  \end{array}
  \]
  \item For the rule
  \[ \infer[\sqcap]{G \sep \vdash \Gamma, \vec{r}.(A\sqcap B)}{G \sep \vdash \Gamma, \vec{r}.A  & G \sep \vdash \Gamma, \vec{r}.B} \]
  the hypothesis is
  \[ \sem{ G \sep \Gamma, \vec{r}.A \vdash \Gamma } \geq 0\]
  \[ \sem{ G \sep \Gamma, \vec{r}.B \vdash \Gamma } \geq 0\]
  So
  \[
    \begin{array}{lllr}
      \sem{ G \sep \vdash \Gamma, \vec{r}.(A \sqcap B) } & = & \sem{ G } \sqcup \sem{  \vdash \Gamma, \vec{r}.(A \sqcap B) } & \\
      & = & \sem{ G } \sqcup (\sem{ \vdash \Gamma, \vec{r}.A } \sqcap \sem{ \Gamma, \vec{r}.B }) & \text{$\sqcap$ distributes over $+$} \\
      & = & (\sem{ G } \sqcup \sem{ \vdash \Gamma, \vec{r}.A }) \sqcap (\sem{ G } \sqcup \sem{ \vdash \Gamma, \vec{r}.B })&\text{$\sqcup$ distributes over $\sqcap$} \\
     & \geq & 0&
    \end{array}
  \]
\end{itemize}


\subsection{Completeness -- Proof of Theorem~\ref{thm:modal_free_completeness}}
\label{subsec:modal_free_completeness}

In order to prove  Theorem~\ref{thm:modal_free_completeness} we first prove an equivalent result (Lemma~\ref{lem:modal_free_completeness_aux} below) stating that if $\semProve A = B$ then the hypersequents $\vdash r.A, r.\negF{B}$ and $\vdash r.B, r.\negF{A}$ are both derivable for all $r > 0$. The advantage of this formulation is that it allows for a simpler proof by induction. 

From Lemma~\ref{lem:modal_free_completeness_aux} one indeed obtains Theorem~\ref{thm:modal_free_completeness} as a corollary.

\begin{proof}[Proof of Theorem~\ref{thm:modal_free_completeness}.]
Recall that $\semProve \sem{G}\geq 0$ is a shorthand for 
$\semProve 0 = \sem{G} \sqcap 0$. 
Hence, from the hypothesis $\semProve \sem{G}\geq 0$  we can deduce, by using Lemma~\ref{lem:modal_free_completeness_aux}, that $\prove \vdash 1.(0 \sqcap \sem{G}), 1.0$ is provable.

From this we can show that $\prove G$ by invoking Lemma~\ref{lem:modal_free_can_full_invertibility}. Indeed, if $G$ is $  \vdash \Gamma_1 \sep ... \sep \vdash \Gamma_n$ then $\sem{G} = \sem{\vdash \Gamma_1} \sqcup ... \sqcup \sem{\vdash \Gamma_n}$ and
\begin{enumerate}
\item by using the invertibility of the $0$ rule, $\vdash 1.(0 \sqcap (\sem{\vdash \Gamma_1} \sqcup ... \sqcup \sem{\vdash \Gamma_n}))$ is derivable,
\item by using the invertibility of the $\sqcap$ rule, $\vdash 1.(\sem{\vdash \Gamma_1} \sqcup ... \sqcup \sem{\vdash \Gamma_n})$ is derivable,
\item by using the invertibility of the $\sqcup$ rule $n-1$ times, $\vdash 1.\sem{\vdash \Gamma_1} \sep ... \sep \vdash 1.\sem{\vdash \Gamma_n}$ is derivable,
\item and finally, by using the invertibility of the $+$ rule and $\times$ rule, $\vdash \Gamma_1 \sep ... \sep \vdash \Gamma_n$ is derivable.\qedhere
\end{enumerate}
\end{proof}

\begin{lem}
  \label{lem:modal_free_completeness_aux}
  If $\semProve A = B$ then $\vdash r.A, r.\negF{B}$ and $\vdash r.B, r.\negF{A}$ are provable for all $r > 0$.
\end{lem}
\begin{proof}
We prove this result by induction on the derivation, in equational logic (see Definition~\ref{eq_logic_rules}) of $\semProve A = B$.
  \begin{itemize}
  \item If the derivation finishes with 
    \[ \infer[\text{refl}]{\semProve A = A}{} \]
    we can conclude with Lemma~\ref{lem:modal_free_ext_ID_rule}.
  \item If the derivation finishes with 
    \[ \infer[\text{sym}]{\semProve A = B}{\semProve B = A} \]
    then the induction hypothesis allows us to conclude.
  \item If the derivation finishes with
    \[ \infer[\text{trans}]{\semProve A = B}{\semProve A = C & \semProve C =B} \]
    then the induction hypothesis is
    \[ \prove \vdash r.A, r.\negF{C} \]
    \[ \prove \vdash r.C, r.\negF{A} \]
    \[ \prove \vdash r'.C, r'.\negF{B} \]
    \[ \prove \vdash r'.B, r'.\negF{C} \]
    for all $r,r' >0$.
    We will show that $\prove\vdash  r.A, r.\negF{B}$ for all $r$, the other one is similar.
    \[ \infer[\text{CAN}]{\vdash r.A, r.\negF{B}}
      {\infer[\text{M}]{\vdash r.A, r.\negF{B},r.C,r.\negF{C}}{\vdash r.A, r.\negF{C} & \vdash r.C, r.\negF{B}}} \]
  \item If the derivation finishes with
    \[ \infer[\text{subst}]{\semProve A[C/x] = B[C/x]}{\semProve A = B} \]
    we conclude using the induction hypothesis and Lemma~\ref{lem:modal_free_subst}.
  \item If the derivation finishes with
    \[ \infer[\text{ctxt}]{\semProve C[A] = C[B]}{\semProve A = B} \]
    we prove the result by induction on C.
    For instance, if $C = sC'$ with $s > 0$, then the induction hypothesis is $\prove \vdash r.C'[A], r.\negF{C'[B]}$ and $\prove \vdash r.C'[B], r.\negF{C'[A]}$ for all $r > 0$ so
    \begin{multicols}{2}
      \setlength{\columnseprule}{0pt}
      \[ \infer[\times^*]{\vdash r.C[A], r.\negF{C[B]}}{\vdash rs.C'[A], rs.\negF{C'[B]}} \]
      
      \[ \infer[\times^*]{\vdash r.C[B], r.\negF{C[A]}}{\vdash rs.C'[B], rs.\negF{C'[A]}} \]
    \end{multicols}

\item It now remains to consider the cases when the derivation finishes with one of the axioms of Figure~\ref{axioms:of:riesz:spaces}. We only show the nontrivial cases.
  \begin{itemize}
  \item If the derivation finishes with
    \[ \infer[\text{ax}]{\semProve (r_1+r_2)x = r_1x + r_2x}{} \]
    then
    \[ \infer[+]{\vdash r.((r_1+r_2)x), r.(r_1\covar{x}+r_2\covar{x})}{\infer[\times^*]{\vdash r.((r_1+r_2)x), r.r_1\covar{x}, r.r_2\covar{x}}{\infer[\text{ID}]{\vdash (r_1+r_2)r.x, r_1r.\covar{x}, r_2r.\covar{x}}{\infer[\text{INIT}]{\vdash}{}}}} \]
    and
    \[ \infer[+]{\vdash r.(r_1x + r_2x), r.((r_1+r_2)\covar{x})}{\infer[\times^*]{\vdash r.r_1x, r.r_2x, r.((r_1+r_2)\covar{x})}{\infer[\text{ID}]{\vdash r_1r.x, r_2r.x, (r_1+r_2)r.\covar{x}}{\infer[\text{INIT}]{\vdash}{}}}} \]
  \item If the derivation finishes with
    \[ \infer[\text{ax}]{\semProve (s(x\sqcap y)) \sqcap sy = s(x\sqcap y)}{} \]
    then
    \[ \infer[\sqcap]{\vdash r .((s(x\sqcap y)) \sqcap sy), r.s(\covar{x} \sqcup \covar{y})}{\infer[\text{Lemma~\ref{lem:modal_free_ext_ID_rule}}]{\vdash r.(s(x\sqcap y)),r.s(\covar{x} \sqcup \covar{y})}{\infer[\text{INIT}]{\vdash}{}} & \infer[\times^*]{\vdash r.sy,r.s(\covar{x} \sqcup \covar{y})}{\infer[\sqcup]{\vdash rs.y, rs.(\covar{x} \sqcup \covar{y})}{\infer[\text{W}]{\vdash rs.y, rs.\covar{x} \sep \vdash rs.y, rs.\covar{y}}{\infer[\text{ID}]{\vdash rs.y, rs.\covar{y}}{\infer[\text{INIT}]{\vdash}{}}}}}} \]
    and
    \[ \infer[\sqcup - \text{W}]{\vdash r.(s(x \sqcap y)), r.((s(\covar{x}\sqcup \covar{y}))\sqcup(s\covar{y}))}{\infer[\text{Lemma~\ref{lem:modal_free_ext_ID_rule}}]{\vdash r.(s(x \sqcap y)), r.(s(\covar{x}\sqcup \covar{y}))}{\infer[\text{INIT}]{\vdash}{}}}\qedhere\]
\end{itemize}
\end{itemize}
\end{proof}

\begin{rem}
By inspecting the proof of Lemma~\ref{lem:modal_free_completeness_aux} it is possible to verify that the T rule is never used in the construction of $\prove G$. This, together with the similar Remark~\ref{rem:modal_free_inv_no_T_rule} regarding Lemma~\ref{lem:modal_free_can_full_invertibility}, implies that the T rule is never used in the proof of the completeness Theorem~\ref{thm:modal_free_completeness}. From this we get the following corollary.
\end{rem}

\begin{cor}
The T rule is admissible in the system \hr.
\end{cor}

It turns out, however, that there is no hope of eliminating both the T rule and the CAN rule from the \hr\ system.

\begin{lem}
  \label{lem:modal_free_not_complete}
Let $r_1$ and $r_2$ be two irrational numbers that are incommensurable (so there is no $q\in\mathbb{Q}$ such that $qr_1 = r_2$). Then the atomic hypersequent $G$
$$\vdash r_1.x \sep \vdash r_2.\covar{x}$$ 
does not have a CAN-free and T-free derivation.
\end{lem}
\begin{proof}
This is a corollary of the next Lemma~\ref{lem:modal_free_int_lambda_prop}. The idea is that in the \hr\ system without the T rule and the CAN rule, the only way to derive $G$ is by applying the structural rules S, C, W, M and the ID rule. Each of these rules can be seen as adding up the sequents in $G$ or multiplying them up by a positive natural number scalar. Since $r_1$ and $r_2$ are incommensurable, it is not possible to construct a derivation.
\end{proof}

\begin{lem}
  \label{lem:modal_free_int_lambda_prop}
  For all atomic hypersequents $G$, built using the variables and negated variables $x_1, \covar{x_1}, \dots, x_{k}, \covar{x_{k}}$, of the form  
  $$\vdash \Gamma_1 \sep \dots \sep \vdash \Gamma_m$$
where $\Gamma_i =  \vec{r}_{i,1}.x_1,...,\vec{r}_{i,k}.x_k, \vec{s}_{i,1}.\covar{x_1},...,\vec{s}_{i,k}.\covar{x_{i,k}}$, the following are equivalent:
  
  \begin{enumerate}
  \item $G$ has a CAN-free and T-free derivation.
  \item there exist natural numbers $n_1,...,n_m \in \mathbb{N}$, one for each sequent  in $G$, such that:
  \begin{itemize}
  \item there exists $i\in [1..m]$ such that $n_i \neq 0$, i.e., the numbers are not all $0$'s, and
  \item for every variable and covariable $(x_j, \covar{x_j})$ pair, it holds that
  $$
  \sum^{m}_{i=1} n_i (\sum \vec{r}_{i,j}) =   \sum^{m}_{i=1} n_i (\sum \vec{s}_{i,j}) 
  $$
i.e., the scaled (by the numbers $n_1$ \dots $n_m$) sum of the coefficients in front of the variable $x_j$ is equal to the scaled sum of the coefficients in from of the covariable $\covar{x_j}$. 
  \end{itemize}
  \end{enumerate}
\end{lem}
\begin{proof}
  We prove $(1) \Rightarrow (2)$ by induction on the derivation of $G$. We show only the M case, the other cases being trivial:
  \begin{itemize}
  \item If the derivation finishes with
    \[ \infer[\text{M}]{\vdash \Gamma_1 \mid ... \mid \vdash \Gamma_m , \Gamma'_m}{\vdash \Gamma_1 \mid ... \mid \vdash \Gamma_m & \vdash \Gamma_1 \mid ... \mid \vdash \Gamma'_m} \]
    by induction hypothesis, there are $n_1,...,n_m \in \mathbb{N}$ and $n'_1,...,n'_m \in \mathbb{N}$ such that :
    \begin{itemize}
    \item there exists $i\in [1..m]$ such that $n_i \neq 0$.
    \item for every variable and covariable $(x_j,\covar{x_j})$ pair, it holds that $\sum_i n_i.\sum \vec{r}_{i,j} = \sum_i n_i.\sum \vec{s}_{i,j}$.
    \item there exists $i\in [1..m]$ such that $n'_i \neq 0$.
    \item for every variable and covariable $(x_j,\covar{x_j})$ pair, it holds that $$\sum_{i=0}^{m-1} n'_i.\sum \vec{r}_{i,j} + n'_m.\sum \vec{r'}_{m,j}= \sum_{i=0}^{m-1} n'_i.\sum \vec{s}_{i,j} + n'_m.\sum \vec{s'}_{m,j}$$
    \end{itemize}
    If $n_m=0$ then $n_1,...,n_{m-1}, 0$ satisfies the property.\\
    Otherwise if $n'_m = 0$ then $n'_1,...,n'_{m-1},0$ satisfies the property.\\
    Otherwise, $n_m.n'_1 + n'_m.n_1, n_m.n'_2 + n'_m.n_2,...,n_m.n'_{m-1} + n'_m.n_{m-1},n_m.n'_m$ satisfies the property.
  \end{itemize}
  The other way ($(2) \Rightarrow (1)$) is more straightforward. If there exist natural numbers $n_1,...,n_m \in \mathbb{N}$, one for each sequent  in $G$, such that:
  \begin{itemize}
  \item there exists $i\in [1..m]$ such that $n_i \neq 0$ and
  \item for every variable and covariable $(x_j, \covar{x_j})$ pair, it holds that
    $$
    \sum^{m}_{i=1} n_i (\sum \vec{r}_{i,j}) =   \sum^{m}_{i=1} n_i (\sum \vec{s}_{i,j}) 
    $$
  \end{itemize}
  then we can use the W rule to remove the sequents corresponding to the numbers $n_i = 0$, and use the C rule $n_i-1$ times then the S rule $n_i-1$ times on the $i$th sequent to multiply it by $n_i$. If we assume that there is a natural number $l$ such that $n_i = 0$ for all $i > l$ and $n_i \neq 0$ for all $i \leq l$, then the CAN-free T-free derivation is:
  \[ \infer[\text{W}^*]{\vdash \Gamma_1 \sep \dots \sep \Gamma_m}{\infer[\text{C-S}^*]{\vdash \Gamma_1 \sep \dots \sep \vdash \Gamma_l}{\infer[\text{S}^*]{\vdash {\Gamma_1}^{n_1} \sep \dots \sep \vdash {\Gamma_l}^{n_l}}{\infer[\text{ID}^*]{\vdash {\Gamma_1}^{n_1},\dots,{\Gamma_l}^{n_l}}{\infer[\text{INIT}]{\vdash}{}}}}} \]
  where $\Gamma ^n$ stands for $\underbrace{\Gamma,\dots,\Gamma}_{n}$.
\end{proof}


\subsection{CAN-free Invertibility -- Proof of Theorem~\ref{thm:modal_free_invertibility}}
\label{subsec:modal_free_invertibility}
In this section, we go through the details of the proof of Theorem~\ref{thm:modal_free_invertibility}.

It is technically convenient, in order to carry out the inductive argument, to prove a slightly stronger result, expressed as the invertibility of more general logical rules that can act on the same term on different sequents of the hypersequent, at the same time. The generalised rules are the following:

\begin{figure}[h!]
  \begin{center}
  \scalebox{0.85}{
     \begin{minipage}{13.5cm}
        \textbf{Logical rules:}
        \begin{center}
          \begin{tabular}{ccc}
            $\infer[0]
            {\left[\vdash \Gamma_i , \vec{r}_i. 0 \right]_{i=1}^n}
            {\left[\vdash \Gamma_i \right]_{i=1}^n}$
            &
              $\infer[+]
              {\left[\vdash \Gamma_i , \vec{r}_i.(A + B) \right]_{i=1}^n}
              {\left[\vdash \Gamma_i , \vec{r}_i.A , \vec{r}_i.B \right]_{i=1}^n}$
            &
              $\infer[\times]
              {\left[\vdash \Gamma_i, \vec{r}_i.(sA)\right]_{i=1}^n}
              {\left[\vdash \Gamma_i, (s\vec{r}_i).A\right]_{i=1}^n}$\\[0.4cm]
          \end{tabular}
          \begin{tabular}{cc}
            $\infer[\sqcup]
            {\left[\vdash \Gamma_i , \vec{r}_i.(A \sqcup B) \right]_{i=1}^n}
            {\left[\vdash \Gamma_i , \vec{r}_i.A \sep \vdash \Gamma_i , \vec{r}_i.B \right]_{i=1}^n}$
            &
              $\infer[\sqcap]
              {\left[\vdash \Gamma_i , \vec{r}_i.(A \sqcap B) \right]_{i=1}^n}
              {\left[\vdash \Gamma_i , \vec{r}_i.A \right]_{i=1}^n & \left[\vdash \Gamma_i , \vec{r}_i.B \right]_{i=1}^n}$
          \end{tabular}
        \end{center}
      \end{minipage}
    }
  \end{center}
  \caption{Generalised logical rules}
  \label{rules:generalised-logical rules}
\end{figure}

We conceptually divide the logical rules in three categories:
\begin{itemize}
\item The rules with only one premise and that do not change the number of sequents --- the $0,+,\times$ rules.
\item The rule with two premises --- the $\sqcap$ rule.
\item The rule with only one premise but that adds one sequent to the hypersequent --- the $\sqcup$ rule.
\end{itemize}
Because of the similarities of the rules in each of these categories,  we just prove the CAN-free invertibility of one rule in each category by means of a sequence of lemmas.

$ \ $ \\

\begin{lem}
  If $[ \vdash \Gamma_i , \vec{r}_i. (A \sqcup B)]_{i=1}^n$ has a CAN-free derivation then $[ \vdash \Gamma_i , \vec{r}_i. A \sep \vdash \Gamma_i , \vec{r}_i. B ]_{i=1}^n$ has a CAN-free derivation.
\end{lem}
\begin{proof}
  By induction on the derivation of $[ \vdash \Gamma_i , \vec{r}_i. (A \sqcup B)]_{i=1}^n$. Most cases are easy except the cases for when the derivation ends with a M rule or a $\sqcap$ rule so we will only show those cases.
  \begin{itemize}
  
  \item If the derivation finishes with
    \[\infer[\text{M}]{G \sep \vdash \Gamma_1 , \Gamma_2, \vec{r}_1.(A\sqcup B), \vec{r}_2. (A \sqcup B) }{G \sep \vdash \Gamma_1 , \vec{r}_1. (A \sqcup B) & G \sep \vdash \Gamma_2 , \vec{r}_2. (A \sqcup B)} \]
    with $G = [ \vdash \Gamma_i , \vec{r}_i. (A \sqcup B) ]_{i=3}^n$ and $G' = [ \vdash \Gamma_i , \vec{r}_i. A \sep \vdash \Gamma_i , \vec{r}_i. B ]_{i=3}^n$ then by induction hypothesis on the CAN-free derivations of the premises we have that
    \[ \proveNC G' \sep \vdash \Gamma_1 , \vec{r}_1. A \sep \vdash \Gamma_1 , \vec{r}_1. B \]
    and 
    \[ \proveNC G' \sep \vdash \Gamma_2 , \vec{r}_2. A \sep \vdash \Gamma_2 , \vec{r}_2. B \] are derivable by CAN-free derivations.   We want to prove that both
    \[ \proveNC G' \sep \vdash \Gamma_1 , \vec{r}_1. A \sep \vdash \Gamma_2 , \vec{r}_2. B \]
    and 
    \[ \proveNC G' \sep \vdash \Gamma_2 , \vec{r}_2. A \sep \vdash \Gamma_1 , \vec{r}_1. B \] are CAN-free derivable, as this will allow us to conclude by application of the M rule:
      \[
        \scalebox{0.75}{
          \infer[\text{M}]{G \sep \vdash \Gamma_1, \Gamma_2, \vec{r}_1.A,\vec{r_2}.A \sep \vdash \Gamma_1, \Gamma_2, \vec{r}_1.B,\vec{r_2}.B}{\infer[\text{M}]{G \sep \vdash \Gamma_1, \vec{r}_1.A \sep \vdash \Gamma_1, \Gamma_2, \vec{r}_1.B,\vec{r_2}.B}{G \sep \vdash \Gamma_1, \vec{r}_1.A \sep \vdash \Gamma_2,\vec{r_2}.B & G \sep \vdash \Gamma_1, \vec{r}_1.A \sep \vdash \Gamma_2,\vec{r_2}.B} & \infer[\text{M}]{G \sep \vdash \Gamma_2,\vec{r_2}.A \sep \vdash \Gamma_1, \Gamma_2, \vec{r}_1.B,\vec{r_2}.B}{G \sep \vdash \Gamma_2, \vec{r}_2.A \sep \vdash \Gamma_1,\vec{r_1}.B & G \sep \vdash \Gamma_2, \vec{r}_2.A \sep \vdash \Gamma_2,\vec{r_2}.B}}
        }\]
   
    If $\vec{r}_1 = \emptyset$ or $\vec{r}_2 = \emptyset$, those two hypersequents are derivable using the C rule then the W rule.\\
    Otherwise, by using the W rule, Lemma~\ref{lem:modal_free_genT2} and the M rule, we have
    \[ \proveNC G' \sep \vdash \Gamma_1 , \vec{r}_1. A \sep \vdash \Gamma_2 , \vec{r}_2. B \sep \vdash \vec{r}_2. \Gamma_1,\vec{r}_1. \Gamma_2, (\vec{r}_1\vec{r}_2) A, (\vec{r}_1\vec{r}_2) B \]
    and 
    \[ \proveNC G' \sep \vdash \Gamma_2 , \vec{r}_2. A \sep \vdash \Gamma_1 , \vec{r}_1. B \sep \vdash \vec{r}_2. \Gamma_1,\vec{r}_1. \Gamma_2, (\vec{r}_1\vec{r}_2) A, (\vec{r}_1\vec{r}_2) B \]
    We can then conclude using the S rule, Lemma~\ref{lem:modal_free_genT} and the C rule.
  \item If the derivation finishes with
    \[\infer[\sqcap]{G \sep \vdash \Gamma_1, \vec{r}_1. (A \sqcup B) , \vec{s}.(C \sqcap D)}{G \sep \vdash \Gamma_1, \vec{r}_1. (A \sqcup B) , \vec{s}.C & G \sep \vdash \Gamma_1, \vec{r}_1. (A \sqcup B) , \vec{s}.D} \]
    with $G = [ \vdash \Gamma_i , \vec{r}_i. (A \sqcup B)]_{i=2}^n $ and $G' =  [ \vdash \Gamma_i , \vec{r}_i. A \sep \vdash \Gamma_i ,\vec{r}_i. B]_{i=2}^n$, then by induction hypothesis on the CAN-free derivations of the premises we have that
    \[ \proveNC G' \sep \vdash \Gamma_1 ,\vec{r}_1. A, \vec{s}.C \sep \vdash \Gamma_1 , \vec{r}_1. B,\vec{s}.C\]
    and
    \[ \proveNC G' \sep \vdash \Gamma_1 ,\vec{r}_1. A, \vec{s}.D \sep \vdash \Gamma_1 , \vec{r}_1. B,\vec{s}.D\]
    so by using the W rule and the M rule, we can derive
    \[ \proveNC G' \sep \vdash \Gamma_1 ,\vec{r}_1. A, \vec{s}.C \sep \vdash \Gamma_1 , \vec{r}_1. B,\vec{s}.D \sep \vdash \Gamma_1, \Gamma_1 , \vec{r}_1. A,\vec{r}_1. B, \vec{s}.C , \vec{s}.D \]
    and
    \[ \proveNC G' \sep \vdash \Gamma_1 ,\vec{r}_1. A, \vec{s}.D \sep \vdash \Gamma_1 , \vec{r}_1. B,\vec{s}.C \sep \vdash \Gamma_1, \Gamma_1 , \vec{r}_1. A,\vec{r}_1. B, \vec{s}.C , \vec{s}.D\]
    and then with the S rule and the C rule
    \[ \proveNC G' \sep \vdash \Gamma_1 ,\vec{r}_1. A, \vec{s}.C \sep \vdash \Gamma_1 , \vec{r}_1. B,\vec{s}.D \]
    and
    \[ \proveNC G' \sep \vdash \Gamma_1 ,\vec{r}_1. A, \vec{s}.D \sep \vdash \Gamma_1 , \vec{r}_1. B,\vec{s}.C\]
    We can then conclude with the $\sqcap$ rule. \qedhere
\end{itemize}
\end{proof}

\begin{lem}
  If $[ \vdash \Gamma_i , \vec{r}_i.(A + B) ]_{i=1}^n$ has a CAN-free derivation then $[ \vdash \Gamma_i , \vec{r}_i.A,\vec{r}_i.B ]_{i=1}^n$ has a CAN-free derivation.
\end{lem}
\begin{proof}
  Straightforward induction on the derivation of $[ \vdash \Gamma_i , \vec{r}_i.(A + B) ]_{i=1}^n$. For instance if the derivation finishes with
    \[\infer[\text{M}]{G \sep \vdash \Gamma_1 , \Gamma_2, \vec{r}_1.(A+B), \vec{r}_2.(A + B)}{G \sep \vdash \Gamma_1 , \vec{r}_1. (A + B) & G \sep \vdash \Gamma_2 , \vec{r}_2. (A + B)} \]
    with $G = [ \vdash\Gamma_i , \vec{r}_i. (A + B) ]_{i=3}^n$ and $G' = [ \vdash\Gamma_i , \vec{r}_i. A,\vec{r}_i. B]_{i=3}^n$, 
    then by induction hypothesis on the CAN-free derivations of the premises we have that
    \[ \proveNC G' \sep \vdash \Gamma_1 ,\vec{r}_1. A, \vec{r}_1. B\]
    and
    \[ \proveNC G' \sep \vdash \Gamma_2 ,\vec{r}_2. A, \vec{r}_2. B\]
    so
    \[ \infer[\text{M}]{G' \sep \vdash \Gamma_1,\Gamma_2 ,\vec{r}_1.A, \vec{r}_2. A, \vec{r}_1.B, \vec{r}_2. B }{G' \sep \vdash \Gamma_1 ,\vec{r}_1. A, \vec{r}_1. B & G' \sep \vdash \Gamma_2 ,\vec{r}_2. A, \vec{r}_2. B}\qedhere\]
  
\end{proof}

\begin{lem}
  \label{lem:modal_free_mElim}
  If $[ \vdash \Gamma_i , \vec{r}_i. (A \sqcap B)  ]_{i=1}^n$ has a CAN-free derivation then $[ \vdash \Gamma_i , \vec{r}_i. A ]_{i=1}^n$ and $[\vdash \Gamma_i , \vec{r}_i. B ]_{i=1}^n$ have a CAN-free derivation.
\end{lem}
\begin{proof}
  A straightforward induction on the derivation of $[ \vdash \Gamma_i , \vec{r}_i. (A \sqcap B)  ]_{i=1}^n$. For instance if the derivation finishes with
    \[\infer[\text{M}]{G \sep \vdash \Gamma_1 , \Gamma_2, \vec{r}_1.(A\sqcap B), \vec{r}_2. (A \sqcap B)}{G \sep \vdash \Gamma_1 , \vec{r}_1. (A \sqcap B) & G \sep \vdash \Gamma_2 , \vec{r}_2. (A \sqcap B)} \]
    with $G = [ \vdash \Gamma_i , \vec{r}_i. (A \sqcap B)  ]_{i=3}^n$, 
    then by induction hypothesis on the CAN-free derivations of the premises we have that
    \[ \proveNC [ \vdash \Gamma_i , \vec{r}_i. A  ]_{i=3}^n \sep \vdash \Gamma_1 ,\vec{r}_1. A\]
    and
    \[ \proveNC [ \vdash \Gamma_i , \vec{r}_i. A  ]_{i=3}^n \sep \vdash \Gamma_2 ,\vec{r}_2. A\]
    so
    \[ \infer[\text{M}]{[ \vdash \Gamma_i , \vec{r}_i. A  ]_{i=3}^n \sep \vdash \Gamma_1,\Gamma_2 ,\vec{r}_1.A, \vec{r}_2. A}{[ \vdash \Gamma_i , \vec{r}_i. A  ]_{i=3}^n \sep \vdash \Gamma_1 ,\vec{r}_1. A & [ \vdash \Gamma_i , \vec{r}_i. A  ]_{i=3}^n \sep \vdash \Gamma_2 ,\vec{r}_2. A}\qedhere\]
\end{proof}


\subsection{M-elimination -- Proof of Theorem~\ref{thm:modal_free_m_elim}}
\label{subsec:modal_free_m_elim}

We need to show that for each hypersequent $G$ and sequents $\Gamma$ and $\Delta$, if there exist CAN-free and M-free derivations $d_1$ of $G \sep \vdash \Gamma$ and $d_2$ of $G \sep \vdash \Delta$, then there exists also a CAN-free and M-free derivation of  $G \sep \vdash \Gamma,\Delta$.

The idea behind the proof is to combine $d_1$ and $d_2$ step-by-step. First we  take the derivation $d_1$ and we modify it into a CAN-free and M-free prederivation (i.e., an open derivation) of $$G \sep G \sep \vdash \Gamma, \Delta$$
where all the leaves in the prederivation are either terminated (by the INIT axiom) or non-terminated and of the form:  
 $$G \sep \vdash \vec{r}.\Delta$$ for some vector $\vec{r}$ of scalars. Then we use the derivation $d_2$ to construct a CAN-free and M-free derivation of each $$G \sep \vdash \vec{r}.\Delta$$ hence completing the prederivation of $$G \sep G \sep \vdash \Gamma, \Delta$$ into a full derivation. From this it is possible to obtain the desired CAN-free and M-free derivation of $G\sep \vdash \Gamma,\Delta$ using several times the  C rule:
 
 $$\infer[\text{C}^*]{G \sep \vdash \Gamma, \Delta}
                                {G \sep G \sep \vdash \Gamma , \Delta}$$

In what follows, the first step is formalized as Lemma~\ref{lemma:modal_free_m-elim-first-step} and the second step as Lemma~\ref{lemma:modal_free_m-elim-second-step}.

\begin{lem}
\label{lemma:modal_free_m-elim-first-step}

Let $d_1$ be a CAN-free and M-free derivation of $G \sep\vdash  \Gamma$ and let $\Delta$ be a sequent. Then there exists a CAN-free M-free prederivation of
$$
G \sep G \sep \vdash \Gamma,\Delta.
$$
where all non-terminated leaves are of the form $G\sep \vdash \vec{r}.\Delta$ for some vector $\vec{r}$.
\end{lem}
\begin{proof}
  This is an instance of the slightly more general statement of Lemma~\ref{lem:modal_free_gen-m-elim-first-step} below where:
  \begin{itemize}
  \item $[ \vdash \Gamma_i ]_{i=1}^{n-1} = G$ and $\Gamma_n = \Gamma$.
  \item $\vec r_i = \emptyset$ for $1 \leq i < n$ and $\vec{r}_n = \vec r$.\qedhere
  \end{itemize}
\end{proof}

\begin{lem}
\label{lemma:modal_free_m-elim-second-step}

Let $d_2$ be CAN-free and M-free derivation of $G \sep \vdash \Delta$. Then, for every vector $\vec{r}$, there exists a CAN-free and M-free derivation of
$$
G \sep \vdash\vec{r}.\Delta
$$
\end{lem}
\begin{proof}
This is an instance of the slightly more general statement of Lemma~\ref{lem:modal_free_copyProof} below where:
  \begin{itemize}
  \item $[ \vdash \Delta_i ]_{i=1}^{n-1} = G$ and $\Delta_n = \Delta$.
  \item $\vec r_i = 1$ for $1 \leq i < n$ and $\vec{r}_n = \vec r$.\qedhere
  \end{itemize}
\end{proof}

\begin{lem}
  \label{lem:modal_free_gen-m-elim-first-step}
  Let $d_1$ be a CAN-free and M-free derivation of $[\vdash \Gamma_i]_{i=1}^n$ and let $G$ be a hypersequent and $\Delta$ be a sequent. Then for every sequence of vectors $\vec{r_i}$, there exists a CAN-free M-free prederivation of
  \[ G \sep [\vdash \Gamma_i, \vec{r_i}.\Delta]_{i=1}^n \]
  where all non-terminated leaves are of the form $G \sep \vec{r}.\Delta$ for some vector $\vec{r}$.
\end{lem}
\begin{proof}
  By straightforward induction on $d_1$.  
\end{proof}

\begin{lem}
  \label{lem:modal_free_copyProof}
  If $\left [ \vdash \Delta_i \right ]_{i=1}^n$ has a CAN-free M-free derivation then for all $\vec{r_i}$, there is a CAN-free M-free derivation of $\left [\vdash \vec{r_i}.\Delta_i \right ] _{i=1}^n$. 
\end{lem}
\begin{proof}
  By induction on the derivation of $\left [ \vdash \Delta_i \right ]_{i=1}^n$. We show the only nontrivial case:
  \begin{itemize}  
  \item If the derivation finishes with
    \[ \infer[\text{S}]{\left [ \vdash \Delta_i   \right ]_{i=3}^n \sep \vdash \Delta_1   \sep \vdash \Delta_2  }{\left [ \vdash \Delta_i   \right ]_{i=3}^n \sep \vdash \Delta_1, \Delta_2} \]
    By induction hypothesis there is CAN-free derivation of \[ 
   \left [ \vdash \vec{r_i}.\Delta_i \right ]_{i=3}^n \sep  \vdash (\vec{r_1}\vec{r_2}).\Delta_1, (\vec{r_1}\vec{r_2}).\Delta_2 \]
    If $\vec{r_1}=\emptyset$ or $\vec{r_2}=\emptyset$, we have the empty sequent which is derivable. Otherwise,
    \[ \infer[\text{Lemma~\ref{lem:modal_free_genT}}]{[\vdash \vec{r_i}.\Delta_i]_{i=3}^n \sep \vdash \vec{r_1}.\Delta_1 \sep \vdash \vec{r_2}.\Delta_2}{\infer[\text{S}]{\left [ \vdash \vec{r_i}.\Delta_i \right ]_{i=3}^n \sep \vdash(\vec{r_1}\vec{r_2}).\Delta_1 \sep \vdash(\vec{r_1}\vec{r_2}).\Delta_2}{\left [ \vdash \vec{r_i}.\Delta_i \right ]_{i=3}^n \sep \vdash (\vec{r_1}\vec{r_2}).\Delta_1, (\vec{r_1}\vec{r_2}).\Delta_2}}\qedhere\]
  \end{itemize}
\end{proof}


\subsection{CAN-elimination -- Proof of Theorem~\ref{thm:modal_free_can_elim}}
\label{subsec:modal_free_can_elim}

The CAN rule has the following form:

$$\infer[\text{CAN}, \sum \vec{r} = \sum \vec{s}]{G \sep \vdash \Gamma}{G \sep \vdash \Gamma, \vec{r}.A , \vec{s}.\negF{A}} $$

We prove Theorem~\ref{thm:modal_free_can_elim} by showing that if the hypersequent $G \sep \vdash \Gamma, \vec{r}.A , \vec{s}.\negF{A}$ has a M-free CAN-free derivation then the hypersequent $G \sep \vdash \Gamma$ also has a M-free CAN-free derivation. 

Our proof proceeds by induction on the complexity of the term $A$. The base case is given by $A=x$ (or equivalently $A=\covar{x})$ for some variable $x$. The following lemma proves this base case.

\begin{lem}
  \label{lem:modal_free_atomic_can_elim-simpler}
  If there is a M-free CAN-free derivation of  $G \sep \vdash \Gamma, \vec{r}.x , \vec{s}.\negF{x}$, where $\sum \vec{r} = \sum \vec{s}$ then there exists a M-free CAN-free derivation of $G \sep \vdash \Gamma$.
  \end{lem}
  \begin{proof}
 The statement follows as a special case of Lemma~\ref{lem:modal_free_atomic_can_elim} below, a stronger version of Lemma~\ref{lem:modal_free_atomic_can_elim-simpler} that allows for a simpler proof by induction on the structure of the derivation of $G \sep \vdash \Gamma, \vec{r}.x , \vec{s}.\negF{x}$, where:
    \begin{itemize}
    \item $[\vdash \Gamma_i]_{i=1}^{n-1} = G$ and $\Gamma_n = \Gamma$.
    \item $\vec{r_i}=\vec{r'_i}=\vec{s_i}=\vec{s'_i}=\emptyset$ for $1 \leq i < n$.
    \item $\vec{r_n}=\vec{r}$, $\vec{s_n}=\vec{s}$ and $\vec{r'_n}=\vec{s'_n}=\emptyset$.\qedhere
    \end{itemize}
  \end{proof}

For complex terms $A$, we proceed by using the CAN-free invertibility Theorem~\ref{thm:modal_free_invertibility} as follows:

 \begin{itemize}
\item If $A = x$, we are in the base case of Lemma~\ref{lem:modal_free_atomic_can_elim-simpler}.
\item If $A = 0$, we can conclude with the CAN-free invertibility of the $0$ rule and the M-elimination theorem.
\item If $A = B + C$, since the $+$ rule is CAN-free invertible, $G \sep \vdash \Gamma, \vec{r}.B,\vec{r}.C, \vec{s}.\negF{B},\vec{s}.\negF{C}$ has a M-free CAN-free derivation. Therefore we can have a M-free CAN-free derivation of the hypersequent $G \sep \vdash \Gamma$ by invoking the induction hypothesis twice, since the complexity of $B$ and $C$ is lower than that of $B+C$.

\item If $A = r'B$, since the $\times$ rule is CAN-free invertible, $G \sep \vdash \Gamma, (r'\vec{r}).B, (r'\vec{s}).\negF{B}$ has a M-free CAN-free derivation. Therefore we can have a M-free CAN-free derivation of the hypersequent $G \sep \vdash \Gamma$ by invoking the inductive hypothesis on the simpler term $B$.

\item If $A = B \sqcup C$, since the $\sqcup$ rule is CAN-free invertible, $G \sep \vdash \Gamma, \vec{r}.B , \vec{s}.(\negF{B} \sqcap \negF{C})  \sep \vdash \Gamma, \vec{r}.C  , \vec{s}.(\negF{B} \sqcap \negF{C})$ has a M-free CAN-free derivation. Then, since the $\sqcap$ rule is CAN-free invertible, $G \sep \vdash \Gamma, \vec{r}.B , \vec{s}.\negF{B}  \sep \vdash \Gamma, \vec{r}.C  , \vec{s}.\negF{C}$ has a M-free CAN-free derivation. Therefore we can have a M-free CAN-free derivation of the hypersequent $G \sep \vdash \Gamma \sep \vdash \Gamma$ by invoking the induction hypothesis twice on the simpler terms $B$ and $C$.

We can then derive the hypersequent $G \sep \vdash \Gamma$ as: 
\begin{center}
$\infer[C]{G \sep \vdash \Gamma}{G \sep \vdash \Gamma \sep \vdash \Gamma}$
\end{center}
\item If $A = B \sqcap C$, since the $\sqcup$ rule is CAN-free invertible, $G \sep \vdash \Gamma, \vec{r}.(B \sqcap C), \vec{s}.\negF{B}  \sep \vdash \Gamma, \vec{r}.(B \sqcap C), \vec{s}.\negF{C}$ has a M-free CAN-free derivation. Then, since the $\sqcap$ rule is CAN-free invertible, $G \sep \vdash \Gamma, \vec{r}.B , \vec{s}.\negF{B}  \sep \vdash \Gamma, \vec{r}.C  , \vec{s}.\negF{C}$ has a M-free CAN-free derivation. Therefore we can have a M-free CAN-free derivation of the hypersequent $G \sep \vdash \Gamma \sep \vdash \Gamma$ by invoking the induction hypothesis twice on the simpler terms $B$ and $C$.

We can then derive the hypersequent $G \sep \vdash \Gamma$ as: 
\begin{center}
$\infer[C]{G \sep \vdash \Gamma}{G \sep \vdash \Gamma \sep \vdash \Gamma}$
\end{center}
\end{itemize}
 
This concludes the proof of Theorem~\ref{thm:modal_free_can_elim}. \qed \\

We now prove Lemma~\ref{lem:modal_free_atomic_can_elim}, the stronger version of Lemma~\ref{lem:modal_free_atomic_can_elim-simpler}.
\begin{lem}
  \label{lem:modal_free_atomic_can_elim}
  If there is a CAN-free and M-free derivation of the hypersequent $$\left[ \vdash \Gamma_i, \vec{r_i}.x, \vec{s_i}.\covar{x} \right]_{i=1}^n$$ then for all $\vec{r'_i}$ and $\vec{s'_i}$, with $1\leq i \leq n$, such that $\sum \vec{r_{i}} - \sum \vec{s_{i}} = \sum \vec{r'_{i}} - \sum \vec{s'_{i}}$, there is a CAN-free, M-free derivation of $$\left[ \vdash \Gamma_i, \vec{r'_i}.x, \vec{s'_i}.\covar{x} \right]_{i=1}^n$$
\end{lem}
\begin{proof}
  By induction on the derivation of $\left[ \vdash \Gamma_i, \vec{r_i}.x, \vec{s_i}.\covar{x} \right]_{i=1}^n$. Most cases are trivial, we just describe the most interesting one.
  \begin{itemize}  
    \item If the derivation finishes with:
      \[ \infer[\text{ID}, \sum (\app{\vec a}{\vec b}) = \sum(\app{\vec a'}{\vec b'})]{ \left [ \vdash \Gamma_i,\vec{r_i}. x  ,\vec{s_i}. \covar{x} \right ]_{i\geq2} \sep \vdash \Gamma_1,(\app{\vec{a}}{\app{\vec{b}}{\vec{c}}}).x  , (\app{\vec{a'}}{\app{\vec{b'}}{\vec{c'}}}).\covar{x}}{ \left [ \vdash \Gamma_i,\vec{r_i}. x  ,\vec{s_i}. \covar{x} \right ]_{i\geq2} \sep \vdash \Gamma_1,\vec{c}.x  , \vec{c'}.\covar{x}} \]
      with $\vec{r}_1 = \app{\vec{b}}{\vec{c}}$ and $\vec{s}_1 = \app{\vec{b'}}{\vec{c'}}$.
      We want to show that $$\proveNC \left [ \vdash \Gamma_i,\vec{r'_i}. x  ,\vec{s'_i}. \covar{x} \right ]_{i\geq 2} \sep \Gamma_1, (\app{\vec{a}}{\vec{r'_1}}) .x , (\app{\vec{a'}}{\vec{s'_1}}).\covar{x}$$
      We will now prove that $\sum \vec{c} - \sum \vec{c'} = \sum \vec{r'}_{1} + \sum \vec{a} - (\sum \vec{s'}_{1} + \sum \vec{a'})$ to be able to conclude with the induction hypothesis.
      \begin{eqnarray*}
        \sum \vec{c} - \sum \vec{c'} & = & (\sum \vec r_1 - \sum\vec b) - (\sum \vec s_1 - \sum\vec b') \\
                                     & = & (\sum \vec r_1 - \sum \vec s_1) + (\sum \vec b ' - \sum \vec b) \\
                                     & = & (\sum \vec {r'}_1 - \sum \vec {s'}_1) + (\sum \vec a - \sum \vec a') \\
                                     & = & \sum \vec{r'}_{1} + \sum \vec{a} - (\sum \vec{s'}_{1} + \sum \vec{a'})
      \end{eqnarray*} so by induction hypothesis, we have
      \[ \proveNC \left [ \vdash \Gamma_i,\vec{r'_i}. x  ,\vec{s'_i}. \covar{x} \right ]_{i\geq 2} \sep \Gamma_1, (\app{\vec{a}}{\vec{r'_1}}) .x , (\app{\vec{a'}}{\vec{s'_1}}).\covar{x} \]
      which is the result we want.\qedhere
  \end{itemize}
\end{proof}


\subsection{Rational T-elimination -- Proof of Theorem~\ref{thm:modal_free_conservativity}}
\label{subsec:modal_free_conservativity}
We need to prove that if a hypersequent sequent $G$, with all scalars in $\mathbb{Q}$, has a CAN-free derivation then it also has a CAN-free and T-free derivation.

Firstly, we observe that we can restrict to the case of $G$ being an atomic hypersequent. Indeed, if $G$ is not atomic, we can iteratively apply the logical rules (see Figure~\ref{fig:algo_to_inv} on page \pageref{fig:algo_to_inv}) and reduce $G$ to a number of atomic hypersequents $G_1, \dots, G_n$. By the CAN-free invertibility Theorem~\ref{thm:modal_free_invertibility}, $G$ is CAN-free derivable if and only if all $G_i$ are CAN-free derivable.

Secondly, assume $G$ is atomic and has a CAN-free derivation. Then, by application of Lemma~\ref{lem:modal_free_lambda_prop} below and using the same notation, there are $t_1,...,t_m$ in $\mathbb{R}_{\geq 0}$ such that
\begin{itemize}
  \item there exists $i\in [1..m]$ such that $t_i \neq 0$ and
  \item for every variable and covariable $(x_j, \covar{x_j})$ pair, it holds that
  $$
  \sum^{m}_{i=1} t_i (\sum \vec{r}_{i,j}) =   \sum^{m}_{i=1} t_i (\sum \vec{s}_{i,j}) 
  $$ 
\end{itemize} Since all coefficients are rational and the theory of linear arithmetic over $\mathbb{R}$ is an elementary extension of that of linear arithmetic over $\mathbb{Q}$~\cite{FR1975}, there are $q_1,...,q_m\in\mathbb{Q}_{\geq 0}$ satisfying the same property of $t_1,...,t_m$. By multiplying all $q_i$ by the least common multiple of their denominators, we get a solution $k_1,...,k_m$ in $\mathbb{N}$. So according to Lemma~\ref{lem:modal_free_int_lambda_prop}, $G$ has also a CAN-free and T-free derivation. This concludes the proof.

We now state a similar result to Lemma~\ref{lem:modal_free_int_lambda_prop} regarding derivations that use the T rule. The only difference is that since the T rule can multiply a sequent by any strictly positive real number, the coefficients in the statement are arbitrary positive real numbers instead of natural numbers.
\begin{lem}
  \label{lem:modal_free_lambda_prop}
  For all atomic hypersequents $G$, built using the variables and negated variables $x_1, \covar{x_1}, \dots, x_{k}, \covar{x_{k}}$, of the form  
  $$\vdash \Gamma_1 \sep \dots \sep \vdash \Gamma_m$$
where $\Gamma_i =  \vec{r}_{i,1}.x_1,...,\vec{r}_{i,k}.x_k, \vec{s}_{i,1}.\covar{x_1},...,\vec{s}_{i,k}.\covar{x_k}$, the following are equivalent:
  
  \begin{enumerate}
  \item $G$ has a derivation.
  \item there exist numbers $t_1,...,t_m \in \mathbb{R}_{\geq 0}$, one for each sequent  in $G$, such that:
  \begin{itemize}
  \item there exists $i\in [1..m]$ such that $t_i \neq 0$, i.e., the numbers are not all $0$'s, and
  \item for every variable and covariable $(x_j, \covar{x_j})$ pair, it holds that
  $$
  \sum^{m}_{i=1} t_i (\sum \vec{r}_{i,j}) =   \sum^{m}_{i=1} t_i (\sum \vec{s}_{i,j}) 
  $$
i.e., the scaled (by the numbers $t_1$ \dots $t_m$) sum of the coefficients in front of the variable $x_j$ is equal to the scaled sum of the coefficients in from of the covariable $\covar{x_j}$. 
  \end{itemize}
  \end{enumerate}
\end{lem}
\begin{proof}
  We prove $(1) \Rightarrow (2)$ by induction on the derivation of $G$. By using Theorem~\ref{thm:modal_free_can_elim}, we can assume that the derivation of $G$ is CAN-free. We will only deal with the case of T rule since every other cases are exactly the same as in Lemma~\ref{lem:modal_free_int_lambda_prop}. If the derivation finishes with
  \[ \infer[\text{T}]{\vdash \Gamma_1 \sep \dots \sep \vdash \Gamma_m}{\vdash \Gamma_1 \sep \dots \sep \vdash r.\Gamma_m} \]
  then by induction hypothesis there are $t_1,...,t_m \in \mathbb{R}$ such that :
  \begin{itemize}
  \item there exists $i\in [1..m]$ such that $t_i \neq 0$.
  \item for every variable and covariable $(x_j,\covar{x_j})$ pair, it holds that $\sum_{i=0}^{m-1} t_i.\sum \vec{r}_{i,j} + t_m.\sum r\vec{r}_{m,j}= \sum_{i=0}^{m-1} t_i.\sum \vec{s}_{i,j} + t_m.\sum r\vec{s}_{m,j}$.
  \end{itemize}
  so $t_1,\dots,t_{m-1},rt_m$ satisfies the property.
  
  The other way ($(2) \Rightarrow (1)$) is very similar to Lemma~\ref{lem:modal_free_int_lambda_prop}, only using the T rule instead of the C and S rules. If there exist numbers $t_1,...,t_m \in \mathbb{R}$, one for each sequent  in $G$, such that:
  \begin{itemize}
  \item there exists $i\in [1..m]$ such that $t_i \neq 0$ and
  \item for every variable and covariable $(x_j, \covar{x_j})$ pair, it holds that
    $$
    \sum^{m}_{i=1} t_i (\sum \vec{r}_{i,j}) =   \sum^{m}_{i=1} t_i (\sum \vec{s}_{i,j}) 
    $$
  \end{itemize}
  then we can use the W rule to remove the sequents corresponding to the numbers $t_i = 0$, and use the T rule on the $i$th sequent to multiply it by $t_i$. If we assume that there is a natural number $l$ such that $t_i = 0$ for all $i > l$ and $t_i \neq 0$ for all $i \leq l$, then the CAN-free derivation is:
  \[ \infer[\text{W}^*]{\vdash \Gamma_1 \sep \dots \sep \Gamma_m}{\infer[\text{T}^*]{\vdash \Gamma_1 \sep \dots \sep \vdash \Gamma_l}{\infer[\text{S}^*]{\vdash t_1.\Gamma_1 \sep \dots \sep \vdash t_l.\Gamma_l}{\infer[\text{ID}^*]{\vdash t_1.\Gamma_1,\dots,t_l.\Gamma_l}{\infer[\text{INIT}]{\vdash}{}}}}} \qedhere\]
\end{proof}


\subsection{Decidability -- Proof of Theorem~\ref{thm:modal_free_decidability}}
\label{subsec:modal_free_decidability}

The previous results give us a simple algorithm for deciding if a hypersequent $G$ is derivable in the system \hr . 
We do not claim that this algorithm is optimal, and we merely prove that it has elementary complexity. It is valuable, however, because it will be adaptable to the context of the more complex system \hmr.

The algorithm works in two steps:
\begin{enumerate}
\item the problem of deciding if $G$ is derivable is reduced to the problem of deciding if a finite number of atomic hypersequents $G_1,\dots, G_n$ are derivable.
\item A decision procedure for atomic hypersequents is executed and it verifies if all hypersequents computed at the first step are derivable.
\end{enumerate}

The first step consists in applying recursively all possible logical rules to $G$ until atomic premises $G_1,...,G_n$ are obtained (see Figure~\ref{fig:algo_to_inv} on page \pageref{fig:algo_to_inv}). Indeed, the CAN-free invertibility Theorem~\ref{thm:modal_free_invertibility} guarantees that $G$ is derivable if and only if all the atomic hypersequents obtained in this way are derivable.

The second step can be performed using Lemma~\ref{lem:modal_free_lambda_prop} which states that the hypersequent $G_{i}$ is derivable if and only if there exists a sequence of real numbers $\vec{t}\in\mathbb{R}_{\geq 0}$ satisfying the system of (in)equations of Lemma~\ref{lem:modal_free_lambda_prop}. This can be expressed directly by a (existentially quantified) formula in the first order theory of the real-closed field $FO(\mathbb{R},+,\times,\leq)$. It is well known that this theory is decidable and admits quantifier elimination~\cite{tarski1951,GRIGOREV198865}. Thus it is possible to decide if this formula is satifiable or not, that is, if the atomic hypersequent $G_i$ is derivable or not.

The idea behind the above algorithm, reducing the problem of derivability to the problem of verifying the satisfiability of formulas in the first order theory of the real-closed field, can in fact be pushed forward. Not only we can decide if $G$ is derivable or not, but we can return a formula $\phi\in  FO(\mathbb{R},+,\times,\leq)$ which describes the set of real-values assigned to the scalars in $G$ that admits a derivation. For example, as explained in Section~\ref{subsec:modal_free_main_results}, consider the following simple hypersequent
$$
\vdash r.x, r.\covar{x}
$$ 
Not only this hypersequent is derivable for a fixed scalar $r\in\mathbb{R}_{>0}$, but the hypersequent
$$
\vdash \alpha.x, \alpha.\covar{x}
$$
is derivable for any assignment of concrete scalars in $r\in\mathbb{R}$ to the scalar-variable $\alpha$ such that $r > 0$.

Similarly, the hypersequent containing the scalar-variable $\alpha$ and two concrete scalars $s$ and $t$
$$
\vdash \alpha.x, s.\covar{x}, t.\covar{x}
$$ 
is derivable for all concrete $r\in\mathbb{R}$
 assignments to $\alpha$ such that $r > 0$ and $r = s + t$.

 Lastly, the hypersequent containing two scalar-variables $\alpha,\beta$ and two concrete scalars $s$ and $t$
 
 $$
\vdash (\alpha^2 - \beta).x, s.\covar{x}, t.\covar{x}
$$ 
is derivable for any  assignment of concrete assignments $r_1,r_2\in \mathbb{R}$ to $\alpha$ and $\beta$ such that $(r_1)^2-r_2 > 0$ and $(r_1)^2-r_2 = s+t$. 

Hence we can generally consider hypersequents having polynomials (over a set $\alpha_1, \dots, \alpha_l$ of scalar-variables) in place of concrete scalars.

We now describe an algorithm that takes a hypersequent $G$ as input, having polynomials $R_1,\dots, R_k\in\mathbb{R}[\alpha_1,...,\alpha_l]$ over scalar-variables $\alpha_1, \dots, \alpha_l$ as coefficients in weighted terms and returns a formula $\phi_G(\alpha_1,\dots, \alpha_l) \in   FO(\mathbb{R},+,\times,\leq)$ with $l$ variables $\alpha_1,\dots, \alpha_l$, such that for all $(s_1,...,s_l)\in\mathbb{R}$ such that for all $i\in[1...k], R_i(s_1,...,s_l) > 0$, we have:
$$
\phi_G(s_1,\dots,s_l)\text{ holds in }\mathbb{R}$$
$$ \Leftrightarrow$$ 
$$G[s_j/\alpha_{j}] \textnormal{ is derivable.}
$$
\noindent where $ G[s_j/\alpha_{j}]$ denotes the concrete hypersequent obtained by instantiating the scalar-variable $\alpha_j$ with the real number $s_j$.

The algorithm takes as input $G$ and proceeds, again, in two steps: 
\begin{enumerate}
\item The algorithm returns 
  \begin{equation}
    \phi_G = \bigwedge^{n}_{i=1} \phi_{G_i}
  \end{equation}
where $G_1,\dots, G_n$ are the atomic hypersequents obtained by iteratively applying the logical rules, and $\phi_{G_i}$ is the formula recursively computed by the algorithm on input $G_i$.

\item if $G$ is atomic then $G$ has the shape \[\vdash \Gamma_1 \sep ... \sep \vdash \Gamma_m\] where $\Gamma_i =  \vec{R}_{i,1}.x_1,...,\vec{R}_{i,k}.x_k, \vec{S}_{i,1}.\covar{x_1},...,\vec{S}_{i,k}.\covar{x_{k}}$ (recall that the hypersequent $\vdash$ is a specific instance of an atomic hypersequent where $m = 1$ and all vectors are empty). For all $I \subsetneq [1...m]$, we then define
  \begin{itemize}
  \item A formula $Z_I(\beta_1,...,\beta_m)$ that states that for all $i \in I$, $\beta_i = 0$.
    \[ Z_I(\beta_1,...,\beta_m) = \bigwedge_{i\in I} (\beta_i = 0) \]
  \item A formula $NZ_I(\beta_1,...,\beta_m)$ that states that for all $i \notin I$, $0 < \beta_i$.
    \[ NZ_I(\beta_1,...,\beta_m) = \bigwedge_{i \notin I} (0 \leq \beta_i) \wedge \neg(\beta_i = 0) \]
  \item A formula $A_I(\beta_1,...,\beta_m)$ that states that all the atoms cancel each other.
    \[ A_I(\beta_1,...,\beta_m) =  \bigwedge_{j=0}^k (\sum_{i=1}^m \beta_i \sum \vec{R}_{i,j} = \sum_{i=1}^m \beta_i \sum \vec{S}_{i,j})\]
     
    \item A formula $\phi_{G,I}$ that corresponds to $\phi_{G'}$ where $G'$ is the hypersequent obtained on using the W rule on all $i$-th sequents for $i \in I$, i.e. the leaf of the following prederivation:
      \[ \infer[\text{W}^*]{\vdash \Gamma_{1} \sep ... \sep \vdash \Gamma_{m}}{\vdash \Gamma_{k_1} \sep ... \sep \vdash \Gamma_{k_l}} \] with $\{k_1,...,k_l\} =[1..m]\backslash I $. Then
    \[ \phi_{G,I}=\exists \beta_1,...,\beta_m, Z_I(\beta_1,...,\beta_m) \wedge NZ_I(\beta_1,...,\beta_m) \wedge A_I(\beta_1,...,\beta_m) \]
  \end{itemize}
  The formula $\phi_G$ is then constructed as follow:
    \[ \phi_G = \bigvee_{I \subsetneq [1...m]} \phi_{G,I} \]
\end{enumerate}

The following theorem states the correctness of the above described algorithm.
\begin{thm}
  \label{thm:modal_free_to_FO}
Let $G$ be a hypersequent having polynomials $R_1,\dots, R_k\in\mathbb{R}[\alpha_1,...,\alpha_l]$ over scalar-variables $\alpha_1, \dots, \alpha_l$. Let $\phi_G(\alpha_1,\dots,\alpha_l)$ be the formula returned by the algorithm described above on input $G$. Then, for all $s_1,\dots, s_l\in\mathbb{R}$ such that for all $i\in[1..k], R_i(s_1,...,s_l) > 0$, the following are equivalent:
\begin{enumerate}
\item $\phi_G(s_1,\dots,s_l)$ holds in $\mathbb{R}$,
\item $G[s_j/\alpha_{j}]$ is derivable in \hr.
\end{enumerate}
\end{thm}
\begin{proof}
  If $G$ is atomic, the theorem is a direct corollary of Lemma~\ref{lem:modal_free_lambda_prop}.  
  So assume $G$ is not atomic, i.e., the terms in $G$ contain some logical connective.   Given any vector of scalars $s_1,\dots, s_l\in\mathbb{R}$,  by using the CAN-free invertibility Theorem~\ref{thm:modal_free_invertibility}, $G[s_j/\alpha_{j}]$ is derivable if and only if all $G_i[s_j/\alpha_{j}]$ are derivable, where the hypersequents $G_i$ are the atomic hypersequents obtainable from $G$ by repeated applications of the logical rules, as show in Figure~\ref{fig:algo_to_inv}. Hence, the set of scalars $s_1,\dots, s_l\in\mathbb{R}$ that allows for a derivation of $G$ is exactly the intersection of the scalars that allow derivations of each $G_i$. This is precisely the semantics of:
     \[\phi_G = \bigwedge^{n}_{i=1} \phi_{G_i}\qedhere\]
       
\end{proof}

The size of the formula $\phi_G$ can be bounded by a double exponential in the number of sequents and operators in $G$, so the algorithm described previously is elementary.
\begin{clm}
  \label{cml:complexity}
  Let $G$ be a hypersequent having polynomials $R_1,...,R_k\in\mathbb{R}[\alpha_1,...,\alpha_l]$ of degree at most $d$. Let $p$ be the number of different variables appearing in the terms of $G$, $q$ the number of sequents in $G$ and $o$ the total number of operators appearing in $G$. Then $\phi_G$ is equivalent to a formula
  \[ \exists \vec{\beta}, P(\alpha_1,...,\alpha_l,\vec{\beta}) \]
  where $P$ is a quantifier-free formula with at most $2^{4\times 2^{(q+o)}}(1 + p)$ polynomials of degree at most $d + 1$ and where the size of $\vec{\beta}$ is at most $2^{3\times 2^{(q+o)}}$.
\end{clm}
\begin{proof}[Idea of the proof]
  If $G$ is atomic, we have $$\phi_G = \bigvee_{I \subsetneq [1...q]} \exists \beta_1,...,\beta_q, Z_I(\beta_1,...,\beta_q) \wedge NZ_I(\beta_1,...,\beta_q) \wedge A_I(\beta_1,...,\beta_q)$$ which, using basic identities of first order logic, is equivalent to
  \[ \exists \beta_1,...,\beta_q, \bigvee_{I \subsetneq [1..m]}  Z_I(\beta_1,...,\beta_q) \wedge NZ_I(\beta_1,...,\beta_q) \wedge A_I(\beta_1,...,\beta_q) \]
  so the upper bounds are satisfied.

  If $G$ is not atomic, then
  \[ \phi_G = \bigwedge^{n}_{i=1} \phi_{G_i} \]
  where the $G_i$ are the atomic hypersequents obtainable from $G$ by repeated applications of the logical rules (see Figure~\ref{fig:algo_to_inv}). Since the $G_i$ are atomic, we have just shown that each $\phi_{G_i}$ is equivalent to
  \[ \exists \vec{\beta}_i, P_i(\alpha_1,...,\alpha_l,\vec{\beta_i}) \]
  so $\phi_G$ is equivalent to
  \[ \exists \vec{\beta}_1,...,\vec{\beta}_n, \bigwedge^n_{i=1} P_i(\alpha_1,...,\alpha_l,\vec{\beta_i}) \]

  and we need to obtain an upper bound on $n$ and the size of all $G_i$.

  Each application of the $\sqcap$ rule done to compute the $G_i$ increases $n$, and each application of the $\sqcup$ rule increases the size of the atomic hypersequents $G_i$ so we have to compute the maximum number of times those rules are used.

  Applications of the $\sqcup$ rule can, potentially, duplicate the number of $\sqcap$ connectives in the hypersequent, and thus duplicate the steps needed. Moreover, since the $\sqcap$ rule has two premises on which the procedure is recursively iterated, each $\sqcap$ operator can also double the number of steps, and so we can bound the number of times the $\sqcap$ rule is used by $2^{(\#\sqcap\text{ operators})2^{\#\sqcup\text{ operators}}}-1$, thus the double exponential.

  Some simplifications are then done to obtain the upper bounds given in the claim, which are easier to manipulate.
\end{proof}
\begin{cor}
  The algorithm is in $3$-ExpTime.
\end{cor}
\begin{proof}
  The reduction of $$\phi_G = \bigwedge^{n}_{i=1} \phi_{G_i}$$ into $$\exists \vec{\beta}_1,...,\vec{\beta}_n, \bigwedge^n_{i=1} P_i(\alpha_1,...,\alpha_l,\vec{\beta_i})$$ has a time complexity linear in $n$ and the sum of the sizes of $\vec{\beta_i}$, and thus is in $2$-ExpTime since those two values can be bounded by a double exponential in the size of $G$.
  
  The algorithm decribed in~\cite{GRIGOREV198865} to decide the existential fragment of the first order theory of the real-closed field has a complexity of $M(kd)^{(O(n))^2}$ where $k$ is the number of polynomials in the formula, $d$ the maximal degree of those polynomials, $2^M$ is an upper bound for the absolute value of every coefficients in those polynomials and $n$ the number of variables of those polynomials.

  By using this algorithm, the time complexity to decide whether or not the reduced $\phi_G$ is provable is smaller than $M(2^{4\times 2^{(q+o)}}(1 + p)(d + 1))^{(O(2^{3\times 2^{(q+o)}}))^2}$ with the notations of Claim~\ref{cml:complexity}, and so is in $3$-ExpTime. Thus our algorithm is in $3$-ExpTime.
\end{proof}


\section{Hypersequent Calculus for Modal Riesz Spaces} \label{modal_section}

In this section we extend the system \hr\ into the hypersequent calculus \hmr\ for the equational theory of modal Riesz spaces. For that purpose we introduce to the system the two new rules of Figure~\ref{rules:hmr} each dealing with the additional operators (the $1$ constant and the  unary $\Diamond$ modality) available in the syntax of modal Riesz spaces.

In the $(\Diamond)$ rule, and in the rest of this section, the notation $\Diamond\Gamma$ stands for the sequence $r^\prime_1.\Diamond A_1,...,r^\prime_n.\Diamond A_n$ when $\Gamma = r^\prime_1.A_1, ...,r^\prime_n.A_n$.

\begin{figure}[h!]
  \begin{center}
  \scalebox{0.85}{
      \begin{minipage}{12.5cm}
      \begin{center}
          \begin{tabular}{cc}
            $\infer[1,\ \  \sum \vec{r} \geq \sum \vec{s}]
            {G \sep\vdash \Gamma , \vec{r}.1, \vec{s}.\covar{1}}
            {G \sep\vdash \Gamma}$
            & \ \ \ \ \ \ \ \ \
            $\infer[\Diamond, \ \ \sum \vec{r} \geq \sum \vec{s}]
            {\vdash \Diamond\Gamma , \vec{r}.1, \vec{s}.\covar{1}}
            {\vdash \Gamma , \vec{r}.1, \vec{s}.\covar{1}}$ \\[0.2cm]
          \end{tabular}
          \end{center}	

            \end{minipage}
    }
  \end{center}
  \caption{Additional rules of \hmr.}
  \label{rules:hmr}
\end{figure}

\begin{defi}[System \hmr] The hypersequent calculus proof system \hmr\ consists of the rules of Figure~\ref{rules:hr} (i.e., those of the hypersequent calculus \hr) plus the rules of Figure~\ref{rules:hmr}.
\end{defi}

The ($1$) rule is quite similar to the ID rule but it reflects the axiom $0 \leq 1$ (see Figure~\ref{axioms:of:modal_riesz:spaces}) of modal Riesz spaces, and thus the side condition expresses an inequality, rather than an equality. Note that the seemlingly simpler variant of the $1$-rule
\[ \infer{G \sep\vdash \Gamma , \vec{r}.1}{G \sep\vdash \Gamma} \]
would not allow us to prove the CAN-elimination theorem. For instance the hypersequent $\vdash 2.1, 1.\covar{1}$ would not have a CAN-free derivation with this rule. The $\Diamond$ rule, as we will show in the soundness and completeness theorems below (Theorem~\ref{thm:soundness} and Theorem~\ref{thm:completeness}), is remarkably capturing in one single rule all three axioms regarding the ($\Diamond$) modality  (see Figure~\ref{axioms:of:modal_riesz:spaces}).

\begin{rem}
  \label{rem:unsound_rules}
Note how the $\Diamond$ rule imposes strong constraints on the shape of its (single) premise and conclusion. First, both the conclusion and the premise are required to be hypersequents consisting of exactly one sequent. Furthermore, in the conclusion, all terms, except those of the form $1$ and $\covar{1}$ need to be of the form $\Diamond A$ for some term $A$. These constraints determine the main difficulties when trying to adapt the proofs of Section~\ref{modal_free_section} for the system \hmr, but they are necessary. Indeed, for example, the following two alternative relaxed rules, while more natural looking, are in fact not sound:
\begin{multicols}{2}
  \setlength{\columnseprule}{0pt}

  \[ \infer[\Diamond_1]{G \sep \vdash \Diamond \Gamma, \vec{r}.1,\vec{s}.\covar 1}{G \sep \vdash \Gamma, \vec{r}.1,\vec{s}.\covar 1}\]
  
  \[ \infer[\Diamond_2]{\vdash \Gamma,\vec{r'}.\Diamond A}{\vdash \Gamma,\vec{r'}.A}  \]
\end{multicols}

\noindent as Remark~\ref{remark:notsound} below shows. Moreover, if we replace the $\Diamond$ rule with the rule:
\[ \infer[\Diamond_3]{\vdash \Diamond \Gamma}{ \vdash \Gamma}\]
the system is not complete as the axiom $\Diamond 1 \leq 1$ is no longer derivable.
\end{rem}

The interpretation of \hmr\ weighted terms,  sequents and hypersequents is defined exactly as in Definition~\ref{defi:modal_free_interpretation} for the system \hr. That is, a weighted term is the term scalar-multiplied by the weight, a sequent is the sum of its weighted terms and a hypersequent is the join of its sequents. Throughout this section we adopt similar notation to that introduced in Section~\ref{modal_free_section} for the system \hr\ and we write $\proveM G$ if $G$ is derivable using the rules of the system \hmr.

The $\Diamond$ rule has a very peculiar place in the system \hmr: a \hmr\ derivation can be seen as a sequence of \hr\ derivations separated by a $\Diamond$ rule. Some results on the system \hmr\ can then be proven by induction on the number of $\Diamond$ rules appearing in a branch of the derivations, which we call the \emph{modal depth} of the derivation: the basic case is very similar to a \hr\ derivation --- we just add the $1$ rule and the proofs for the system \hr\ can be easily adapted to deal with this additionnal rule.
\begin{defi}[Modal depth]
  The \emph{modal depth} of a derivation is the maximal number of $\Diamond$ rules used in a branch of the derivation.
\end{defi}
\begin{rem}
Note that the modal depth of a derivation is not necessarily the same as the modal depth of the end hypersequent. Indeed, the derivation could introduce terms with $\Diamond$ operators by using the CAN rule, and thus can make the modal depth of the derivation greater than the modal depth of the end hypersequent.
\end{rem}

\subsection{Main results regarding the system \hmr}
\label{subsec:main_results}

This section presents our main results regarding the hypersequent calculus \hmr\ and has the same pattern of Section~\ref{subsec:modal_free_main_results} as we have tried to follow the same lines of presentation and reasoning and, whenever possible, to adapt the same proof techniques.  We have been able to obtain variants of all the results proved for the system $\hr$ with the notable exception of the Rational T-elimination Theorem (Theorem~\ref{thm:modal_free_conservativity}) which remains an open problem in the context of the system \hmr (see Section~\ref{open:problem}).

Our first two technical results about \hmr{}, the soundness and completeness theorems, state that the system \hmr{} can derive all and only those hypersequents $G$ such that $\mathcal{A}^\Diamond_{\textnormal{Riesz}}\vdash \sem{G} \geq 0$.
\begin{thm}[Soundness]
\label{thm:soundness}
For every hypersequent $G$, 
$$
\proveM G\ \ \ \Longrightarrow\ \ \ \mathcal{A}^\Diamond_{\textnormal{Riesz}}\vdash \sem{G} \geq 0.
$$
\end{thm}

\begin{thm}[Completeness]
\label{thm:completeness}
For every hypersequent $G$, 
$$
\mathcal{A}^\Diamond_{\textnormal{Riesz}} \vdash\sem{G} \geq 0 \ \ \ \Longrightarrow\ \ \  \proveM G.
$$
\end{thm}

\begin{rem}\label{remark:notsound}
  The alternative rules considered in Remark~\ref{rem:unsound_rules} are unsound since the hypersequent $\vdash 1.\Diamond(\covar{x}\sqcap \covar{y}), 1.\Diamond(x) \sqcup \Diamond(y)$ would be derivable (see below) while the hypersequent $\mathcal{A}^\Diamond_{\textnormal{Riesz}}\not\vdash \Diamond(-x \sqcap -y) + \big(\Diamond(x) \sqcup \Diamond(y)\big) \geq 0$ or, equivalently,  $\mathcal{A}^\Diamond_{\textnormal{Riesz}}\not\vdash\Diamond(x \sqcup y) \leq \Diamond(x) \sqcup \Diamond(y)$ would not (see Example~\ref{example:false-equality}). 
  \[ \scalebox{0.85}{\infer[\sqcup]{\vdash 1.\Diamond(\covar{x}\sqcap \covar{y}), 1.\Diamond(x) \sqcup \Diamond(y)}{\infer[\Diamond_1]{\vdash 1.\Diamond(\covar{x}\sqcap \covar{y}), 1.\Diamond(x) \sep \vdash 1.\Diamond(\covar{x}\sqcap \covar{y}), 1.\Diamond(y)}{\infer[\Diamond_1]{\vdash 1.\covar{x}\sqcap \covar{y}, 1.x \sep \vdash 1.\Diamond(\covar{x}\sqcap \covar{y}), 1.\Diamond(y)}{\infer[\sqcap]{\vdash 1.\covar{x}\sqcap \covar{y}, 1.x \sep \vdash 1.\covar{x}\sqcap \covar{y}, 1.y}{\infer[\text{W}]{\vdash 1.\covar{x}, 1.x \sep \vdash 1.\covar{x}\sqcap \covar{y}, 1.y}{\infer[\text{ID}]{\vdash 1.\covar{x},1.x}{\infer[\text{INIT}]{\vdash}{}}} & \infer[\sqcap]{\vdash 1.\covar{y}, 1.x \sep \vdash 1.\covar{x}\sqcap \covar{y}, 1.y}{\infer[\text{S}]{\vdash 1.\covar{y}, 1.x \sep \vdash 1.\covar{x}, 1.y}{\infer[\text{ID}]{\vdash 1.\covar{x},1.\covar{y},1.x,1.y}{\infer[\text{ID}]{\vdash 1.\covar{y},1.y}{\infer[\text{INIT}]{\vdash}{}}}} & \infer[\text{W}]{\vdash 1.\covar{y}, 1.x \sep \vdash 1.\covar{y}, 1.y}{\infer[\text{ID}]{\vdash 1.\covar{y},1.y}{\infer[\text{INIT}]{\vdash}{}}}}}}}}} \]
  \[ \scalebox{0.85}{\infer[\Diamond_2]{\vdash 1.\Diamond(\covar{x}\sqcap \covar{y}), 1.\Diamond(x) \sqcup \Diamond(y)}{\infer[\sqcup]{\vdash 1.(\covar{x}\sqcap\covar{y}),1.\Diamond(x)\sqcup\Diamond(y)}{\infer[\sqcap]{\vdash 1.\covar{x}\sqcap \covar{y}, 1.\Diamond(x) \sep \vdash 1.\covar{x}\sqcap \covar{y}, 1.\Diamond(y)}{\infer[\text{W}]{\vdash 1.\covar{x}, 1.\Diamond(x) \sep \vdash 1.\covar{x}\sqcap \covar{y}, 1.\Diamond(y)}{\infer[\Diamond_2]{\vdash 1.\covar{x},1.\Diamond(x)}{\infer[\text{ID}]{\vdash 1.\covar{x},1.x}{\infer[\text{INIT}]{\vdash}{}}}} & \infer[\sqcap]{\vdash 1.\covar{y}, 1.\Diamond(x) \sep \vdash 1.\covar{x}\sqcap \covar{y}, 1.\Diamond(y)}{\infer[\text{S}]{\vdash 1.\covar{y}, 1.\Diamond(x) \sep \vdash 1.\covar{x}, 1.\Diamond(y)}{\infer[\Diamond_2]{\vdash 1.\covar{x},1.\covar{y},1.\Diamond(x),1.\Diamond(y)}{\infer[\Diamond_2]{\vdash 1.\covar{x},1.\covar{y},1.x, 1.\Diamond(y)}{\infer[\text{ID}]{\vdash 1.\covar{x},1.\covar{y},1.x,1.y}{\infer[\text{ID}]{\vdash 1.\covar{y},1.y}{\infer[\text{INIT}]{\vdash}{}}}}}} & \infer[\text{W}]{\vdash 1.\covar{y}, 1.\Diamond(x) \sep \vdash 1.\covar{y}, 1.\Diamond(y)}{\infer[\Diamond_2]{\vdash 1.\covar{y},1.\Diamond(y)}{\infer[\text{ID}]{\vdash 1.\covar{y},1.y}{\infer[\text{INIT}]{\vdash}{}}}}}}}}} \]
\end{rem}

Our next theorem states that all logical rules already present in the system $\hr$ and the $\Diamond$ rule are CAN-free invertible.

\begin{thm}[CAN-free Invertibility]
\label{thm:invertibility}
All the logical rules $\{ 0, +, \times, \sqcup, \sqcap, \Diamond \}$ are CAN-free invertible in \hmr.
\end{thm}
The proof of this result is obtained by induction of the structure of derivations and is essentially identical to the one provided in Section~\ref{subsec:modal_free_invertibility}. The new case represented by the $\Diamond$ rule presents no specific difficulties.

\begin{rem}
Note that, while the proof of CAN-free invertibility of the $\Diamond$ rule (Theorem~\ref{thm:invertibility} above) is not particularly difficult, the general invertibility property, i.e., that if the conclusion of a $\Diamond$ rule is derivable (possibly using CAN rules) then also its premise is derivable, appears to be quite nontrivial. We are able to prove this general form of invertibility only as a corollary of the CAN elimination theorem (Theorem~\ref{thm:can_elim} below) for \hmr. Note how this contrasts with the case of the other logical rules already present in the system $\hr$ ($ 0, +, \times, \sqcup, \sqcap\}$) whose general invertibility is straightforward to prove (see Lemma~\ref{lem:modal_free_can_full_invertibility} and~\ref{lem:can_full_invertibility} below) without invoking the CAN elimination theorem.
\end{rem}

Note, instead, that the rule for the constant ($1$) is not CAN-free invertible in \hmr{}. For example, the conclusion of the following valid instance of the rule:
\[ \infer[1,\ \ \ \frac{2}{3} \geq 0]{\vdash \frac{2}{3}.1,\frac{1}{3}.\covar{1}}{\vdash \frac{1}{3}.\covar{1}} \] is derivable but its premise (whose semantics is $\sem{ \vdash \frac{1}{3}. \covar{1}} = \frac{1}{3}(-1)$) is, by the Soundness theorem, not derivable.

The importance of the invertibility theorem, in the context of the \hr\ system (and also similar systems which inspired our work, such as the system GA of~\cite{MOGbook,MOG2005}), stems from the fact that it allows us to reduce the logical complexity of terms in a given hypersequent. This has allowed us to structure the proof of the CAN elimination theorem (see Section~\ref{subsec:modal_free_can_elim}) as follows:

 \begin{itemize}
 \item first prove the atomic CAN elimination result (Lemma~\ref{lem:modal_free_atomic_can_elim-simpler}),
\item then use the CAN-free invertibility of the logical rules to reduce the complexity of arbitrary hypersequents and CAN terms to atomic hypersequents and atomic terms. 
\end{itemize}

This general proof technique is, however, not applicable in the context of the system $\hmr$. This is because it is not possible to just invoke the CAN-free invertibility of the $\Diamond$-rule to reduce the complexity of a term of the form $\Diamond A$ in an arbitrary hypersequent due to the very constrained shape of the $\Diamond$-rule  (see Figure~\ref{rules:hmr}) which requires the hypersequent to consist of only one sequent, and forces that only sequent to contain only $1$ terms, $\covar{1}$ terms and $\Diamond$ terms (i.e., terms whose outermost connective is a $\Diamond$)

It is still possible, however, to reduce the logical complexity of terms in arbitrary hypersequents when the outermost connective of these terms is in $\{ 0, +, \times, \sqcup, \sqcap \}$ by applying the invertibility of the associated rule. By systematically applying these simplification steps to a complex hypersequent it is possible to obtain hypersequents having only atoms, $1$ terms, $\covar{1}$ terms or $\Diamond$ terms. These simplified hypersequents are called \emph{basic}.

\begin{defi}[Basic Hypersequent]
A hypersequent $G$ is \emph{basic} if it contains only atoms, $1$ terms, $\covar{1}$ terms or $\Diamond$ terms.
\end{defi}

The following technical result is of key importance.

\begin{thm}[M elimination]
\label{thm:m_elim}
If a hypersequent has a CAN-free derivation, then it has a CAN-free and M-free derivation.
\end{thm}

In the context of the system $\hr$, the M elimination theorem allows for a very simple proof of CAN elimination for atomic CAN terms (Lemma~\ref{lem:modal_free_atomic_can_elim-simpler}). Similarly, in the context of $\hmr$, it will allow for a simple proof of a similar result regarding atomic CAN terms (see Lemma~\ref{lem:atomic_can_elim-simpler}).

 However, compared to $\hr$, where after being useful in proving Lemma~\ref{lem:modal_free_atomic_can_elim-simpler}, the M elimination theorem is not really needed to complete the proof of CAN elimination, in the context of $\hmr$ it appears to be of crucial importance.  As already discussed above,  in the context of $\hmr$, it is not possible to simplify the complexity of CAN terms of the form $\Diamond A$ simply by invoking the CAN-free invertibility of the $\Diamond$ rule. To address this limitation, it is possible to deal with the case of CAN terms being $\Diamond$-terms in a different way, by induction on the structure of the derivation (in the style of the classic inductive proof techniques for eliminating CUT applications in sequent calculi, see, e.g.,~\cite{Buss98}).  In this inductive proof, however,  there is a critically difficult case when the derivation ends with an M rule, as this rule breaks the proviso $\sum\vec{r} = \sum \vec{s}$ of the CAN rule.  For instance, we do not know how to deal with the following instance of the M rule:
\[ \infer[\text{CAN}, 2 + 3 = 2 + 3]{G \sep \vdash \Gamma_1, \Gamma_2}{\infer[\text{M}]{G \sep \vdash \Gamma_1, \Gamma_2, 2.\Diamond A, 3.\Diamond A, 2.\Diamond(\negF{A}), 3.\Diamond (\negF{A})}{G \sep \vdash \Gamma_1, 2.\Diamond A, 3.\Diamond(\negF{A}) & G \sep \vdash \Gamma_2, 3.\Diamond A, 2.\Diamond(\negF{A})}} \]
\noindent
since we cannot use the induction hypothesis on the two premises (because $2\neq 3$). 

The M elimination Theorem~\ref{thm:m_elim} is crucially important in eliminating this difficult case. The rest of the CAN elimination proof can then be carried out without serious technical difficulties. This is our main motivation for proving the M elimination theorem.

We can now state our main theorem regarding the system \hmr{}.

\begin{thm}[CAN elimination]
\label{thm:can_elim}
If a hypersequent $G$ has a derivation, then it has a CAN-free derivation.
\end{thm}

\begin{proof}[Proof sketch] The full proof appears in Section~\ref{subsec:can_elim}. The CAN rule has the following form:
$$\infer[\text{CAN}, \sum \vec{r} = \sum \vec{s}]{G \sep \vdash \Gamma}{G \sep \vdash \Gamma, \vec{s}.A , \vec{r}.\negF{A}} $$

Following the same proof structure as in Theorem~\ref{thm:modal_free_can_elim}, we show how to eliminate one application of the CAN rule. Namely, we prove that if the premise $G \sep \vdash \Gamma, \vec{s}.A , \vec{r}.\negF{A}$, with  $\sum\vec{r} = \sum \vec{s}$, has a CAN-free derivation then the conclusion $G \sep \vdash \Gamma$ also has a CAN-free derivation. This of course implies the statement of the CAN-elimination theorem by using a simple inductive argument on the number of applications of the CAN rule in a derivation.

As in Theorem~\ref{thm:modal_free_can_elim}, it is useful to first invoke the M-elimination Theorem~\ref{thm:m_elim} on the derivation of $G \sep \vdash \Gamma, \vec{s}.A , \vec{r}.\negF{A}$ to remove possible occurrences of the  M rule. This is critical since the M rule is problematic to deal with in our inductive proof because its two premises can generally break the condition $\sum\vec{r} = \sum \vec{s}$.

Hence, in what follows we assume that the derivation of $G \sep \vdash \Gamma, \vec{s}.A , \vec{r}.\negF{A}$ is M-free and the proof proceeds by double induction on the structure of $A$ and the M-free derivation of $G \sep \vdash \Gamma, \vec{s}.A , \vec{r}.\negF{A}$.

The base cases are when $A=x$, i.e., when $A$ is atomic, and when $A=1$. Proving those cases is relatively straightforward, once the critical case regarding the M rule can be ignored.

For the inductive case, when $A$ is a complex term which is not a $\Diamond$ term we invoke the CAN-free invertibility theorem. For example, if $A=B+C$, the invertibility theorem states that $G \sep \vdash \Gamma, \vec{s}.B,  \vec{s}.C, \vec{r}.\negF{B},\vec{r}.\negF{B}$ must also have a CAN-free derivation (and also M-free by application of the M-elimination Theorem~\ref{thm:m_elim}). We then note that, since $B$ and $C$ both have lower complexity than $A$, it follows from two applications of the inductive hypothesis that $G \sep \vdash \Gamma$ has a CAN-free derivation, as desired.

Finally, when $A = \Diamond B$ for some $B$ we prove the result by decreasing the complexity of the derivation while keeping $\Diamond B$ as the CAN term: we use the induction hypothesis on the premises of the last rule used in the derivation --- all cases are straightforward under the hypothesis that the derivation is M-free. This simplification process is repeated until we reach an application of the $\Diamond$ rule, necessarily (due to the constraints of the $\Diamond$ rule) of the form:
\[ \infer[\Diamond]{\vdash \Diamond\Gamma, \vec{r}.\Diamond B,\vec{s}.\Diamond \negF{B}, \vec{r'}.1,\vec{s'}.\covar{1}}{\vdash \Gamma, \vec{r}.B,\vec{s}.\negF{B}, \vec{r'}.1,\vec{s'}.\covar{1}} \]
We can then use the induction hypothesis on the simpler term $B$.
\end{proof}

Finally the algorithm introduced in the proof of Theorem~\ref{thm:modal_free_decidability} can be adapted to the \hmr{} system to prove the following theorem.
\begin{thm}[Decidability]
\label{thm:decidability}
There is an algorithm to decide whether or not a hypersequent has a derivation.
\end{thm}

\subsection{Some technical lemmas}
\label{subsec:tech_lemmas}
All the lemmas presented in Section~\ref{subsec:modal_free_tech_lemmas} in the context of the system \hr\ need to be adapted to the new system \hmr. In most cases, as in for instance Lemma~\ref{lem:modal_free_genT}, the proof is essentially identical and, for this reason, we omit it. In some cases, however,  like Lemma~\ref{lem:modal_free_ext_ID_rule}, the $\Diamond$ rule makes the proof different and more complicated and, for this reason, we discuss how to prove the new difficult aspects of the proof.

We first adapt Lemma~\ref{lem:modal_free_ext_ID_rule} to the system \hmr.

\begin{lem}
\label{lem:ext_ID_rule}
For all $A,\vec{r_i},\vec{s_i}$ such that $\sum \vec{r}_{i} = \sum \vec{s}_{i}$, it holds that:
\begin{center}
  if $\proveM \left [ \vdash \Gamma_i \right ] _{i = 1}^n$ then $\proveM \left [ \vdash \Gamma_i, \vec{r_i}.A,\vec{s_i}.\negF{A} \right] _{i = 1}^n$.
\end{center}
\end{lem}
\begin{proof}
  Let $d$ be a derivation of $\proveM \left [ \vdash \Gamma_i \right ] _{i = 1}^n$. We prove the result by double induction on $(A,d)$. If $A$ is not a $\Diamond$ term, we prove the result as in Lemma~\ref{lem:modal_free_ext_ID_rule} - which decreases the complexity of the term each time. Otherwise $A = \Diamond B$. We prove the result by induction on the derivation $d$. We will only show three cases: the other cases are similar to the $+$ case.
  \begin{itemize}
    \item If $d$ finishes with
      \[\infer[\text{INIT}]{\vdash}{} \]
      then by induction hypothesis on $B$, $\proveM\ \vdash \vec{r_1}.B,\vec{s_1}.\negF{B}$
      so \[\infer[\Diamond]{\vdash \vec{r_1}.\Diamond(B),\vec{s_1}.\Diamond(\negF{B})}{\vdash \vec{r_1}.B,\vec{s_1}.\negF{B}}\]
    \item If $d$ finishes with
      \[\infer[+]{[\vdash \Gamma_i ]_{i=2}^n \sep \vdash \Gamma_1, \vec{s}.(C + D) }{[ \vdash \Gamma_i ]_{i=2}^n \sep \vdash \Gamma_1 , \vec{s}.C , \vec{s}.D} \]
      then by induction hypothesis on the subderivation
      \[ \proveM[\vdash \Gamma_i , \vec{r_i}.\Diamond(B),\vec{s_i}.\Diamond(\negF{B}) ]_{i=2}^n \sep \vdash \Gamma_1, \vec{s}.C,\vec{s}.D,\vec{r_1}.\Diamond(B),\vec{s_1}.\Diamond(\negF{B})\]
      so
      \[\infer[+]{[ \vdash \Gamma_i ,\vec{r_i}.\Diamond(B),\vec{s_i}.\Diamond(\negF{B}) ]_{i=2}^n \sep\vdash \Gamma_1, \vec{s}.(C + D),\vec{r_1}.\Diamond(B),\vec{s_1}.\Diamond(\negF{B}) }{[ \vdash \Gamma_i , \vec{r_i}.\Diamond(B),\vec{s_i}.\Diamond(\negF{B}) ]_{i=2}^n \sep \vdash \Gamma_1 , \vec{s}.C , \vec{s}.D,\vec{r_1}.\Diamond(B),\vec{s_1}.\Diamond(\negF{B})} \]
    \item If $d$ finishes with
      \[ \infer[\Diamond]{\vdash \Diamond \Gamma_1, \vec{r}.1,\vec{s}.\covar{1}}{\vdash \Gamma_1,\vec{r}.1,\vec{s}.\covar{1}} \]
      then by induction hypothesis on $B$
      \[ \proveM\ \vdash \Gamma_1, \vec{r_1}.B,\vec{s_1}.\negF{B},\vec{r}.1,\vec{s}.\covar{1} \]
      so
      \[ \infer[\Diamond]{\vdash \Diamond(\Gamma_1), \vec{r_1}.\Diamond(B),\vec{s_1}.\Diamond(\negF{B}),\vec{r}.1,\vec{s}.\covar{1}}{\vdash \Gamma_1, \vec{r_1}.B,\vec{s_1}.\negF{B},\vec{r}.1,\vec{s}.\covar{1}} \qedhere\]
    \end{itemize}
      
  \end{proof}

The next lemma states that if $G$ is provable then the hypersequent obtained by substituting an atom for a term in $G$ is also provable. 
\begin{lem}
  \label{lem:subst}
  If $\proveM G$ then for all terms $A$, $\proveM G[A/x]$.
\end{lem}
\begin{proof}
  Similar to the proof of Lemma~\ref{lem:modal_free_subst}.
\end{proof}

The following lemma, which will be useful in the proof of the completeness theorem, states that the rules $\{ 0, +, \times, \sqcup, \sqcap \}$, are invertible in \hmr, in the sense that if the conclusion of one of these rules is derivable (possibly using CAN rules) then its premises are also derivable (possibly using CAN rules). 

\begin{lem}
  \label{lem:can_full_invertibility}
  All the logical rules $\{ 0, +, \times, \sqcup, \sqcap \}$ are invertible.
\end{lem}
\begin{proof}
Similar to Lemma~\ref{lem:modal_free_can_full_invertibility}.
\end{proof}

\begin{rem}
  \label{rem:can_full_invertibility_no_T}
  The proof of Lemma~\ref{lem:can_full_invertibility} does not introduce any new T rule, so if the conclusion of one of the logical rules $\{ 0, +, \times, \sqcup, \sqcap \}$ has a T-free derivation, then the premises also have T-free derivations.
\end{rem}

The next lemmas state that CAN-free derivability in the \hmr{} system is preserved by scalar multiplication.
      
\begin{lem}
  \label{lem:genT}
  Let $\vec{r}\in\mathbb{R}_{>0}$ be a non-empty vector and $G$ a hypersequent.
  If $\proveMNC G \sep \vdash \vec{r}.\Gamma$ then $\proveMNC G \sep \vdash \Gamma$.
\end{lem}
\begin{proof}
  Similar to Lemma~\ref{lem:modal_free_genT}.
\end{proof}

\begin{lem}
  \label{lem:genT2}
  Let $\vec{r}\in\mathbb{R}_{>0}$ be a vector and $G$ a hypersequent.
  If $\proveMNC G \sep \vdash \Gamma$ then $\proveMNC G \sep \vdash \vec{r}.\Gamma$.
\end{lem}
\begin{proof}
  Similar to Lemma~\ref{lem:modal_free_genT2}.
\end{proof}

The above lemmas have two useful corollaries.

\begin{cor}
  \label{cor:andAlphaBeta}
    If $\proveMNC G \sep \vdash \Gamma,\vec{r}.A, \vec{s}.A$ and $\proveMNC G \sep \vdash \Gamma, \vec{r}.B, \vec{s}.B$ then $\proveMNC G \sep \vdash \Gamma, \vec{r}.A, \vec{s}.B$.
\end{cor}
\begin{proof}
  Similar to Corollary~\ref{cor:modal_free_andAlphaBeta}
\end{proof}

\begin{cor}
  \label{cor:orAlphaBeta}
If $\proveMNC G \sep \vdash \vec{r}.A,\vec{s}.A, \Gamma \sep \vdash \vec{r}.B,\vec{s}.B , \Gamma \sep \vdash \vec{r}.A,\vec{s}.B,\Gamma$, then $\proveMNC G \sep \vdash \vec{r}.A,\vec{s}.A, \Gamma \sep \vdash \vec{r}.B,\vec{s}.B, \Gamma$.
\end{cor}
\begin{proof}
  Similar to Corollary~\ref{cor:modal_free_orAlphaBeta}.
\end{proof}


\subsection{Soundness -- Proof of Theorem~\ref{thm:soundness}}
  The proof is similar to the one for Theorem~\ref{thm:modal_free_soundness}: we prove that every rules is sound and then conclude by induction on the complexity of derivations. Since the proofs of soundness for the rules already in \hr{} are exactly the same as for the system \hr\ done in Section~\ref{subsec:modal_free_soundness}, we only prove the soundness of the new rules.
  \begin{itemize}
  \item For the rule
  \[ \infer[\Diamond, \sum \vec{r} \geq \sum \vec{s}]{\vdash \Diamond\Gamma, \vec{r}.1,\vec{s}.\covar{1}}{\vdash \Gamma, \vec{r}.1,\vec{s}.\covar{1}} \]
  the hypothesis is $\sem{ \vdash \Gamma, \vec{r}.1,\vec{s}.\covar{1} } \geq 0$ so
  \begin{eqnarray*}
  \sem{\vdash \Diamond\Gamma, \vec{r}.1,\vec{s}.\covar{1}} & = & \sem{ \vdash \Diamond\Gamma, (\sum \vec{r} - \sum \vec{s}).1 } \text{ by distributivity} \\
  & \geq & \sem{ \vdash \Diamond\Gamma, (\sum \vec{r} - \sum \vec{s}).\Diamond 1 } \text{ since } \Diamond 1 \leq 1 \\
  & = & \Diamond(\sem{ \vdash \Gamma, (\sum \vec{r} - \sum \vec{s}).1 }) \text{ by linearity of }\Diamond \\
  & = & \Diamond(\sem{ \vdash \Gamma, \vec{r}.1, \vec{s}.\covar{1})} \text{ by distributivity} \\
  & \geq & 0 \text{ by the hypothesis and the monotonicity of }\Diamond .
  \end{eqnarray*}
  \item For the rule
  \[ \infer[1, \sum \vec{r} \geq \sum \vec{s}]{G \sep \vdash \Gamma, \vec{r}.1, \vec{s}.\covar{1}}{G \sep \vdash \Gamma} \]
  the hypothesis is $\sem{ G \sep \vdash \Gamma } \geq 0$ so
  \begin{eqnarray*}
  \sem{ G \sep \vdash \Gamma, \vec{r}.1, \vec{s}.\covar{1} } & \geq & \sem{ G \sep \vdash \Gamma } \text{ since } \sum \vec{r} \geq \sum \vec{s} \text{ and } 0 \leq 1\\
  & \geq & 0
  \end{eqnarray*}
  \end{itemize}

  \subsection{Completeness -- Proof of Theorem~\ref{thm:completeness}}
  The proof follows the same pattern as in Section~\ref{subsec:modal_free_completeness}: we first prove a similar result that admits a simple proof by induction on the derivation of $\mathcal{A}_{\text{Riesz}}^\Diamond \vdash A = B$ and then we use it and the invertibility of the logical rules (Lemma~\ref{lem:can_full_invertibility}) to prove Theorem~\ref{thm:completeness}, as shown in Section~\ref{subsec:modal_free_completeness}.

  \begin{lem}
    \label{lem:completeness_aux}
    If $\mathcal{A}_{\text{Riesz}}^\Diamond \vdash A = B$ then $\vdash r.A, r.\negF{B}$ and $\vdash r.B, r.\negF{A}$ are provable in $\hmr$ for all $r > 0$.
  \end{lem}
  \begin{proof}
    Since the other cases are proven in the exact same way as in Theorem~\ref{thm:modal_free_completeness}, we will only derive the new axioms.
    \begin{itemize}
    \item For the axiom $0 \leq 1$.
      \[ \infer[0]{\vdash r.(0\sqcap 1), r.0}{\infer[\sqcap]{\vdash r.(0\sqcap 1)}{\infer[0]{\vdash r.0}{\infer[\text{INIT}]{\vdash}{}} & \infer[1, r \geq 0]{\vdash r.1}{\infer[\text{INIT}]{\vdash}{}}}} \]
      and
      \[ \infer[0]{\vdash r.0, r.(0\sqcup\covar{1})}{\infer[\sqcup]{\vdash r.(0\sqcup\covar{1})}{\infer[\text{W}]{\vdash r.0 \sep \vdash r.\covar{1}}{\infer[0]{\vdash r.0}{\infer[\text{INIT}]{\vdash}{}}}}}  \]
    \item For the axiom $\Diamond(1) \leq 1$.
      \[ \infer[\sqcap]{\vdash r.(\Diamond(1)\sqcap 1),r.\Diamond(\covar{1})}{\infer[\Diamond]{\vdash r.\Diamond(1), r.\Diamond(\covar{1})}{\infer[1]{\vdash r.1, r.\covar{1}}{\infer[\text{INIT}]{\vdash}{}}} & \infer[\Diamond]{\vdash r.1, r.\Diamond(\covar{1})}{\infer[1]{\vdash r.1, r.\covar{1}}{\infer[\text{INIT}]{\vdash}{}}}} \]
      and
      \[ \infer[\sqcup]{\vdash r.\Diamond(1), r.(\Diamond(\covar{1}) \sqcup \covar{1})}{\infer[\text{W}]{\vdash r.\Diamond(1), r.\Diamond(\covar{1}) \sep \vdash r.\Diamond(1), r.\covar{1}}{\infer[\Diamond]{\vdash r.\Diamond(1), r.\Diamond(\covar{1})}{\infer[1]{\vdash r.1, r.\covar{1}}{\infer[\text{INIT}]{\vdash}{}}}}} \]
  \item For the axiom $\Diamond(r_1x + r_2y) = r_1\Diamond(x) + r_2\Diamond(y)$.
    \[ \infer[+]{\vdash r.\Diamond(r_1x + r_2y), r.(r_1\Diamond(\covar{x}) + r_2\Diamond(\covar{y}))}{\infer[\times^2]{\vdash r.\Diamond(r_1x + r_2y), r.r_1\Diamond(\covar{x}), r.r_2\Diamond(\covar{y})}{\infer[\Diamond]{\vdash r.\Diamond(r_1x + r_2y), r_1r.\Diamond(\covar{x}), r_2r.\Diamond(\covar{y})}{\infer[+]{\vdash r.(r_1x + r_2y), r_1r.\covar{x}, r_2r.\covar{y}}{\infer[\times^2]{\vdash r.r_1x, r.r_2y, r_1r.\covar{x}, r_2r.\covar{y}}{\infer[\text{ID}^2]{\vdash r_1r.x, r_2r.y, r_1r.\covar{x}, r_2r.\covar{y}}{\infer[\text{INIT}]{\vdash}{}}}}}}} \]
    and
    \[ \infer[+]{\vdash r.(r_1\Diamond(x) + r_2\Diamond(y)), r.\Diamond(r_1\covar{x} + r_2\covar{y})}{\infer[\times^2]{\vdash r.r_1\Diamond(x) , r.r_2\Diamond(y), r.\Diamond(r_1\covar{x} + r_2\covar{y})}{\infer[\Diamond]{\vdash r_1r.\Diamond(x) , r_2r.\Diamond(y), r.\Diamond(r_1\covar{x} + r_2\covar{y})}{\infer[+]{\vdash r_1r.x, r_2r.y, r.(r_1\covar{x} + r_2\covar{y})}{\infer[\times^2]{\vdash r_1r.x, r_2r.y, r.r_1\covar{x}, r.r_2\covar{y}}{\infer[\text{ID}^2]{\vdash r_1r.x, r_2r.y, r_1r.\covar{x}, r_2r.\covar{y}}{\infer[\text{INIT}]{\vdash}{}}}}}}} \]
  \item For the axiom $0 \leq \Diamond (0 \sqcup x)$.
    \[ \infer[0]{\vdash r.(0\sqcap \Diamond(0 \sqcup x)), r.0}{\infer[\sqcap]{\vdash r.(0\sqcap \Diamond(0 \sqcup x))}{\infer[0]{\vdash r.0}{\infer[\text{INIT}]{\vdash}{}} & \infer[\Diamond]{\vdash r.\Diamond(0\sqcup x)}{\infer[\sqcup]{\vdash r.(0 \sqcup x)}{\infer[\text{W}]{\vdash r.0 \sep \vdash r.x}{\infer[0]{\vdash r.0}{\infer[\text{INIT}]{\vdash}{}}}}}}} \]
  and
  \[ \infer[0]{\vdash r.0 , r.(0 \sqcup \Diamond (0 \sqcap \covar{x}))}{\infer[\sqcup]{\vdash r.(0\sqcup \Diamond (0\sqcap \covar{x}))}{\infer[\text{W}]{\vdash r.0 \sep \vdash r.\Diamond(0 \sqcap \covar{x})}{\infer[0]{\vdash r.0}{\infer[\text{INIT}]{\vdash}{}}}}} \qedhere\]
\end{itemize}
\end{proof}

\begin{rem}
By inspecting the proof of Lemma~\ref{lem:completeness_aux} it is possible to verify that the T rule is never used in the construction of $\proveM G$. This, together with the similar Remark~\ref{rem:can_full_invertibility_no_T} regarding the Lemma~\ref{lem:can_full_invertibility}, implies that the T rule is never used in the proof of the completeness Theorem~\ref{thm:completeness}. From this we get the following corollary.
\end{rem}

\begin{cor}
The T rule is admissible in the system \hmr.
\end{cor}

As in the case of the system \hr\ (see Lemma~\ref{lem:modal_free_not_complete}) there is no hope of eliminating both the T rule and the CAN rule from the \hmr\ system. 
\begin{lem}
  \label{lem:not_complete}
Let $r_1$ and $r_2$ be two irrational numbers that are incommensurable over $\mathbb{Q}$ (so there is no $q\in\mathbb{Q}$ such that $qr_1 = r_2$). Then the atomic hypersequent $G$
$$\vdash r_1.x \sep \vdash r_2.\covar{x}$$ 
does not have a CAN-free and T-free derivation.
\end{lem}
\begin{proof}
The proof is similar to that of Lemma~\ref{lem:modal_free_not_complete} but Lemma~\ref{lem:int_lambda_prop} below takes the place of Lemma ~\ref{lem:modal_free_int_lambda_prop}.
\end{proof}

\begin{lem}
  \label{lem:int_lambda_prop}
  For all basic hypersequents $G$, built using the variables and negated variables $x_1, \covar{x_1}, \dots, x_{k}, \covar{x_{k}}$, of the form
  \[\vdash \Gamma_1, \Diamond \Delta_1,\vec{r'}_1.1,\vec{s'}_1.\covar{1} \sep ... \sep \vdash \Gamma_m, \Diamond \Delta_m,\vec{r'}_m.1,\vec{s'}_m.\covar{1}\]
where $\Gamma_i =  \vec{r}_{i,1}.x_1,...,\vec{r}_{i,k}.x_k, \vec{s}_{i,1}.\covar{x_1},...,\vec{s}_{i,k}.\covar{x_k}$, the following are equivalent:
  
  \begin{enumerate}
  \item $G$ has a CAN-free and T-free derivation.
  \item there exist natural numbers $n_1,...,n_m \in \mathbb{N}$, one for each sequent  in $G$, such that:
  \begin{itemize}
  \item there exists $i\in [1..m]$ such that $n_i \neq 0$, i.e., the numbers are not all $0$'s, and
  \item for every variable and covariable $(x_j, \covar{x_j})$ pair, it holds that
    $$
    \sum^{m}_{i=1} n_i (\sum \vec{r}_{i,j}) =   \sum^{m}_{i=1} n_i (\sum \vec{s}_{i,j}) 
    $$
    i.e., the scaled (by the numbers $n_1$ \dots $n_m$) sum of the coefficients in front of the variable $x_j$ is equal to the scaled sum of the coefficients in from of the covariable $\covar{x_j}$, and
  \item $\sum_{i=1}^m n_i \sum \vec{s'}_i\leq \sum_{i=1}^m n_i \sum \vec{r'}_i$, i.e., there are more $1$ than $\covar{1}$, and
  \item the hypersequent consisting of only one sequent $$\vdash \Delta_1^{n_1},...,\Delta_m^{n_m},(\vec{r'}_1.1)^{n_1},...,(\vec{r'}_m.1)^{n_m}, (\vec{s'}_1.\covar{1})^{n_1},...,(\vec{s'}_m.\covar{1})^{n_m}$$
has a CAN-free and T-free derivation,  where the notation $\Gamma^n$ means $\underbrace{\Gamma,...,\Gamma}_{n\text{ times}}$.  \end{itemize}
  \end{enumerate}
\end{lem}
\begin{proof}
  We prove $(1) \Rightarrow (2)$ by induction on the derivation of $G$. We show only the M case, the other cases being simple. We write $\Gamma'_i$ for $\Gamma_i, \Diamond \Delta_i,\vec{r'}_i.1,\vec{s'}_i.\covar{1}$.
  \begin{itemize}  
  \item If the derivation finishes with
    \[ \infer[\text{M}]{\vdash \Gamma'_1 \sep ... \sep \vdash \Gamma'_m , \Gamma'_{m+1}}{\vdash \Gamma'_1 \sep ... \sep \vdash \Gamma'_m & \vdash \Gamma'_1 \sep ... \sep \vdash \Gamma'_{m+1}} \]
    by induction hypothesis, there are $n_1,...,n_m \in \mathbb{N}$ such that :
    \begin{itemize}
    \item there exists $i\in [1..m]$ such that $n_i \neq 0$.
    \item for every variable and covariable $(x_j,\covar{x_j})$ pair, it holds that $\sum_i n_i.\sum \vec{r}_{i,j} = \sum_i n_i.\sum \vec{s}_{i,j}$.
    \item $\sum_{i=1}^m n_i\sum \vec{s'}_i\leq \sum_{i=1}^m n_i\sum \vec{r'}_i$.
    \item $\vdash \Delta_1^{n_1},...,\Delta_m^{n_m},(\vec{r'}_1.1)^{n_1},...,(\vec{r'}_m.1)^{n_m}, (\vec{s'}_1.\covar{1})^{n_1},...,(\vec{s'}_m.\covar{1})^{n_m}$ has a CAN-free and T-free derivation.
    \end{itemize}
    and $n'_1,...,n'_m \in \mathbb{N}$ such that :
    \begin{itemize}
    \item there exists $i\in [1..m]$ such that $n'_i \neq 0$.
    \item for every variable and covariable $(x_j,\covar{x_j})$ pair, it holds that $$\sum_{i=0}^{m-1} n'_i.\sum \vec{r}_{i,j} + n'_m.\sum \vec{r}_{m+1,j}= \sum_{i=0}^{m-1} n'_i.\sum \vec{s}_{i,j} + n'_m.\sum \vec{s}_{m+1,j}$$
    \item $\sum_{i=1}^{m-1} n'_i\sum \vec{s'}_i + n'_m\sum\vec{s'}_{m+1}\leq \sum_{i=1}^{m-1} n'_i\sum \vec{r'}_i + n'_m\sum\vec{r'}_{m+1}$.
    \item $\vdash \Delta_1^{n'_1},...,\Delta_{m+1}^{n'_m},(\vec{r'}_1.1)^{n'_1},...,(\vec{r'}_{m+1}.1)^{n'_m}, (\vec{s'}_1.\covar{1})^{n'_1},...,(\vec{s'}_{m+1}.\covar{1})^{n_m}$ has a CAN-free and T-free derivation.
    \end{itemize}
    If $n_m=0$ then $n_1,...,n_{m-1}, 0$ satisfies the property.\\
    Otherwise if $n'_m = 0$ then $n'_1,...,n'_{m-1},0$ satisfies the property.\\
    Otherwise, $n_m.n'_1 + n'_m.n_1, n_m.n'_2 + n'_m.n_2,...,n_m.n'_{m-1} + n'_m.n_{m-1},n_m.n'_m$ satisfies the property.
  \end{itemize}
  The other way ($(2) \Rightarrow (1)$) is more straightforward. If there exist natural numbers $n_1,...,n_m \in \mathbb{N}$, one for each sequent  in $G$, such that:
  \begin{itemize}
  \item there exists $i\in [1..m]$ such that $n_i \neq 0$ and
  \item for every variable and covariable $(x_j, \covar{x_j})$ pair, it holds that
    $$
    \sum^{m}_{i=1} n_i (\sum \vec{r}_{i,j}) =   \sum^{m}_{i=1} n_i (\sum \vec{s}_{i,j}) 
    $$
    and
    \item $\sum_{i=1}^m n_i\sum \vec{s'}_i\leq \sum_{i=1}^m n_i\sum \vec{r'}_i$.and
    \item $\vdash \Delta_1^{n_1},...,\Delta_m^{n_m},(\vec{r'}_1.1)^{n_1},...,(\vec{r'}_m.1)^{n_m}, (\vec{s'}_1.\covar{1})^{n_1},...,(\vec{s'}_m.\covar{1})^{n_m}$ has a CAN-free and T-free derivation.
  \end{itemize}
  then we can use the W rule to remove the sequents corresponding to the numbers $n_i = 0$, and use the C rule $n_i-1$ times then the S rule $n_i-1$ times on the $i$th sequent to multiply it by $n_i$. If we assume that there is a natural number $l$ such that $n_i = 0$ for all $i > l$ and $n_i \neq 0$ for all $i \leq l$, then the CAN-free T-free derivation is:
  \[ \infer[\text{W}^*]{\vdash \Gamma_1, \Diamond \Delta_1,\vec{r'}_1.1,\vec{s'}_1.\covar{1} \sep ... \sep \vdash \Gamma_m, \Diamond \Delta_m,\vec{r'}_m.1,\vec{s'}_m.\covar{1}}{\infer[\text{C-S}^*]{\vdash \Gamma_1, \Diamond \Delta_1,\vec{r'}_1.1,\vec{s'}_1.\covar{1} \sep ... \sep \vdash \Gamma_l, \Diamond \Delta_l,\vec{r'}_l.1,\vec{s'}_l.\covar{1}}{\infer[\text{S}^*]{\vdash (\Gamma_1, \Diamond\Delta_1, \vec{r'}_1.1, \vec{s'}_1.\covar{1})^{n_1} \sep \dots \sep \vdash (\Gamma_l, \Diamond \Delta_l, \vec{r'}_l.1, \vec{s'}_l.\covar{1})^{n_l}}{\infer[\text{ID}^*]{\vdash (\Gamma_1, \Diamond\Delta_1, \vec{r'}_1.1, \vec{s'}_1.\covar{1})^{n_1} ,\dots, (\Gamma_l, \Diamond \Delta_l, \vec{r'}_l.1, \vec{s'}_l.\covar{1})^{n_l}}{\infer[\Diamond]{\vdash (\Diamond\Delta_1, \vec{r'}_1.1, \vec{s'}_1.\covar{1})^{n_1} ,\dots,(\Diamond \Delta_l, \vec{r'}_l.1, \vec{s'}_l.\covar{1})^{n_l}}{\vdash (\Delta_1, \vec{r'}_1.1, \vec{s'}_1.\covar{1})^{n_1} ,\dots,(\Delta_l, \vec{r'}_l.1, \vec{s'}_l.\covar{1})^{n_l}}}}}} \]
  and we can conclude using the last condition that states that $$\vdash (\Delta_1, \vec{r'}_1.1, \vec{s'}_1.\covar{1})^{n_1} ,\dots,(\Delta_l, \vec{r'}_l.1, \vec{s'}_l.\covar{1})^{n_l}$$ has a CAN-free and T-free derivation.
\end{proof}


\subsection{CAN-free Invertibility -- Proof of Theorem~\ref{thm:invertibility} }
\label{subsec:invertibility}

The proofs presented in this section follow the same pattern of those in Section~\ref{subsec:modal_free_invertibility}: we will prove the CAN-free invertibility of more general rules. The generalised non-modal rules are the same as those in Figure~\ref{rules:generalised-logical rules} from Section~\ref{subsec:modal_free_invertibility} and the generalised $\Diamond$ rule  has the following shape:
\[ \infer{[\vdash \Diamond \Gamma_i, \vec{r}_i.1, \vec{s}_i.\covar{1}]_{i=1}^n}{[\vdash \Gamma_i, \vec{r}_i.1, \vec{s}_i.\covar{1}]_{i=1}^n} \]

\begin{rem}
  The generalized $\Diamond$ rule is unsound, the hypersequent $\vdash 1.\Diamond( \covar{x} \sqcap \covar{y}), 1.\Diamond(x) \sqcup \Diamond(y)$ is derivable using this rule (see Remark~\ref{remark:notsound}, a similar derivation can be used to derive the hypersequent). Yet, even if the generalized $\Diamond$ rule is not sound, it still enjoys CAN-free invertibility.
\end{rem}

We will prove that those rules are CAN-free invertible by induction on the derivation of the conclusion. The proof steps dealing with the rules already present in \hr\ are the same as in Section~\ref{subsec:modal_free_invertibility}. In what follows we just show the details of the proof steps associated with the new cases associated with the $\Diamond$-rule and $1$-rule of \hmr.

\begin{lem}
  If $[ \vdash \Gamma_i , \vec{r}_i. (A \sqcup B)]_{i=1}^n$ has a CAN-free derivation then $[ \vdash \Gamma_i , \vec{r}_i. A \sep \vdash \Gamma_i , \vec{r}_i. B ]_{i=1}^n$ has a CAN-free derivation.
\end{lem}
\begin{proof}
  By induction on the derivation.
  \begin{itemize}
  \item If the derivation finishes with
    \[\infer[1]{[\vdash \Gamma_i , \vec{r}_i. (A \sqcup B) ]_{i=2}^n \sep \vdash \Gamma_1 , \vec{r}_1. (A \sqcup B) , \vec{r}.1,\vec{s}.\covar{1}}{[ \vdash \Gamma_i , \vec{r}_i. (A \sqcup B) ]_{i=2}^n \sep \vdash \Gamma_1, \vec{r}_1. (A \sqcup B)} \]
    then by induction hypothesis on the CAN-free derivation of the premise we have that
    \[\proveMNC [ \vdash \Gamma_i , \vec{r}_i. A \sep \vdash \Gamma_i,\vec{r}_i. B ]_{i=2}^n \sep \vdash \Gamma_1,\vec{r}_1. A \sep \vdash \Gamma_1,\vec{r}_1. B \]
    so
    \[\infer[1^*]{ G' \sep \vdash \Gamma_1,\vec{r}_1. A,\vec{r}.1,\vec{s}.\covar{1} \sep \vdash \Gamma_1,\vec{r}_1. B,\vec{r}.1,\vec{s}.\covar{1}}{ G' \sep \vdash \Gamma_1,\vec{r}_1. A \sep \vdash \Gamma_1,\vec{r}_1. B} \]
    with $G' = [ \vdash \Gamma_i , \vec{r}_i. A \sep \vdash \Gamma_i,\vec{r}_i. B ]_{i=2}^n$
  \item If the derivation finishes with an application of the $\Diamond$ rule, the shape of the conclusion is
    \[ \vdash \Diamond \Gamma_1, \vec{r}.1,\vec{s}.\covar{1} \]
    with $\vec{r_1} = \emptyset$ so the hypersequent \[ \vdash \Diamond \Gamma_1,\vec{r_1}.A, \vec{r}.1,\vec{s}.\covar{1} \sep \vdash \Diamond \Gamma_1,\vec{r_1}.B, \vec{r}.1,\vec{s}.\covar{1} =\ \vdash \Diamond \Gamma_1, \vec{r}.1,\vec{s}.\covar{1} \sep \vdash \Diamond \Gamma_1, \vec{r}.1,\vec{s}.\covar{1} \] is CAN-free derivable using the C rule.\qedhere
  \end{itemize}
\end{proof}

\begin{lem}
  If $[ \vdash \Gamma_i , \vec{r}_i (A + B) ]_{i=1}^n$ has a CAN-free derivation then $[ \vdash \Gamma_i , \vec{r}_i A,\vec{r}_i B ]_{i=1}^n$ has a CAN-free derivation.
\end{lem}
\begin{proof}
  By induction on the derivation.
  \begin{itemize}
  \item If the derivation finishes with
    \[\infer[1]{[\vdash \Gamma_i , \vec{r}_i. (A + B) ]_{i=2}^n \sep \vdash \Gamma_1 , \vec{r}_1. (A + B) , \vec{r}.1,\vec{s}.\covar{1}}{[ \vdash \Gamma_i , \vec{r}_i. (A + B) ]_{i=2}^n \sep \vdash \Gamma_1, \vec{r}_1. (A + B)} \]
    then by induction hypothesis on the CAN-free derivation of the premise we have that
    \[\proveMNC [ \vdash \Gamma_i , \vec{r}_i. A,\vec{r}_i. B ]_{i=2}^n \sep \vdash \Gamma_1,\vec{r}_1. A,\vec{r}_1. B \]
    so
    \[\infer[1]{[ \vdash \Gamma_i , \vec{r}_i. A,\vec{r}_i. B]_{i=2}^n \sep \vdash \Gamma_1,\vec{r}_1. A,\vec{r}_1. B ,\vec{r}.1,\vec{s}.\covar{1}}{[ \vdash \Gamma_i , \vec{r}_i. A,\vec{r}_i. B ]_{i=2}^n \sep \vdash \Gamma_1,\vec{r}_1. A,\vec{r}_1. B} \]
  \item If the derivation finishes with an application of the $\Diamond$ rule, the shape of the conclusion is
    \[ \vdash \Diamond \Gamma_1, \vec{r}.1,\vec{s}.\covar{1} \]
    with $\vec{r_1} = \emptyset$ so the hypersequent $ \vdash \Diamond \Gamma_1,\vec{r_1}.A,\vec{r_1}.B, \vec{r}.1,\vec{s}.\covar{1} =\ \vdash \Diamond \Gamma_1, \vec{r}.1,\vec{s}.\covar{1} $ is derivable.\qedhere
  \end{itemize}
\end{proof}

\begin{lem}
  If $\vdash \Gamma_i , \vec{r}_i. (A \sqcap B)  ]_{i=1}^n$ has a CAN-free derivation then $[ \vdash \Gamma_i , \vec{r}_i. A ]_{i=1}^n$ and $[\vdash \Gamma_i , \vec{r}_i. B ]_{i=1}^n$ have CAN-free derivations.
\end{lem}
\begin{proof}
  By induction on the derivation. We will only show that $\proveMNC [ \vdash \Gamma_i , \vec{r}_i. A ]_{i=1}^n$, the other case is similar.
  \begin{itemize}
  \item If the derivation finishes with
    \[\infer[1]{[\vdash \Gamma_i , \vec{r}_i. (A \sqcap B) ]_{i=2}^n \sep \vdash \Gamma_1 , \vec{r}_1. (A \sqcap B) , \vec{r}.1,\vec{s}.\covar{1}}{[ \vdash \Gamma_i , \vec{r}_i. (A \sqcap B) ]_{i=2}^n \sep \vdash \Gamma_1, \vec{r}_1. (A \sqcap B)} \]
    then by induction hypothesis on the CAN-free derivation of the premise we have that
    \[\proveMNC [ \vdash \Gamma_i , \vec{r}_i. A ]_{i=2}^n \sep \vdash \Gamma_1,\vec{r}_1. A \]
    so
    \[\infer[1]{[ \vdash \Gamma_i , \vec{r}_i. A]_{i=2}^n \sep \vdash \Gamma_1,\vec{r}_1. A ,\vec{r}.1,\vec{s}.\covar{1}}{[ \vdash \Gamma_i , \vec{r}_i. A ]_{i=2}^n \sep \vdash \Gamma_1,\vec{r}_1. A} \]
  \item If the derivation finishes with an application of the $\Diamond$ rule, the shape of the conclusion is
    \[ \vdash \Diamond \Gamma_1, \vec{r}.1,\vec{s}.\covar{1} \]
    with $\vec{r_1} = \emptyset$ so the hypersequent $ \vdash \Diamond \Gamma_1,\vec{r_1}.A, \vec{r}.1,\vec{s}.\covar{1} =\ \vdash \Diamond \Gamma_1, \vec{r}.1,\vec{s}.\covar{1} $ is derivable.\qedhere
  \end{itemize}
\end{proof}

\begin{lem}
  If $[\vdash \Diamond \Gamma_i, \vec{r}_i.1, \vec{s}_i.\covar{1}]_{i=1}^n$ has a CAN-free derivation then $[\vdash \Gamma_i, \vec{r}_i.1, \vec{s}_i.\covar{1}]_{i=1}^n$ has a CAN-free derivation.
\end{lem}
\begin{proof}
  By induction on the derivation. Since the hypersequent under consideration is basic, we do not need to deal with any logical rule beside the $\Diamond$-rule, which leads immediately to the desired result, and the $1$-rule. The cases regarding the structural rules and the $1$-rule  are very simple. For instance, if the derivation finishes with the W rule:
      \[ \infer[\text{W}]{ [\vdash \Diamond\Gamma_i, \vec{r}_i.1,\vec{s}_i.\covar{1}]_{i=2}^n \sep \vdash \Diamond\Gamma_1, \vec{r}_1.1,\vec{s}_1.\covar{1}}{ [\vdash \Diamond\Gamma_i, \vec{r}_i.1,\vec{s}_i.\covar{1}]_{i=2}^n} \]
      then by induction hypothesis
      \[ \proveMNC  [\vdash \Gamma_i, \vec{r}_i.1,\vec{s}_i.\covar{1}]_{i=2}^n \]
      so
      \[ \infer[\text{W}]{ [\vdash \Gamma_i, \vec{r}_i.1,\vec{s}_i.\covar{1}]_{i=2}^n \sep \vdash \Gamma_1, \vec{r}_1.1,\vec{s}_1.\covar{1}}{ [\vdash \Gamma_i, \vec{r}_i.1,\vec{s}_i.\covar{1}]_{i=2}^n} \]
\end{proof}


\subsection{M-elimination -- Proof of Theorem~\ref{thm:m_elim}}\label{subsec:m_elim}

Following the same pattern of Section~\ref{subsec:modal_free_m_elim}, we need to show that for each hypersequent $G$ and sequents $\Gamma$ and $\Delta$, if there exist CAN-free and M-free derivations $d_1$ of $G \sep \vdash \Gamma$ and $d_2$ of $G \sep \vdash \Delta$, then there also exists a CAN-free and M-free derivation of  $G \sep \vdash \Gamma,\Delta$.

The general idea presented in Section~\ref{subsec:modal_free_m_elim} is to combine the derivations $d_1$ and $d_2$ in a sequential way,  first constructing a prederivation $d_1^\prime$ of $G \sep G \sep \vdash \Gamma, \Delta$ (using $d_1$) whose leaves are either axioms or hypersequents of the form $G\sep \vdash \vec{r}.\Delta$, and then by completing this prederivation into a derivation (using $d_2$). Finally, $G \sep G \sep \vdash \Gamma, \Delta$ can be easily turned into a derivation of $G \sep \vdash \Gamma,\Delta$ as desired.

 However, this technique cannot be directly applied in the context of the system \hmr\ due to the constraints imposed on the shape of the hypersequent by the $\Diamond$ rule. Indeed an application of the $\Diamond$ rule in $d_1$ acting on some hypersequent of the form  \[ \vdash \Diamond \Gamma', \vec{s}.1, \vec{t}.\covar{1}\] cannot turned into an application of the $\Diamond$ rule on
\[ G \sep \vdash \vec{r}.\Delta,  \Diamond \Gamma', \vec{s}.1, \vec{t}.\covar{1}\] because this hypersequent cannot be the conclusion of a $\Diamond$ rule as it does not satisfy the constraints. To deal with the $\Diamond$ rule, we will expand the construction of Section~\ref{subsec:modal_free_m_elim} by induction on the modal depth of the derivation $d_1$.

Indeed, when constructing the prederivation  $d'_1$ inductively from $d_1$, we stop at the applications of the $\Diamond$-rule. Hence, the inductive procedure takes the derivation $d_1$ and produces a CAN-free and M-free prederivation $d_1^\prime$ of
\[ G \sep G \sep \vdash \Gamma, \Delta \]
where all the leaves in the prederivation are either:

\begin{enumerate}
\item terminated, or
\item non-terminated and having the shape
\[ G\sep \vdash \vec{r}.\Delta \]
which can then be completed using the derivation $d_2$ in the exact same way explained in Section~\ref{subsec:modal_free_m_elim}, or
\item non-terminated and having the shape:
\[ G\sep \vdash \Diamond \Gamma', \vec{r}.\Delta, \vec{s}.1,\vec{t}.\covar{1} \]
for some sequent $\Gamma'$ and vectors $\vec{r}, \vec{s},\vec{t}$. For each of these leaves there is a corresponding derivation of

\begin{equation}
\label{eq_subproof_1}
  \vdash \Gamma', \vec{s}.1, \vec{t}.\covar{1}
  \end{equation}
\end{enumerate}

\noindent
To obtain derivations of the leaves of $d_1^\prime$ of the third type, and thus complete the derivation, we proceed as follows. First, we use the derivation $d_2$ to construct a CAN-free and M-free derivation $d_{2,\vec{r}}$ 
\[ G \sep \vec{r}.\Delta \]
for each vector of scalars $\vec{r}$ in the leaves. We then modify each derivation $d_{2,\vec{r}}$ into a prederivation $d_2^\prime$ of
\[ G \sep \vdash \Diamond \Gamma', \vec{r}.\Delta, \vec{s}.1,\vec{t}.\covar{1}\]
using the exact same inductive procedure (which stops when reaching applications of $\Diamond$ terms) introduced above for producing $d_1^\prime$ from $d_1$.  Note that in this case, the leaves of the third kind in $d_2^\prime$  are of the form:
\[ \vdash \vec{r'}.(\Diamond \Gamma',\vec{s}.1,\vec{t}.\covar{1}), \Diamond \Delta', \vec{s'}.1,\vec{t'}.\covar{1} \]
and have associated derivations of 
\begin{equation}
\label{eq_subproof_2}
  \vdash \Delta', \vec{s'}.1, \vec{t'}.\covar{1}
  \end{equation}
Therefore, we can legitimately apply the $\Diamond$ rule (Lemma~\ref{lemma:m-elim-first-step} below ensures that the proviso of the rule is respected) and reduce these leaves to leaves of the form 
\[ \vdash \vec{r'}.(\Gamma',\vec{s}.1,\vec{t}.\covar{1}),  \Delta', \vec{s'}.1,\vec{t'}.\covar{1} \]
which, importantly, have a lower modal depth compared to the conclusion $G \sep \vdash \Gamma$ of the derivation $d_1$  we started with above. 

In order the produce a derivation for the leaves $\vdash \vec{r'}.(\Gamma',\vec{s}.1,\vec{t}.\covar{1}),  \Delta', \vec{s'}.1,\vec{t'}.\covar{1} $, and thus conclude the completion of $d_1^\prime$ into a full derivation,  it is sufficient to re-apply the whole process using the derivations of Equation~\ref{eq_subproof_1} and Equation~\ref{eq_subproof_2} above. This process is well founded and eventually terminates because the modal depth is decreasing.

We now proceeds with the technical statements.

\begin{lem}
\label{lemma:m-elim-first-step}

Let $d_1$ be a CAN-free and M-free derivation of $G \sep\vdash \Gamma$ using the $\Diamond$ rule and let $\Delta$ be a sequent. Then there exists a prederivation of
$$
G \sep G \sep \vdash \Gamma,\Delta.
$$
where all non-terminated leaves are either of the form $G\sep \vdash \vec{r}.\Delta$ or of the form $G\sep \vdash \Diamond \Gamma', \vec{r}.\Delta, \vec{s}.1,\vec{t}.\covar{1}$ for some sequent $\Gamma'$ and vectors $\vec{r},\vec{s},\vec{t}$ such that
\begin{itemize}
\item $\sum\vec{s} \geq \sum\vec{t}$ and
\item $\vdash \Gamma', \vec{s}.1,\vec{t}.\covar{1}$ has a derivation $d_1'$ with a strictly lower modal depth than $d_1$.
\end{itemize}
\end{lem}
\begin{proof}
  This is an instance of the slightly more general statement of Lemma~\ref{lem:gen-m-elim-aux} below where:
  \begin{itemize}
  \item $[ \vdash \Gamma_i ]_{i=1}^{n-1} = G$ and $\Gamma_n = \Gamma$.
  \item $\vec r_i = \emptyset$ for $1 \leq i < n$ and $\vec{r}_n = \vec r$.\qedhere
  \end{itemize}
\end{proof}

\begin{lem}
\label{lemma:m-elim-second-step}
Let $d_2$ be CAN-free and M-free derivation of $G \sep \vdash \Delta$. Then, for every vector $\vec{r}$, there exists a CAN-free and M-free derivation of
$$
G \sep \vdash\vec{r}.\Delta
$$
with a modal depth lower or equal than $d_2$.
\end{lem}
\begin{proof}
This is an instance of the slightly more general statement of Lemma~\ref{lem:copyProof} below where:
  \begin{itemize}
  \item $[ \vdash \Delta_i ]_{i=1}^{n-1} = G$ and $\Delta_n = \Delta$.
  \item $\vec r_i = 1$ for $1 \leq i < n$ and $\vec{r}_n = \vec r$.\qedhere
  \end{itemize}
\end{proof}

\begin{lem}
  \label{lem:m-elim-last-step}
  Let $d_1$ be a CAN-free and M-free derivation of $G \sep\vdash \Gamma$ without any $\Diamond$ rule and let $\Delta$ be a sequent. Then there exists a prederivation of
$$
G \sep G \sep \vdash \Gamma,\Delta.
$$
where all non-terminated leaves are of the form $G\sep \vdash \vec{r}.\Delta$ for some vector $\vec{r}$.
\end{lem}
\begin{proof}
  This is an other intance of Lemma $\ref{lem:gen-m-elim-aux}$ where:
  \begin{itemize}
  \item $[ \vdash \Gamma_i ]_{i=1}^{n-1} = G$ and $\Gamma_n = \Gamma$.
  \item $\vec r_i = \emptyset$ for $1 \leq i < n$ and $\vec{r}_n = \vec r$.
  \end{itemize}

  \noindent Since the leaves of the form $G\sep \vdash \Diamond \Gamma', \vec{r}.\Delta, \vec{s}.1,\vec{t}.\covar{1}$ are generated only by the $\Diamond$ rule, and there is no $\Diamond$ rule in $d_1$, then all non-terminated leaves are of the form $G\sep \vdash \vec{r}.\Delta$ for some vector $\vec{r}$.
\end{proof}

\begin{lem}
  \label{lem:gen-m-elim-aux}
  Let $d_1$ be a CAN-free and M-free derivation of $[\vdash \Gamma_i]_{i=1}^n$ and let $G$ be a hypersequent and $\Delta$ be a sequent. Then for every sequence of vectors $\vec{r_i}$, there exists a prederivation of
  \[ G \sep [\vdash \Gamma_i, \vec{r_i}.\Delta]_{i=1}^n \]
  where all non-terminated leaves are of the form $G\sep \vdash \Diamond \Gamma', \vec{r}.\Delta, \vec{s}.1,\vec{t}.\covar{1}$ for some sequent $\Gamma'$ and vectors $\vec{r},\vec{s},\vec{t}$ such that
  \begin{itemize}
  \item $\sum\vec{s} \geq \sum\vec{t}$ and
  \item $\vdash \Gamma', \vec{s}.1,\vec{t}.\covar{1}$ has a derivation $d_1'$ with a strictly lower modal depth than $d_1$.
  \end{itemize}
\end{lem}
\begin{proof}
  We prove the result by induction on $d_1$. We will only show the $\Diamond$ and the $1$ rules, since all other cases are done in the same way as in Lemma~\ref{lem:modal_free_gen-m-elim-first-step}.
  \begin{itemize}
  \item if $d_1$ finishes with:
    \[ \infer[1 , \sum \vec{s} \geq \sum\vec{t}]{[\vdash \Gamma_i]_{i=2}^n \sep \vdash \Gamma_1, \vec{s}.1,\vec{t}.\covar{1}}{[\vdash \Gamma_i]_{i=2}^n \sep \vdash \Gamma_1} \]
    then by induction hypothesis, there is a prederivation of $G\sep[\vdash \Gamma_i,\vec{r}_i.\Delta]_{i=2}^n \sep \vdash \Gamma_1,\vec{r}_1.\Delta$ where all non-terminated leaves are of the form $G\sep \vdash \Diamond \Gamma', \vec{r}.\Delta, \vec{s}.1,\vec{t}.\covar{1}$ for some sequent $\Gamma'$ and vectors $\vec{r},\vec{s},\vec{t}$ such that
    \begin{itemize}
    \item $\sum\vec{s} \geq \sum\vec{t}$ and
    \item $\vdash \Gamma', \vec{s}.1,\vec{t}.\covar{1}$ has a derivation $d_1'$ with a strictly lower modal depth than $d_1$.
    \end{itemize}
    We continue the prederivation with
    \[ \infer[1 , \sum \vec{s} \geq \sum\vec{t}]{G \sep [\vdash \Gamma_i, \vec{r}_i.\Delta]_{i=2}^n \sep \vdash \Gamma_1,\vec{r}_1.\Delta, \vec{s}.1,\vec{t}.\covar{1}}{G \sep [\vdash \Gamma_i,\vec{r}_i.\Delta]_{i=2}^n \sep \vdash \Gamma_1,\vec{r}_1.\Delta} \]
  \item If $d_1$ finishes with:
    \[ \infer[\Diamond , \sum \vec{s} \geq \sum\vec{t}]{ \vdash \Diamond\Gamma_1, \vec{s}.1,\vec{t}.\covar{1}}{ \vdash \Gamma_1, \vec{s}.1,\vec{t}.\covar{1}} \]
    then the prederivation is simply the leaf $G \sep \vdash \Diamond\Gamma_1, \vec{r}_1.\Delta, \vec{s}.1,\vec{t}.\covar{1}$ which satisfies both
    \begin{itemize}
      \item $\sum\vec{s} \geq \sum\vec{t}$ and
      \item $\vdash \Gamma_1, \vec{s}.1,\vec{t}.\covar{1}$ is derivable using stricly less $\Diamond$ rule than in $d_1$.\qedhere
    \end{itemize}
  \end{itemize}
\end{proof}

\begin{lem}
  \label{lem:copyProof}
  If $d_2$ is a CAN-free M-free derivation of $\left [ \vdash \Delta_i \right ]_{i=1}^n$ then for all $\vec{r_i}$, there is a CAN-free M-free derivation of $\left [\vdash \vec{r_i}.\Delta_i \right ] _{i=1}^n$ with a modal depth lower or equal than $d_2$.
\end{lem}
\begin{proof}
  We will only show the $\Diamond$ and $1$ rules, the other cases being similar to Lemma~\ref{lem:modal_free_copyProof} --- and so do not introduce any new $\Diamond$ rule.
  \begin{itemize}
  \item if $d_2$ finishes with:
    \[ \infer[1 , \sum \vec{s} \geq \sum\vec{t}]{[\vdash \Delta_i]_{i=2}^n \sep \vdash \Delta_1, \vec{s}.1,\vec{t}.\covar{1}}{[\vdash \Delta_i]_{i=2}^n \sep \vdash \Delta_1} \]
    then by induction hypothesis, there is a CAN-free M-free derivation of  $[\vdash \vec{r}_i.\Delta_i]_{i=2}^n \sep \vdash \vec{r}_1.\Delta_1$ with a modal depth lower or equal  than $d_2$.
    We continue the derivation with 
    \[ \infer[1 , \sum \vec{r}_1\vec{s} \geq \sum\vec{r}_1\vec{t}]{[\vdash \vec{r}_i.\Delta_i]_{i=2}^n \sep \vdash \vec{r}_1.\Delta_1, (\vec{r}_1\vec{s}).1,(\vec{r}_1\vec{t}).\covar{1}}{[\vdash \vec{r}_i.\Delta_i]_{i=2}^n \sep \vdash \vec{r}_1.\Delta_1} \]
    which does not increase the modal depth of the derivation.
  \item If $d_2$ finishes with:
    \[ \infer[\Diamond , \sum \vec{s} \geq \sum\vec{t}]{ \vdash \Diamond\Delta_1, \vec{s}.1,\vec{t}.\covar{1}}{ \vdash \Delta_1, \vec{s}.1,\vec{t}.\covar{1}} \]
    by induction hypothesis, there is a derivation of $\vdash \vec{r}_1.\Delta_1, (\vec{r}_1\vec{s}).1,(\vec{r}_1\vec{t}).\covar{1}$ with a modal depth strictly less than $d_2$.
    We continue the derivation with
    \[ \infer[\Diamond , \sum \vec{r}_1\vec{s} \geq \sum\vec{r}_1\vec{t}]{ \vdash \vec{r}_1.\Diamond\Delta_1, (\vec{r}_1\vec{s}).1,(\vec{r_1}\vec{t}).\covar{1}}{ \vdash \vec{r}_1.\Delta_1, (\vec{r}_1\vec{s}).1,(\vec{r_1}\vec{t}).\covar{1}} \]
      which gives a derivation with a modal depth less or equal than $d_2$.\qedhere
  \end{itemize}
\end{proof}

\subsection{CAN elimination -- Proof of Theorem~\ref{thm:can_elim}}
\label{subsec:can_elim}
Recall that the CAN rule has the following form:

$$\infer[\text{CAN}, \sum \vec{r} = \sum \vec{s}]{G \sep \vdash \Gamma}{G \sep \vdash \Gamma, \vec{s}.A , \vec{r}.\negF{A}} $$

As in Section~\ref{subsec:modal_free_can_elim}, we prove Theorem~\ref{thm:can_elim} by showing that if the hypersequent $G \sep \vdash \Gamma, \vec{s}.A , \vec{r}.\negF{A}$ has a M-free CAN-free derivation, then so does the hypersequent $G \sep \vdash \Gamma$.

As explained in the discussion before Theorem~\ref{thm:can_elim} and in its proof sketch, the proof cannot just invoke the CAN-free invertibility Theorem~\ref{thm:invertibility} to simplify the logical complexity of the CAN term, due to the constraints imposed by the $\Diamond$-rule (the $1$-rule is dealt with in a similar fashion to the ID-rule).

To circumvent this issue, we prove the slightly more general Lemma~\ref{lem:gen_can_elim} by double induction on both the term $A$ and the derivation of $G \sep \vdash \Gamma, \vec{r}.A , \vec{s}.\negF{A}$.

We first prove the two basic cases where $A = x$ (or equivalently $ A = \covar{x})$ in Lemma~\ref{lem:atomic_can_elim-simpler} and $A = 1$ (or equivalently $A = \covar{1}$) in Lemma~\ref{lem:1_can_elim-simpler}, and the general case in Lemma~\ref{lem:gen_can_elim}.

\begin{lem}
  \label{lem:atomic_can_elim-simpler}
  If there is a M-free CAN-free derivation of  $G \sep \vdash \Gamma, \vec{r}.x , \vec{s}.\negF{x}$, where $\sum \vec{r} = \sum \vec{s}$, then there exists a M-free CAN-free derivation of $G \sep \vdash \Gamma$.
  \end{lem}
  \begin{proof}
    The statement follows as a special case of Lemma~\ref{lem:atomic_can_elim} below, a stronger version of Lemma~\ref{lem:atomic_can_elim-simpler} that allows for a simpler proof by induction on the structure of the derivation of $G \sep \vdash \Gamma, \vec{r}.x , \vec{s}.\negF{x}$, where:
    \begin{itemize}
    \item $[\vdash \Gamma_i]_{i=1}^{n-1} = G$ and $\Gamma_n = \Gamma$.
    \item $\vec{r_i}=\vec{r'_i}=\vec{s_i}=\vec{s'_i}=\emptyset$ for $1 \leq i < n$.
    \item $\vec{r_n}=\vec{r}$, $\vec{s_n}=\vec{s}$ and $\vec{r'_n}=\vec{s'_n}=\emptyset$.\qedhere
    \end{itemize} 
\end{proof}

\begin{lem}
  \label{lem:1_can_elim-simpler}
  If there is a M-free CAN-free derivation of  $G \sep \vdash \Gamma, \vec{r}.1 , \vec{s}.\negF{1}$, where $\sum \vec{r} = \sum \vec{s}$ then there exists a M-free CAN-free derivation of $G \sep \vdash \Gamma$.
  \end{lem}
  \begin{proof}
    The statement follows as a special case of Lemma~\ref{lem:1_can_elim} below, a stronger version of Lemma~\ref{lem:1_can_elim-simpler} that allows for a simpler proof by induction on the structure of the derivation of $G \sep \vdash \Gamma, \vec{r}.1 , \vec{s}.\negF{1}$, where:
    \begin{itemize}
    \item $[\vdash \Gamma_i]_{i=1}^{n-1} = G$ and $\Gamma_n = \Gamma$.
    \item $\vec{r_i}=\vec{r'_i}=\vec{s_i}=\vec{s'_i}=\emptyset$ for $1 \leq i < n$.
    \item $\vec{r_n}=\vec{r}$, $\vec{s_n}=\vec{s}$ and $\vec{r'_n}=\vec{s'_n}=\emptyset$.\qedhere
    \end{itemize} 
\end{proof}

We are now ready to prove the general case.

\begin{lem}
  \label{lem:gen_can_elim}
  For all terms $A$ and numbers $n > 0$ and for all sequents $\Gamma_i$ and vectors $\vec{r}_i,\vec{s}_i$ such that $\sum \vec{r}_i = \sum \vec{s}_i$,  for $1 \leq i \leq n$, 
  \begin{center}
    if $[\vdash \Gamma_i, \vec{r}_i.A,\vec{s}_i.\negF{A} ]_{i=1}^n$ has a M-free CAN-free derivation, then so does $[\vdash \Gamma_i ]_{i=1}^n$.
  \end{center}
\end{lem}
\begin{proof}
  For the basic cases $A = x$, $A = \covar{x}$, $A = 1$ and $A = \covar{1}$, we use Lemmas~\ref{lem:atomic_can_elim-simpler} and~\ref{lem:1_can_elim-simpler}.
  For complex terms $A$ which are not $\Diamond$ terms, we proceed by invoking the CAN-free invertibility Theorem~\ref{thm:invertibility} as follows:
  \begin{itemize}
  \item If $A = 0$, we can conclude with the CAN-free invertibility of the rule $0$.
  \item If $A = B + C$, since the $+$ rule is CAN-free invertible, $\left [ \vdash \Gamma_i, \vec{r_i}.B,\vec{r_i}.C, \vec{s_i}.\negF{B},\vec{s_i}.\negF{C} \right ]$ has a CAN-free, M-free derivation. Therefore we can have a CAN-free derivation of the hypersequent $[\vdash \Gamma_i ]_{i=1}^n$ by invoking the induction hypothesis twice, since the complexity of $B$ and $C$ is lower than that of $B + C$.
  \item If $A = r'B$, since the $\times$ rule is CAN-free invertible, $\left [ \vdash \Gamma_i, (r'\vec{r_i}.).B, (r'\vec{s_i}.).\negF{B} \right ]$ has a CAN-free, M-free derivation. Therefore we can have a CAN-free derivation of the hypersequent $[\vdash \Gamma_i ]_{i=1}^n$ by invoking the induction hypothesis on the simpler term $B$.
  \item If $A = B \sqcup C$, since the $\sqcup$ rule is CAN-free invertible, $$\left [ \vdash \Gamma_i, \vec{r_i}.B  , \vec{s_i}.(\negF{B} \sqcap \negF{C}) \right ] \sep \left [ \vdash \Gamma_i, \vec{r_i}.C  , \vec{s_i}.(\negF{B} \sqcap \negF{C}) \right ] $$ has a CAN-free, M-free derivation. Then since the $\sqcap$ is CAN-free invertible, $$\left [ \vdash \Gamma_i, \vec{r_i}.B  , \vec{s_i}.\negF{B} \right ] \sep \left [ \vdash \Gamma_i, \vec{r_i}.C  , \vec{s_i}.\negF{C} \right ] $$  has a CAN-free, M-free derivation. Therefore we can obtain a CAN-free derivation of the hypersequent $[\vdash \Gamma_i ]_{i=1}^n$ by invoking the induction hypothesis twice on the simpler terms $B$ and $C$.
  \item If $A = B \sqcap C$, we proceed in a similar way as for the case $A = B \sqcup C$.
  \item Finally, if $A = \Diamond B$, we distinguish two cases:
  \begin{enumerate}
  \item the derivation ends with an application of the $\Diamond$ rule which simplifies $A=\Diamond B$ to $B$. In this case we can simply conclude by invoking the induction hypothesis on $B$.
  \item The derivation ends with some other rule (recall that no CAN rules and no M rules appear in the derivation). In this case we decrease the complexity of the derivation, keeping $\Diamond B$ as the CAN term, and then invoke the induction hypothesis on the derivation having reduced complexity. This proof step is rather long to prove, as it requires analysing all possible cases. We just illustrate the two cases when the derivation ends with a logical rule $(+)$ and a structural rule (C)  to illustrate the general method.
    \begin{itemize}
    \item if the derivation finishes with
      \[ \infer[+]{[\vdash \Gamma_i, \vec{r}_i.\Diamond B,\vec{s}_i.\Diamond\negF{B} ]_{i=2}^n \sep \vdash \Gamma_1, \vec{r}_1.\Diamond B, \vec{s}_1.\Diamond\negF{B}, \vec{r'}.(C + D)}{[\vdash \Gamma_i, \vec{r}_i.\Diamond B,\vec{s}_i.\Diamond\negF{B} ]_{i=2}^n \sep \vdash \Gamma_1, \vec{r}_1.\Diamond B, \vec{s}_1.\Diamond\negF{B}, \vec{r'}.C, \vec{r'}.D} \]
      by induction hypothesis, there is a CAN-free M-free derivation of
      \[ [\vdash \Gamma_i]_{i=2}^n \sep \vdash \Gamma_1, \vec{r'}.C, \vec{r'}.D \]
      We continue the derivation with
      \[ \infer[+]{[\vdash \Gamma_i ]_{i=2}^n \sep \vdash \Gamma_1,\vec{r'}.(C + D)}{[\vdash \Gamma_i]_{i=2}^n \sep \vdash \Gamma_1, \vec{r'}.C, \vec{r'}.D} \]
    \item if the derivation finishes with
      \[ \infer[\text{C}]{[\vdash \Gamma_i, \vec{r}_i.\Diamond B,\vec{s}_i.\Diamond\negF{B} ]_{i=2}^n \sep \vdash \Gamma_1, \vec{r}_1.\Diamond B, \vec{s}_1.\Diamond\negF{B}}{[\vdash \Gamma_i, \vec{r}_i.\Diamond B,\vec{s}_i.\Diamond\negF{B} ]_{i=2}^n \sep \vdash \Gamma_1, \vec{r}_1.\Diamond B, \vec{s}_1.\Diamond\negF{B} \sep \vdash \Gamma_1, \vec{r}_1.\Diamond B, \vec{s}_1.\Diamond\negF{B}} \]
      by induction hypothesis, there is a CAN-free M-free derivation of
      \[ [\vdash \Gamma_i]_{i=2}^n \sep \vdash \Gamma_1 \sep \vdash \Gamma_1 \]
      We continue the derivation with
      \[ \infer[\text{C}]{[\vdash \Gamma_i ]_{i=2}^n \sep \vdash \Gamma_1}{[\vdash \Gamma_i]_{i=2}^n \sep \vdash \Gamma_1 \sep \vdash \Gamma_1} \qedhere\]
    \end{itemize}
  \end{enumerate} 
  \end{itemize}
\end{proof}

We now have all necessary tools to prove the CAN-elimination theorem.
\begin{proof}[Proof of Theorem~\ref{thm:can_elim}]
  We want to prove that if $G$ has a derivation, then $G$ has a CAN-free derivation. We prove this result by induction on the derivation of $G$:
  \begin{itemize}
  \item If the derivation finishes with an application of a rule that is not the CAN-rule, then by induction, the premises have CAN-free derivations and we can conclude by using the exact same rule to obtain a CAN-free derivation of $G$. For
  \item If the derivation finishes with
    $$\infer[\text{CAN}, \sum \vec{r} = \sum \vec{s}]{G \sep \vdash \Gamma}{G \sep \vdash \Gamma, \vec{s}.A , \vec{r}.\negF{A}} $$
    then by induction $G \sep \vdash \Gamma, \vec{s}.A , \vec{r}.\negF{A}$ has a CAN-free derivation. By invoking the M-elimination Theorem~\ref{thm:m_elim}, $G \sep \vdash \Gamma, \vec{s}.A , \vec{r}.\negF{A}$ has a CAN-free M-free derivation and we can conclude by using Lemma~\ref{lem:gen_can_elim}.\qedhere
  \end{itemize}
\end{proof}

Finally, we prove Lemma~\ref{lem:atomic_can_elim} and Lemma~\ref{lem:1_can_elim}, the stronger versions of Lemma~\ref{lem:atomic_can_elim-simpler} and Lemma~\ref{lem:1_can_elim-simpler}.
\begin{lem}
  \label{lem:atomic_can_elim}
  If there is a CAN-free, M-free derivation of $\left[ \vdash \Gamma_i, \vec{r_i}.x, \vec{s_i}.\covar{x} \right]_{i=1}^n$ then for all $\vec{r'_i}.$ and $\vec{s'_i}.$ such that for all $i$,$\sum \vec{r}_i - \sum \vec{s}_i = \sum \vec{r'}_i - \sum \vec{s'}_i$, there is a CAN-free, M-free derivation of $\left[ \vdash \Gamma_i, \vec{r'_i}.x, \vec{s'_i}.\covar{x} \right]_{i=1}^n$.
\end{lem}

\begin{proof}
  The proof is done by induction on the derivation and is similar to the proof of Lemma~\ref{lem:modal_free_atomic_can_elim}.
\end{proof}

\begin{lem}
  \label{lem:1_can_elim}
  If there is a CAN-free, M-free derivation of $\left[ \vdash \Gamma_i, \vec{r_i}.1, \vec{s_i}.\covar{1} \right]_{i=1}^n$ then for all $\vec{r'_i}.$ and $\vec{s'_i}.$ such that for all $i$,$\sum \vec{r}_i - \sum \vec{s}_i \leq \sum \vec{r'}_i - \sum \vec{s'}_i$, there is a CAN-free, M-free derivation of $\left[ \vdash \Gamma_i, \vec{r'_i}.1, \vec{s'_i}.\covar{1} \right]_{i=1}^n$.
\end{lem}

\begin{proof}
  By induction on derivation. We show only the non-trivial case.
  \begin{itemize}
  \item If the derivation finishes with:
    \[ \infer[1, \sum \vec{a} + \sum \vec{b} \geq \sum \vec{a'} + \vec{b'}]{ \left [ \vdash \Gamma_i,\vec{r_i}. 1 ,\vec{s_i}. \covar{1} \right ]_{i\geq2} \sep \vdash \Gamma_1,(\app{\vec{a}}{\app{\vec{b}}{\vec{c}}}).1,(\app{\vec{a'}}{\app{\vec{b'}}{\vec{c'}}}).\covar{1}}{ \left [ \vdash \Gamma_i,\vec{r_i}. 1  ,\vec{s_i}. \covar{1} \right ]_{i\geq2} \sep \vdash \Gamma_1,\vec{c} 1 , \vec{c'} \covar{1}} \]
    with $\vec{r_1} = \app{\vec{b}}{\vec{c}}$ and $\vec{s_1} = \app{\vec{b'}}{\vec{c'}}$.
    We want to show that $$\proveM \left [ \vdash \Gamma_i,\vec{r'_i}. 1  ,\vec{s'_i}. \covar{1} \right ]_{i\geq2} \sep \vdash \Gamma_1,(\app{\vec{r'_1}}{\vec{a}}).1  ,(\app{\vec{s'_1}}{\vec{a'}}). \covar{1}$$
    We will now prove that $\sum \vec{c} - \sum \vec{c'} \leq \sum \vec{r'}_{1} + \sum \vec{a} - (\sum \vec{s'}_{1} + \sum \vec{a'})$ to be able to conclude with the induction hypothesis.
    \begin{eqnarray*}
      \sum \vec{c} - \sum \vec{c'} & = & (\sum \vec r_1 - \sum\vec b) - (\sum \vec s_1 - \sum\vec b') \\
                                   & = & (\sum \vec r_1 - \sum \vec s_1) + (\sum \vec b ' - \sum \vec b) \\
                                   & \leq & (\sum \vec {r'}_1 - \sum \vec {s'}_1) + (\sum \vec a - \sum \vec a') \\
                                   & = & \sum \vec{r'}_{1} + \sum \vec{a} - (\sum \vec{s'}_{1} + \sum \vec{a'})
    \end{eqnarray*} so by induction hypothesis \[\proveM \left [ \vdash \Gamma_i,\vec{r'_i}. 1  ,\vec{s'_i}. \covar{1} \right ]_{i\geq2} \sep \vdash \Gamma_1,(\app{\vec{r'_1}}{\vec{a}}).1  ,(\app{\vec{s'_1}}{\vec{a'}}). \covar{1} \]
    which is the result we want.\qedhere
  \end{itemize}
\end{proof}


\subsection{Decidability -- Proof of Theorem~\ref{thm:decidability}}
\label{subsec:decidability}

In this section we adapt the algorithm presented in Section~\ref{subsec:modal_free_decidability} and prove the decidability of the  \hmr\ system.  

The procedure takes  a hypersequent $G$, where scalars are polynomials over scalar-variables $\vec{\alpha}$ as coefficients in weighted terms, and construct a formula $\phi_G(\vec{\alpha}) \in FO(\mathbb{R},+,\times,\leq)$ in the language of the first order theory of the reals. The procedure is recursive and terminates because each recursive call decreases the logical complexity and the modal complexity (i.e., the maximal modal depth of any terms) of its input $G$. The key property is that a sequence of scalars $\vec{s}\in\mathbb{R}$ satisfies $\phi_G$ if and only if the hypersequent  $G[s_j/\alpha_j]$ is derivable in the system \hmr. The decidability then follows from the well-known fact that the theory $FO(\mathbb{R},+,\times,\leq)$ admits quantifier elimination and is decidable~\cite{tarski1951,GRIGOREV198865}.

The algorithm to construct $\phi_G$ takes as input $G$ and proceeds as follows: 
\begin{enumerate}
\item if $G$ is not a basic hypersequent (i.e., if it contains any complex term whose outermost connective is not $\Diamond$ or $1$ or $\overline{1}$), then the algorithm returns 
$$\phi_G = \bigwedge^{n}_{i=1} \phi_{G_i}$$
where $G_1,\dots, G_n$ are the basic hypersequents obtained by iteratively applying the logical rules, and $\phi_{G_i}$ is the formula recursively computed by the algorithm on input $G_i$.
\item if $G$ has the shape $\vdash$ then $\phi_G = \top$.
\item if $G$ is a basic hypersequent which is not $\vdash$ then $G$ has the shape \[\vdash \Gamma_1, \Diamond \Delta_1,\vec{R'}_1.1,\vec{S'}_1.\covar{1} \sep ... \sep \vdash \Gamma_m, \Diamond \Delta_m,\vec{R'}_m.1,\vec{S'}_m.\covar{1}\] where $\Gamma_i =  \vec{R}_{i,1}.x_1,...,\vec{R}_{i,k}.x_k, \vec{S}_{i,1}.\covar{x_1},...,\vec{S}_{i,k}.\covar{x_{k}}$. For all $I \subsetneq [1...m]$, we define:
  \begin{itemize}
  \item A formula $Z_I(\beta_1,...,\beta_m)$ that states that for all $i \in I$, $\beta_i = 0$.
    \[ Z_I(\beta_1,...,\beta_m) = \bigwedge_{i\in I} (\beta_i = 0) \]
  \item A formula $NZ_I(\beta_1,...,\beta_m)$ that states that for all $i \notin I$, $0 < \beta_i$.
    \[ NZ_I(\beta_1,...,\beta_m) = \bigwedge_{i \notin I} (0 \leq \beta_i) \wedge \neg(\beta_i = 0) \]
  \item A formula $A_I(\beta_1,...,\beta_m)$ that states that all the atoms cancel each other.
    \[ A_I(\beta_1,...,\beta_m) =  \bigwedge_{j=0}^k (\sum_{i=1}^m \beta_i \sum \vec{R}_{i,j} = \sum_{i=1}^m \beta_i \sum \vec{S}_{i,j})\]
  \item A formula $O_I(\beta_1,...,\beta_m)$ that states that there are more $1$ than $\covar{1}$, 
    \[ O_I(\beta_1,...,\beta_m) = \sum_{i=1}^m \beta_i \sum \vec{S'}_{i} \leq \sum_{i=1}^m \beta_i \sum \vec{R'}_{i} \]
  \item A hypersequent $H_I(\beta_1,...,\beta_m)$ which is the result of cancelling the atoms using $\beta_1,...,\beta_m$ and then using the $\Diamond$ rule, i.e. is the leaf of the following prederivation:
    \[ \scalebox{0.85}{\infer[\text{T}^*]{\vdash \Gamma_{k_1}, \Diamond \Delta_{k_1},\vec{R'}_{k_1}.1,\vec{S'}_{k_1}.\covar{1} \sep ... \sep \vdash \Gamma_{k_l}, \Diamond \Delta_{k_l},\vec{R'}_{k_l}.1,\vec{S'}_{k_l}.\covar{1}}{\infer[\text{S}^*]{\vdash \beta_{k_1}.\Gamma_{k_1}, \beta_{k_1}.\Diamond \Delta_{k_1},(\beta_{k_1}\vec{R'}_{k_1}).1,(\beta_{k_1}\vec{S'}_{k_1}).\covar{1} \sep ... \sep \vdash  \beta_{k_l}.\Gamma_{k_l}, \beta_{k_l}.\Diamond \Delta_{k_l},(\beta_{k_l}\vec{R'}_{k_l}).1,(\beta_{k_l}\vec{S'}_{k_l}).\covar{1}}{\infer[\text{ID}^*]{\vdash \beta_{k_1}.\Gamma_{k_1}, \beta_{k_1}.\Diamond \Delta_{k_1},(\beta_{k_1}\vec{R'}_{k_1}).1,(\beta_{k_1}\vec{S'}_{k_1}).\covar{1} ,..., \beta_{k_l}.\Gamma_{k_l}, \beta_{k_l}.\Diamond \Delta_{k_l},(\beta_{k_l}\vec{R'}_{k_l}).1,(\beta_{k_l}\vec{S'}_{k_l}).\covar{1}}{\infer[\Diamond]{\vdash \beta_{k_1}.\Diamond \Delta_{k_1},(\beta_{k_1}\vec{R'}_{k_1}).1,(\beta_{k_1}\vec{S'}_{k_1}).\covar{1} ,..., \beta_{k_l}.\Diamond \Delta_{k_l},(\beta_{k_l}\vec{R'}_{k_l}).1,(\beta_{k_l}\vec{S'}_{k_l}).\covar{1}}{\vdash \beta_{k_1}. \Delta_{k_1},(\beta_{k_1}\vec{R'}_{k_1}).1,(\beta_{k_1}\vec{S'}_{k_1}).\covar{1},..., \beta_{k_l}. \Delta_{k_l},(\beta_{k_l}\vec{R'}_{k_l}).1,(\beta_{k_l}\vec{S'}_{k_l}).\covar{1}}}}} } \]
    where $\{k_1,...,k_l\}=[1..m]\backslash I$.    
    \item The formula $\phi_{H_I(\beta_1,...,\beta_m)}$ computed recursively from $H_I(\beta_1,...,\beta_m)$ above.
     
    \item A formula $\phi_{G,I}$ that corresponds to $\phi_{G'}$ where $G'$ is the hypersequent obtained on using the W rule on all $i$-th sequents for $i \in I$, i.e. the leaf of the following prederivation:
      \[ \infer[\text{W}^*]{\vdash \Gamma_{1}, \Diamond \Delta_{1},\vec{R'}_{1}.1,\vec{S'}_{1}.\covar{1} \sep ... \sep \vdash \Gamma_{m}, \Diamond \Delta_{m},\vec{R'}_{m}.1,\vec{S'}_{m}.\covar{1}}{\vdash \Gamma_{k_1}, \Diamond \Delta_{k_1},\vec{R'}_{k_1}.1,\vec{S'}_{k_1}.\covar{1} \sep ... \sep \vdash \Gamma_{k_l}, \Diamond \Delta_{k_l},\vec{R'}_{k_l}.1,\vec{S'}_{k_l}.\covar{1}} \] with $\{k_1,...,k_l\} =[1..m]\backslash I $. Then $\phi_{G,I}=$
    \[ \exists \beta_1,...,\beta_m, Z_I(\beta_1,...,\beta_m) \wedge NZ_I(\beta_1,...,\beta_m) \wedge A_I(\beta_1,...,\beta_m)\wedge O_I(\beta_1,...,\beta_m) \wedge \phi_{H_I(\beta_1,...,\beta_m)}  \]
  \end{itemize}
  Finally, we return $\phi_G$ defined as follows:
    \[ \phi_G = \bigvee_{I \subsetneq [1...m]} \phi_{G,I} \]
\end{enumerate}

The following theorem states the correctness of the above described algorithm.
\begin{thm}
  \label{thm:to_FO}
Let $G$ be a hypersequent having polynomials $R_1,\dots, R_k\in\mathbb{R}[\vec{\alpha}]$ over scalar-variables $\vec{\alpha}$. Let $\phi_G(\vec{\alpha})$ be the formula returned by the algorithm described above on input $G$. Then, for all $\vec{s}\in\mathbb{R}$ such that for all $i\in[1..k], R_i(\vec{s}) > 0$, the following are equivalent:
\begin{enumerate}
\item $\phi_G(\vec{s})$ holds in $\mathbb{R}$,
\item $G[s_j/\alpha_{j}]$ is derivable in \hmr.
\end{enumerate}
\end{thm}
\begin{proof}
  As in Theorem~\ref{thm:modal_free_to_FO}, by using the CAN-free invertibility Theorem~\ref{thm:invertibility}, we can assume that $G$ is a basic hypersequent. If $G$ has the shape $\vdash$, the result is trivial. Otherwise, the result is a direct corollary of Lemma~\ref{lem:lambda_prop} below since the formula $NZ_I$ corresponds to the first property, the formula $A_I$ corresponds to the second property, the formula $O_I$ corresponds to the third one and the formula $\phi_{H_I}$ corresponds to the last one.
\end{proof}

Even though the problem is decidable, the algorithm described previously is non elementary since the size of the formula $\phi_G$ can not be bound by a finite tower of exponentials.
\begin{lem}
  Let $A_n$ be defined by induction on $n$ as follows:
  \begin{itemize}
  \item $A_0 = x$ for some variable $x$
  \item $A_{n+1} = \Diamond A_n \sqcup \Diamond A_n$
  \end{itemize}
  For $i \in \mathbb{N}$, let $G_i=\ \vdash 1.A_i$. Then for all $i$, $\phi_{G_i}$ has at least $\underbrace{2^{2^{\iddots^{2}}}}_{i\text{ times}}$ existentials (with the convention $\underbrace{2^{2^{\iddots^{2}}}}_{0\text{ times}} = 1$), i.e., each use of the $\Diamond$ rule will add one exponential to the number of existentials.
\end{lem}
\begin{proof}
  For $i,j,k \in \mathbb{N}$ and $R_1,...,R_j\in\mathbb{R}[\alpha_1,...,\alpha_k]$, we define $$H_{i,j}(R_1,...,R_j) = \ \vdash R_1.A_i,...,R_j.A_i$$ We will show by induction that for all $0 < i$ and $0 < j$, $\phi_{H_{i,j}}$ has at least $\underbrace{2^{2^{\iddots^{2^j}}}}_{i\text{ times}}$ existentials.

  For all $j\in\mathbb{N}_{>0}$, $k\in\mathbb{N}$ and $R_1,...,R_j\in\mathbb{R}[\alpha_1,...,\alpha_k]$, $H_{0,j}[R_1,...,R_j]$ is an atomic hypersequent with only one sequent, so $\phi_{H_{0,j}}$ has at least one existential.

  Let $j\in\mathbb{N}_{>0}$, $k\in\mathbb{N}$ and $R_1,...,R_j\in\mathbb{R}[\alpha_1,...,\alpha_k]$. By applying iteratively the logical rules on $H_{1,j}$, we obtain the basic hypersequent $H^b = [\vdash R_1.\Diamond x,...,R_j.\Diamond x]^{2^j}$ and $\phi_{H^b}$ has $(2^{2^j}-1)(2^j + 1)$ existentials. So $\phi_{H_{1,j}}(R_1,...,R_j)$ has $(2^{2^j}-1)(2^j+1) \geq 2^j$ existentials.

  Let us now analyse $\phi_{H_{i+1,j}}(R_1,...,R_j)$ for $i > 0$. By applying iteratively the logical rules, we obtain only one basic hypersequent : $H^b_{i+1,j}(R_1,...,R_j)=\ [\vdash R_1.\Diamond A_i,...,R_j.\Diamond A_i]^{2^j}$. We notice that one of the subformulae of $\phi_{H^b_{i+1,j}}(R_1,...,R_j)$ is $\phi_{H_{i,j\times 2^j}}(\vec{R'})$ for some $\vec{R'}\in(\mathbb{R}[\alpha_1,....,\alpha_k,\beta_1,...,\beta_{2^j}])^{j\times 2^j}$, which has at least $\underbrace{2^{2^{\iddots^{2^{j\times 2^j}}}}}_{i\text{ times}} \geq \underbrace{2^{2^{\iddots^{2^j}}}}_{i+1\text{ times}}$ existentials.

  So $\phi_{H_{i+1,j}}(R_1,...,R_j)$ has at least $\underbrace{2^{2^{\iddots^{2^j}}}}_{i+1\text{ times}}$ existentials. Since $G_i = H_{i,1}(1)$, $G_i$ has at least $\underbrace{2^{2^{\iddots^{2}}}}_{i\text{ times}}$ existentials.
\end{proof}

The following result is similar to Lemma~\ref{lem:int_lambda_prop}, the only difference is that since the T rule can multiply a sequent by any strictly positive real number, the coefficients in the property are real numbers instead of natural numbers.
\begin{lem}
  \label{lem:lambda_prop}
  For all basic hypersequents $G$, built using the variables and negated variables $x_1, \covar{x_1}, \dots, x_{k}, \covar{x_{k}}$, of the form  
  \[\vdash \Gamma_1, \Diamond \Delta_1,\vec{r'}_1.1,\vec{s'}_1.\covar{1} \sep ... \sep \vdash \Gamma_m, \Diamond \Delta_m,\vec{r'}_m.1,\vec{s'}_m.\covar{1}\]
where $\Gamma_i =  \vec{r}_{i,1}.x_1,...,\vec{r}_{i,k}.x_k, \vec{s}_{i,1}.\covar{x_1},...,\vec{s}_{i,k}.\covar{x_k}$, the following are equivalent:
  
  \begin{enumerate}
  \item $G$ has a derivation.
  \item there exist numbers $t_1,...,t_m \in \mathbb{R}_{\geq 0}$, one for each sequent  in $G$, such that:
  \begin{itemize}
  \item there exists $i\in [1..m]$ such that $t_i \neq 0$, i.e., the numbers are not all $0$'s, and
  \item for every variable and covariable $(x_j, \covar{x_j})$ pair, it holds that
  $$
  \sum^{m}_{i=1} t_i (\sum \vec{r}_{i,j}) =   \sum^{m}_{i=1} t_i (\sum \vec{s}_{i,j}) 
  $$
i.e., the scaled (by the numbers $t_1$ \dots $t_m$) sum of the coefficients in front of the variable $x_j$ is equal to the scaled sum of the coefficients in from of the covariable $\covar{x_j}$. 
  \item $\sum_{i=1}^n t_i \sum \vec{s'}_i\leq \sum_{i=1}^n t_i \sum \vec{r'}_i$, i.e, there are more $1$ than $\covar{1}$ and,
  \item the hypersequent consisting of only one sequent
    $$\vdash t_1.\Delta_1,...,t_m.\Delta_m,(t_1\vec{r'}_1).1,...,(t_m\vec{r'}_m).1, (t_1\vec{s'}_1).\covar{1},...,(t_m\vec{s'}_m).\covar{1}$$ has a derivation, where the notation $0.\Gamma$ means $\emptyset$.
  \end{itemize}
  \end{enumerate}
\end{lem}
\begin{proof}
 Similar to Lemma~\ref{lem:modal_free_lambda_prop}.
\end{proof}


\subsection{An open problem on \hmr}
\label{open:problem}

We have not been able to prove or disprove the equivalent of Theorem~\ref{thm:modal_free_conservativity} in the context of the system \hmr. We leave this as an open problem.

\paragraph{Question} Let $G$ be a hypersequent whose scalars are all  rational numbers. Is it true that, if $G$ has a  CAN-free derivation in \hmr\ then $G$ also has a CAN-free and T-free derivation in \hmr?

This is a question of practical importance. Indeed the T rule is the only non-analytical rule (beside the CAN rule which, however, can be eliminated) of the system \hmr\ and, as a consequence, it makes the proof search endeavour more difficult.

\section{Conclusions}\label{conclusion:sec}
We have introduced structural proof systems, based on the machinery of hypersequent calculus, for the theories of Riesz spaces and modal Riesz spaces. Remarkably, the system \hmr{} has allowed us to obtain new results regarding modal Riesz spaces such as the decidability of their equational theory. A main direction for further research is to investigate additional applications of \hr{} and \hmr{}. One example is the recent work \cite{LM21} where the authors used results regarding the system \hmr{} to prove a novel result (previously left open in \cite[\S 6.3]{FMM2020}): free modal Riesz spaces are Archimedean.

\noindent 
\paragraph{\textbf{Acknowledgments.}} The authors are very grateful to the anonymous reviewers for their useful comments and suggestions. 

\bibliography{biblio} 
\bibliographystyle{alphaurl}

\newpage
\appendix
\section{The Hypersequent Calculus \gasystem}\label{ga:section}

For reference, the following table lists the rules of the hypersequent calculus GA of \cite{MOG2005,MOGbook} for the theory of lattice-ordered abelian groups, which provides the basis of the  systems \hr\ and \hmr\ introduced in this paper.

\begin{figure}[h!]
  \begin{center}
  \scalebox{0.85}{
      \begin{minipage}{12.5cm}
        \textbf{Axioms:}
        \begin{center}
          \begin{tabular}{cc}
            $\infer[\Delta\textnormal{--ax}]
            {\vdash}
            {}$ & $ \ \ \ \ \ \ \ \ \ \ $
                  
                  $\infer[\textnormal{ID--ax}]
                  {A\vdash A}
                  {}$\\
          \end{tabular}
        \end{center}
        \textbf{Structural rules:}
        \begin{center}
          \begin{tabular}{ccc}
            $\infer[$Weakening (W)$]
            {G | \Gamma \vdash \Delta}
            {G}$ & &
                                          $\infer[$Contraction (C)$]
                                          {G | \Gamma \vdash \Delta}
                                          {G | \Gamma \vdash \Delta | \Gamma \vdash \Delta}$\\[0.2cm]
            $\infer[$Split (S)$]
            {G | \Gamma_1 \vdash \Delta_1 | \Gamma_2 \vdash \Delta_2}
            {G | \Gamma_1,\Gamma_2 \vdash \Delta_1, \Delta_2}$ & \hspace{0.4cm} &
                                                                 $\infer[$Mix (M)$]
                                                                 {G | \Gamma_1 , \Gamma_2 \vdash \Delta_1 , \Delta_2}
                                                                 {G | \Gamma_1 \vdash \Delta_1 & G | \Gamma_2 \vdash \Delta_2}$
          \end{tabular}
        \end{center}
        \textbf{Logical rules:}
        \begin{center}
          \begin{tabular}{cc}
            $\infer[0_L]
            {G | \Gamma , 0 \vdash \Delta}
            {G | \Gamma \vdash \Delta}$ &
                                          $\infer[0_R]
                                          {G | \Gamma \vdash \Delta , 0}
                                          {G | \Gamma \vdash \Delta}$\\[0.2cm]
            $\infer[+_L]
            {G | \Gamma , A + B \vdash \Delta}
            {G | \Gamma , A , B \vdash \Delta}$ &
                                                  $\infer[+_R]
                                                  {G | \Gamma \vdash \Delta, A + B}
                                                  {G | \Gamma \vdash \Delta , A , B}$\\[0.2cm]
            $\infer[-_L]
            {G | \Gamma, -A \vdash \Delta}
            {G | \Gamma \vdash \Delta , A}$ &
                                              $\infer[-_R]
                                              {G | \Gamma \vdash \Delta, - A}
                                              {G | \Gamma , A \vdash \Delta}$\\[0.2cm]
            $\infer[\sqcup_L]
            {G | \Gamma , A \sqcup B \vdash \Delta}
            {G | \Gamma , A \vdash \Delta & G | \Gamma , B \vdash \Delta}$ &
                                                                             $\infer[\sqcup_R]
                                                                             {G | \Gamma \vdash \Delta , A \sqcup B}
                                                                             {G | \Gamma \vdash \Delta , A | \Gamma \vdash \Delta , B}$\\[0.2cm]
            $\infer[\sqcap_L]
            {G | \Gamma , A \sqcap B \vdash \Delta}
            {G | \Gamma , A \vdash \Delta | \Gamma , B \vdash \Delta}$ &
                                                                         $\infer[\sqcap_R]
                                                                         {G | \Gamma \vdash \Delta , A \sqcap B}
                                                                         {G | \Gamma \vdash \Delta , A & G | \Gamma \vdash \Delta , B}$
          \end{tabular}
                  \end{center}
                  
                   \textbf{CUT rule:}
        \begin{center}
           $ \infer[Cut]
            {G \mid \Gamma_1,\Gamma_2 \vdash \Delta_1, \Delta_2}
            {G \mid \Gamma_1 \vdash \Delta_1, A & G \mid \Gamma_2 , A \vdash \Delta_2}$ 
        \end{center}

            \end{minipage}
    }
  \end{center}
  \caption{Inference rules of the hypersequent system GA of  \cite{MOG2005}.}
  \label{rules:GA}
\end{figure}


\end{document}